\documentclass[11pt]{article}
 \usepackage[letterpaper, left=0.6in, right=0.6in, top=1in, bottom=1in]{geometry}
\usepackage{graphicx}
\usepackage{physics}
\usepackage{amsmath}
\usepackage{amssymb}
\usepackage{amsthm}
\usepackage{enumitem}
\usepackage{framed}
\usepackage{caption}
\usepackage{subfig}
\usepackage{adjustbox}
\usepackage{multirow}
\usepackage{multicol}
\usepackage{listings}
\usepackage{tikz}
\usepackage{xcolor}
\usepackage{bm}
\usepackage{mathtools}
\usepackage{algorithm}
\usepackage{algpseudocode,eqparbox}
\usepackage[T1]{fontenc}
\usepackage{wrapfig}
\usetikzlibrary{arrows,decorations.pathmorphing,decorations.footprints,decorations.pathreplacing,fadings,calc,trees,mindmap,shadows,decorations.text,patterns,positioning,shapes,matrix,fit}
\usetikzlibrary{shapes.misc}

\mathtoolsset{showonlyrefs}
\allowdisplaybreaks
\usepackage[colorlinks,citecolor=red,urlcolor=blue,bookmarks=false,hypertexnames=true]{hyperref}

\usepackage{letltxmacro}
\LetLtxMacro{\originaleqref}{\eqref}
\renewcommand{\eqref}{Eq.~\originaleqref}

\definecolor{codegreen}{rgb}{0,0.6,0}
\definecolor{codegray}{rgb}{0.5,0.5,0.5}
\definecolor{codepurple}{rgb}{0.58,0,0.82}
\definecolor{backcolour}{rgb}{0.95,0.95,0.92}

\usepackage{color}
\definecolor{bitcolor}{rgb}{0.99,0.93,0.0}
\definecolor{checkcolor}{rgb}{0.52941,0.80784,1}
\definecolor{channelcolor}{rgb}{0.67,0.88,0.69}
\definecolor{infocolor}{rgb}{0.82,0.62,0.91}
\definecolor{cqcolor}{rgb}{0.99,0.56,0.67}
\lstdefinestyle{mystyle}{
    backgroundcolor=\color{backcolour},   
    commentstyle=\color{codegreen},
    keywordstyle=\color{magenta},
    numberstyle=\tiny\color{codegray},
    stringstyle=\color{codepurple},
    basicstyle=\ttfamily\footnotesize,
    breakatwhitespace=false,         
    breaklines=true,                 
    captionpos=b,                    
    keepspaces=true,                 
    numbers=left,                    
    numbersep=5pt,                  
    showspaces=false,                
    showstringspaces=false,
    showtabs=false,                  
    tabsize=2
}

\usepackage{stmaryrd}

\usepackage[style=ieee,backend=biber,url=false,doi=false,eprint=false]{biblatex}

\bibliography{ref}

\usepackage{ifthen}

\newif\ifarxiv
\arxivtrue

\newif\ifisit
\isitfalse

\newif\ifextra
\extrafalse

\usepackage{mathtools}

\usetikzlibrary{shapes.geometric, arrows}
\tikzstyle{startstop} = [rectangle,  minimum width=3cm, minimum height=1cm,text centered,text width=10cm, draw=black ,fill=gray!20]
\tikzstyle{process} = [rectangle, minimum width=3cm, minimum height=1cm, text centered,text width=10cm, draw=black,fill=orange!20]
\tikzstyle{arrow} = [thick,->,>=stealth]
\tikzstyle{state}=[shape=circle,draw=blue!50,fill=blue!20]
\tikzstyle{observation}=[shape=rectangle,draw=orange!50,fill=orange!20]
\tikzstyle{lightedge}=[<-,dotted]
\tikzstyle{mainstate}=[state,thick]
\tikzstyle{mainedge}=[<-,thick]
\usepackage{color}
\definecolor{bitcolor}{rgb}{1,0.84314,0}
\definecolor{checkcolor}{rgb}{0.52941,0.80784,1}
\usepackage{pgfplots}
\pgfplotsset{compat=1.18}
\renewcommand{\epsilon}{\varepsilon}

\newtheorem{theorem}{Theorem}

\newtheorem{lem}[theorem]{Lemma}
\newtheorem{defn}[theorem]{Definition}

\newtheorem{example}[theorem]{Example}
\newcommand{\cF}{\mathcal{F}}
\newcommand{\cG}{\mathcal{G}}
\newcommand{\cH}{\mathcal{H}}

\newcommand{\cM}{\mathcal{M}}

\newcommand{\cT}{\mathcal{T}}
\newcommand{\cU}{\mathcal{U}}

\newcommand{\cX}{\mathcal{X}}
\newcommand{\cY}{\mathcal{Y}}

\newcommand{\bg}{\bm{g}}

\newcommand{\bw}{\bm{w}}

\newcommand{\by}{\bm{y}}

\newcommand{\mI}{\mathbb{I}}

\newcommand{\sym}[1]{S_N}

\newcommand{\F}[1]{\widehat{#1}}

\newcommand{\vnop}{\varoast}

\newcommand{\cnop}{\boxast}

\newcommand{\blambda}{\bm{\lambda}}
\newcommand{\psuc}{P_{\text{suc}}}
\newcommand{\perr}{P_{\text{err}}}

\newcommand{\psia}{\psi^{(1)}}
\newcommand{\psib}{\psi^{(2)}}
\newcommand{\lambdaa}{\lambda^{(1)}}
\newcommand{\lambdab}{\lambda^{(2)}}

\newcommand{\bmu}{\bm{\mu}}

\newcommand{\mone}{\bm{1}}
\newcommand{\mzero}{\bm{0}}

\newcommand{\hcG}{\widehat{\cG}}

\newcommand{\hcH}{\widehat{\cH}}

\newcommand{\hcU}{\widehat{\cU}}

\newcommand{\hcGm}{\widehat{\cG^m}}
\newcommand{\hcGma}{\widehat{\cG^{m+1}}}
\newcommand{\chitriv}{\chi_{\text{triv}}}
\newcommand{\image}{\text{Im}}

\usepackage[title]{appendix}
\usepackage{authblk}

\title{Quantum Message Passing for Factor Graphs over Finite Abelian Groups}

\author[1,3]{Avijit Mandal\footnote{avijit.mandal@duke.edu}}
\author[1,2,3]{Henry D. Pfister}
\affil[1]{Department of Electrical and Computer Engineering, Duke University}
\affil[2]{Department of Mathematics, Duke University}
\affil[3]{Duke Quantum Center, Duke University}

\date{}

\begin{document}
\onecolumn
\maketitle

% \begin{abstract}
%     Factor graphs provide the graphical representation for the conditional independence structure. Consider a set of classical random variables from a group and a group covariant pure state channel (PSC) that maps each element in the group $\cG$ to a quantum state, the quantum message passing builds the framework for obtaining sufficient statistics for the quantum states and propagating coherently along factor graphs. This work focuses on abelian groups and provides a rigorous analysis for quantum message passing based on group characters. This generalizes the belief propagation with quantum messages (BPQM) framework to abelian group covariant PSCs. 
% \end{abstract}

\begin{abstract}
% We develop a quantum message-passing framework for factor graphs over finite abelian groups. The channel model is a group-covariant pure-state channel (PSC), where the input alphabet is a finite abelian group $\cG$ and the output states transform covariantly under the group action. For such channels, we show that the Gram matrix is diagonalized by the character basis of the dual group $\hcG$, so the channel is characterized, up to isometric equivalence, by its character-indexed eigen list.

% Based on this representation, we analyze the induced classical-quantum channels associated with check, equality, homomorphism, marginalization, and automorphism factors. For each factor, we derive explicit update rules showing that if the incoming messages are group-covariant PSCs, or finite heralded mixtures of such PSCs, then the outgoing message remains in the same class. This gives a closed quantum message-passing calculus for tree-structured factor graphs assembled from these primitives.

% The framework applies directly to several standard code families over finite abelian groups, including polar codes, LDPC, group codes, and convolutional and turbo codes. It recovers the previously studied \(q\)-ary formulation as the special case \(\cG=\mathbb{Z}_q\), while extending BPQM to non-cyclic alphabets and more general factor-graph constraints described by homomorphisms between products of abelian groups.

% where the states are related by the group action. This setting is naturally described by a
We develop a quantum message-passing framework for factor graphs over finite abelian groups. Our starting point is the task of discriminating between a collection of quantum states indexed by the elements of a finite abelian group $\cG$  whose overlaps respect the structure of a group-covariant pure-state channel (PSC). For such channels, we show that the Gram matrix constructed from the output states is diagonalized by the character basis of the dual group $\hcG$. Hence, the channel is characterized, up to isometric equivalence, by its character-indexed eigen list.

Based on this representation, we analyze the induced classical-quantum channels associated with check, equality, homomorphism, marginalization, and automorphism factors. For each factor, we derive explicit update rules showing that if the incoming messages are heralded mixtures of group-covariant PSCs, then the outgoing message remains in the same class. This provides a closed quantum message-passing framework for tree-structured factor graphs assembled from these primitives.

The framework applies directly to several standard code families over finite abelian groups, including polar codes, LDPC codes, and convolutional and turbo codes. It recovers the previously studied \(q\)-ary formulation as the special case \(\cG=\mathbb{Z}_q\), while extending the belief propagation with quantum messages (BPQM) framework introduced by Renes to non-cyclic alphabets and more general factor-graph constraints described by homomorphisms between products of abelian groups.

\end{abstract}

\section{Introduction}

Factor graphs provide a common language for inference in coding theory, machine learning, statistical physics, and constraint satisfaction. They isolate the local constraints of a model and thereby expose algorithmic structure that can be exploited for efficient message-passing algorithms. In coding theory, this viewpoint is particularly powerful because sparse-graph-based codes (LDPC), trellis-based codes (convolutional and turbo codes), and polar constructions can all be described in terms of local factor operations such as equality constraints, parity constraints, and more general homomorphisms \cite{arikan2009channel, gallager1962low, berrou1993near, mceliece1998turbo, mceliece2002bcjr, bahl2003optimal, forney2002codes, forney2014codes}. For classical channels, belief propagation and related message-passing methods then yield exact marginalization on trees and effective approximate decoding rules on graphs with cycles.

For classical-quantum (CQ) channels, the corresponding decoding problem is substantially more difficult. Holevo, Schumacher, and Westmoreland \cite{holevo1998capacity, schumacher1997sending} showed that reliable communication up to capacity is achievable in principle by collective measurements on long output blocks, but such measurements are generally difficult to describe and implement in a structured way. This motivates the search for quantum decoders that retain some of the architectural advantages of classical message passing while operating coherently on quantum outputs \cite{wilde2012polar,wilde2013towards,mandal2025reed}. Belief propagation with quantum messages (BPQM) is one such approach in which, rather than measuring each channel output separately, it combines quantum sufficient statistics recursively along the factor graph, mirroring the structure of a classical belief-propagation decoder \cite{Renes-njp17,rengaswamy2021belief, brandsen2022belief, mandal2023belief,mandal2024polar,mandal2026belief, piveteau2025efficient}.

% Prior BPQM constructions and analyses have focused mainly on binary settings, and more recently on symmetric $q$-ary pure-state channels \cite{mandal2026belief} whose Gram matrices are circulant. In the $q$-ary case, the key observation is that the effect of local combining rules can be tracked entirely through the eigenvalues of the output Gram matrix, independent of the particular physical realization of the channel states. While this already yields useful recursions for check-node and bit-node combining, many coding and inference problems are more naturally expressed not in terms of a specific cyclic alphabet, but in terms of finite abelian groups, homomorphisms between such groups, and automorphisms acting on them. From the factor-graph point of view, this is the natural level of generality.

% In this work, we develop a quantum message-passing framework for factor graphs over finite abelian groups. We consider $\cG$-covariant pure-state channels (PSCs), where the input alphabet is a finite abelian group $\cG$ and the output states transform covariantly with respect to the group action. Our starting point is that, for such channels, the Gram matrix is diagonalized by the character basis of the dual group $\hcG$. This implies that the channel is characterized, up to isometric equivalence, by its eigen list indexed by characters of $\hcG$. Thus, the natural sufficient statistic for quantum message passing is no longer tied to a particular alphabet representation, but instead to the Fourier analysis of the underlying abelian group.

Prior BPQM constructions and analyses have focused mainly on binary settings \cite{Renes-njp17,brandsen2022belief,mandal2023belief,mandal2024polar}. In \cite{mandal2026belief}, BPQM was extended to symmetric $q$-ary pure-state channels whose Gram matrices are circulant, and it was shown that the effect of the local combining rules can be tracked via the eigen lists of the output Gram matrices, without regard for the particular physical realization of the channel states. 
% In \cite{piveteau2025efficient}, a closely related binary BPQM representation was formulated in terms of probability distributions associated with the residual channels, in the setting of convolutional and turbo decoding over qubit PSCs. 
In \cite{piveteau2025efficient}, a closely related binary BPQM representation was formulated in terms of distribution ensembles associated with BPQM messages, with applications to convolutional and turbo decoding over qubit PSCs.
These developments suggest that the relevant sufficient statistics for quantum message passing can be described  using representations adapted to the symmetry of the underlying alphabet.

In this work, we develop a quantum message-passing framework for factor graphs over finite abelian groups. Our approach is based on $\cG$-covariant pure-state channels (PSCs), where the input alphabet is a finite abelian group $\cG$ and the output states transform covariantly with respect to the group action. Our starting point is that, for such channels, the Gram matrix is diagonalized by the character basis of the dual group $\hcG$. This implies that the channel is characterized, up to isometric equivalence, by its eigen list indexed by characters of $\hcG$. In this sense, the eigen list representation from \cite{mandal2026belief} extends naturally from the cyclic case to the abelian-group setting considered here.

Using this approach, we analyze induced CQ channels associated with basic local factors that arise in abelian-group factor graphs. These are check factors, equality factors, homomorphism factors, marginalization factors, and automorphism factors. Our main goal is to show that these local operations preserve a tractable class of messages. 
% More specifically, we show that when the incoming messages are group-covariant PSCs, or, more generally, finite heralded mixtures of such PSCs, the outgoing message produced by each local factor remains in the same class and can be described explicitly through transformed eigen lists and classical side information. 
Specifically, we show that each local factor preserves this class of messages. If the incoming messages are heralded mixtures of group-covariant PSCs, then the outgoing message remains in the same class and admits an explicit description through transformed eigen lists and classical side information
This gives a closed quantum message-passing calculus for tree-structured factor graphs assembled from these primitives.

The resulting framework unifies several coding constructions within a single formalism. Polar transforms over finite abelian groups can be decomposed into automorphism, check, and equality factors, so their synthetic channels can be tracked recursively using the same eigen list machinery. Tanner graphs for LDPC and more general group codes are built from equality and parity-type homomorphism constraints, and therefore admit local quantum message-passing update rules directly from our factor analysis. Trellis-based constructions such as convolutional and turbo codes also fit naturally into this framework through finite-state-machine descriptions whose branch operations are homomorphisms followed by marginalization.
For example, our approach generalizes some results in~\cite{piveteau2025efficient} for binary convolutional codes on qubit PSCs and extends the previously studied $q$-ary setting~\cite{mandal2026belief} to a broader group-theoretic decoding framework that includes non-cyclic alphabets and more general local constraints.

The main contributions of this paper are as follows. First, we characterize finite abelian group-covariant PSCs through the eigen decomposition of their Gram matrices in the character basis. Second, we derive explicit quantum message-passing update rules for a set of specific local factor types associated with abelian group factor graphs. Third, we show that these rules are closed under composition on trees through the class of finite heralded mixtures of group-covariant PSCs. Finally, we demonstrate how this machinery specializes to standard code families, including polar codes, LDPC, group codes, convolutional codes, and turbo codes.
% For polar and LDPC/group-code constructions, our emphasis is on showing that the existing eigen-list based quantum message-passing description extends once the local constraints are written in terms of abelian-group homomorphism and automorphism factors; the more detailed development is reserved for convolutional and turbo decoding, where the finite-state-machine recursion with marginalization is developed explicitly in the present framework.
Taken together, these results provide a rigorous abelian-group extension of BPQM and identify the dual-group character domain as the natural representation for structured quantum inference on such factor graphs. 

The remainder of the paper is organized as follows. Section~\ref{sec:background} reviews the necessary background on finite abelian groups, characters, and group-covariant PSCs. Section~\ref{sec:quantum message passing factor graph lemmas} develops the local quantum message-passing rules for the factor types considered in the paper and establishes their closure properties. Section~\ref{sec:application to decoding} explains how these rules apply to several standard decoding architectures over finite abelian groups. Section~\ref{sec:DE} outlines how this quantum message-passing framework is used for the density-evolution analysis of tree factor-graph based codes over finite abelian groups.

\section{Background}\label{sec:background}
\subsection{Notation}
The natural numbers are denoted by $\mathbb{N}=\{1,2,\ldots\}$ and
$\mathbb{N}_0=\mathbb{N}\cup\{0\}$. For $m\in\mathbb{N}_0$, we use the shorthand
\[
[m]\coloneqq \{0,\ldots,m-1\}.
\]
We use boldface notation to denote vectors, e.g., $\by$ and $\bw$. The symbols $\mzero$ and $\mone$ denote the all-zero and all-one vector or matrix, respectively, whenever the meaning is clear from the context. The identity matrix is denoted by $\mI$. For a matrix $M$, the entry in row $i$ and column $j$ is denoted by $M_{i,j}$, while $M_{i,:}$ denotes the $i$th row of $M$.
 We use calligraphic letters such as $\cG,\cH,\cU$ to denote finite groups and $\hcG,\hcH,\hcU$ to denote their dual groups. For a group $\cG$, let $\{\ket{g}\}_{g\in\cG}$ denote the canonical orthonormal basis of a $|\cG|$-dimensional Hilbert space.
\subsection{Finite Abelian Groups and Homomorphisms}

Throughout the paper, all groups are finite and abelian. For group elements
$g_1,g_2\in\cG$, the group operation is denoted by $g_1g_2$. 
% \AM{not needed} For the sake of
% simplicity, we do not explicitly write a separate symbol for the group operation.
The identity element of $\cG$ is denoted by $e$, and for $g\in\cG$, the inverse element is denoted by $g^{-1}$. If $\cG_1,\ldots,\cG_d$ are finite abelian groups, then $ \cG_1\times\cdots\times\cG_d$ denotes their direct product, equipped with componentwise group operation. A map $\phi \colon \cG\to\cH$ is a homomorphism if
\[
\phi(g_1g_2)=\phi(g_1)\phi(g_2), \qquad \forall g_1,g_2\in\cG.
\]
Its kernel and image are defined by
\[
\ker \phi \coloneqq \{g\in\cG:\phi(g)=e\},
\qquad
\image \phi \coloneqq \{\phi(g):g\in\cG\}.
\]
An automorphism of $\cG$ is a bijective homomorphism from $\cG$ to itself. If $\phi:\cG\to\cH$ is surjective, then for every $h\in\cH$, the set
\[
\{g\in\cG:\phi(g)=h\}
\]
has cardinality $|\ker\phi|$. In particular,
\[
|\cG| = |\ker\phi|\,|\cH|.
\]
This identity will be used in the normalization of the induced channel associated with a homomorphism factor.
% \subsection{Group Characters}

% For a finite abelian group $\cG$, the dual group $\hcG$ consists of all characters,
% that is, all homomorphisms $\chi:\cG\to\mathbb{C}^{\times}$. The trivial character is
% denoted by $\chitriv$ and satisfies $\chitriv(g)=1$ for all $g\in\cG$.

% For two characters $\chi_1,\chi_2\in\hcG$, define their product by
% \[
% (\chi_1\chi_2)(g)\coloneqq \chi_1(g)\chi_2(g), \qquad g\in\cG.
% \]
% It is easy to check that $\chi_1\chi_2\in\hcG$. For a character $\chi\in\hcG$, its
% inverse character $\chi^{-1}\in\hcG$ is defined by
% \[
% \chi\chi^{-1}=\chitriv.
% \]
% Equivalently,
% \[
% \chi^{-1}(g)=\chi(g^{-1})=\frac{1}{\chi(g)}=\overline{\chi(g)}, \qquad g\in\cG.
% \]
\subsection{Group Characters}
For an abelian group $\cG$, define the dual group $\hcG$ to be the set of characters $\chi$, which are group homomorphisms
$\chi :\cG \to \mathbb{T}$ where $\mathbb{T}=\{z\in \mathbb{C}\colon |z|=1\}$ is the unit circle under multiplication~\cite{luong2009fourier}. It follows that 
\begin{align*}
    \chi(g_{1}g_{2})=\chi(g_1)\chi(g_{2}), \forall g_{1},g_{2}\in \cG,\forall\chi\in\hcG.
\end{align*}
Since $\cG$ is finite, each character $\chi\in\hcG$ satisfies $|\chi(g)|=1$ for all  $g\in \cG$. The group $\hcG$ is equipped with trivial character $\chitriv$ such that $\chitriv(g)=1$ for all $g\in \cG$.
For two characters $\chi_1, \chi_2$, we define $\chi_1\chi_2$ as the character such that
\begin{align*}
   (\chi_{1}\chi_{2})(g) := \chi_{1}(g)\chi_{2}(g).
\end{align*}
One can readily check that $\chi_1\chi_2\in \hcG$. For character $\chi$, let $\chi^{-1}$ denote its inverse which satisfies
\begin{align*}
    \chi\chi^{-1}=\chitriv.
\end{align*}
For a character $\chi\in\hcG$, the inverse character $\chi^{-1}\in \hcG$ is unique and satisfies the following relations $\forall g\in \cG$
\begin{align*}
    \chi^{-1}(g)=\chi(g^{-1})=\frac{1}{\chi(g)}=\overline{\chi(g)}.
\end{align*}
For characters $\chi_{1},\chi_{2}\in \hcG$, the following orthogonality relation holds
    \begin{align*}
        \sum_{g\in\cG}\chi_{1}^{-1}(g)\chi_{2}(g)= \sum_{g\in\cG}\overline{\chi_{1}(g)}\chi_{2}(g)=|\cG|\delta_{\chi_{1},\chi_{2}}.
    \end{align*}
For character $\chi\in \hcG$, define the character state as 
\begin{align*}
    \ket{\chi}\coloneqq\frac{1}{\sqrt{|\cG|}}\sum_{g\in\cG}\chi^{-1}(g)\ket{g}.
\end{align*}
\begin{lem}
    The states $\{\ket{\chi}\}_{\chi\in \hcG}$ form an orthonormal basis
\end{lem}
\begin{proof}
For $\chi_1,\chi_2\in \hcG$, we have 
\begin{align*}
    \braket{\chi_1}{\chi_2} & =\frac{1}{|\cG|}\sum_{g,g'\in\cG }\chi_{1}(g)\chi_{2}^{-1}(g')\braket{g}{g'}\\
    & = \frac{1}{|\cG|}\sum_{g\in \cG}\chi_{1}(g)\chi_{2}^{-1}(g)\\
    & = \delta_{\chi_1,\chi_2}.
\end{align*}
    Since $|\hcG|=|\cG|$, this orthonormal family has the full dimension of the Hilbert space and hence is an orthonormal basis.
\end{proof}
If $\phi\colon \cG\to\cH$ is a homomorphism, then it induces a dual map $\hat{\phi}:\hcH\to\hcG$ defined by $\hat{\phi}(\chi)\coloneqq \chi\circ\phi$. When $\phi$ is surjective, the map $\hat{\phi}$ is injective.
\subsection{Quantum Preliminaries}
\begin{defn}
    A classical-quantum (CQ) channel $W\colon x\rightarrow \rho_{x}$ takes a classical input $x$ from a finite alphabet $\cX$ and outputs a density matrix $\rho_{x}$. 
\end{defn}

\begin{defn}
    Two CQ channels $W$ and $W'$ with the same input alphabet are isometrically (or unitarily) equivalent if there exists an isometry (or unitary) $V$, such that $W'(x)=VW(x)V^{\dagger}$  for all $x\in \cX$.
\end{defn}

\begin{defn}
    A CQ channel $W\colon x\rightarrow \rho_{x}$ is called a pure state channel (PSC), if the output state $\rho_{x}=\ketbra{\psi_x}{\psi_x}$ is a pure state, i.e. a rank-1 matrix  for all $x\in \cX$.
\end{defn}
\begin{defn}
Consider a finite group $\cG$ and a unitary representation $\{U_{g}\}_{g\in \cG}$ of the group.
    A CQ channel $W\colon g\rightarrow \rho_g$, with input alphabet $\cG$, is called a $\cG$-covariant CQ channel if the output states satisfy the following relation
    \begin{align}\label{eq:CQ group covariance}
        \rho_{g'g}=U_{g'}\rho_{g}U_{g'}^{\dagger},\quad \forall g,g'\in \cG.
    \end{align}
\end{defn}
% If there exists a group which satisfies the relation in \eqref{eq:CQ group covariance} for some group, we will refer to $W$ as a group-covariant CQ channel.
More generally, if there exists some finite group $\cG$ and the unitary representation $\{U_g\}_{g\in \cG}$ for which the covariance relation in \eqref{eq:CQ group covariance} holds, then we refer to $W$ simply as a group-covariant CQ channel.

\begin{defn}
Let $\cG$ be a finite group and let $\{U_g\}_{g\in\cG}$ be a unitary representation of $\cG$.
A pure-state channel $W:g\mapsto \ketbra{\psi_g}{\psi_g}$ is called $\cG$-covariant if the following relation holds
\begin{align*}
    \ket{\psi_{g'g}}=U_{g'}\ket{\psi_g},\qquad \forall g,g'\in \cG.
\end{align*}

\end{defn}
\begin{defn}
    A matrix $M\in \mathbb{C}^{|\cG|\times |\cG|}$ indexed by group elements $g\in \cG$, is called $\cG$-circulant, if the elements of $M$ satisfy the following relation 
    \begin{align*}
        M_{gg_1,gg_2}=M_{g_1,g_2}, \qquad \forall g,g_1,g_2\in \cG.
    \end{align*}
\end{defn}
\begin{defn}
    For a $\cG$-covariant PSC $W\colon g\rightarrow\ketbra{\psi_g}{\psi_g}$, the Gram matrix $\Gamma_{W}\in \mathbb{C}^{|\cG|\times |\cG|}$ is  a matrix which is indexed by group elements of $\cG$ and corresponds to the pairwise overlaps of output states $\{\ket{\psi_g}\}_{g\in\cG}$. More specifically, for all $g,g'\in \cG$, the elements of $\Gamma_{W}$ satisfy 
    \begin{align*}
       (\Gamma_{W})_{g,g'}=\braket{\psi_{g}}{\psi_{g'}}. 
    \end{align*}
\end{defn}
%  Consider the Gram matrix $\Gamma_W$ for the group-covariant PSC $W$. Let the elements of the Gram matrix $\Gamma_W$ formed by the collection of states $\{\ket{\psi_g}\}_{g\in\cG}$ be indexed by group elements as well. More specifically,$\forall g,g'\in \cG$, the Gram matrix element satisfies 
% \begin{align*}
%     (\Gamma_{W})_{g,g'}=\braket{\psi_{g}}{\psi_{g'}}.
% \end{align*}
\begin{lem}
    Consider a $\cG$-covariant PSC $W\colon g\rightarrow \ketbra{\psi_g}{\psi_g}$ where $\cG$ is a finite group. Then, the Gram matrix $\Gamma_W$ formed by the states $\{\ket{\psi_g}\}_{g\in\cG}$ is $\cG$-circulant; i.e., the elements of $\Gamma_W$ satisfy 
    \begin{align*}
        (\Gamma_W)_{g,g'}=\gamma_{g^{-1}g'}
    \end{align*}
    for all $g,g'\in \cG$ and $\{\gamma_{g}\}_{g\in \cG}$ are the elements of the first row of $\Gamma_W$.
    \end{lem}

\begin{proof}
    Let $e\in \cG$ be the identity element. Then we compute $(\Gamma_{W})_{g,g'}$ as follows
    \begin{align*}
        (\Gamma_W)_{g,g'}  =\braket{\psi_{g}}{\psi_{g'}}
     = \bra{\psi_{e}}U^{\dagger}_{g}U_{g'}\ket{\psi_{e}}.
        \end{align*}
Since $\{U_{g}\}_{g\in \cG}$ is a unitary representation, it satisfies $U_{g}^{\dagger}=U_{g^{-1}}$. Hence, we get 
\begin{align*}
(\Gamma_W)_{g,g'}=\bra{\psi_{e}}U_{g^{-1}g'}\ket{\psi_{e}}.
\end{align*}
Defining the elements of the first row of the Gram matrix $\Gamma_W$ as follows
\begin{align*}
    \gamma_{g}\coloneqq (\Gamma_{W})_{e,g}=\braket{\psi_e}{\psi_g}=\bra{\psi_e}U_{g}\ket{\psi_e},
\end{align*}
completes the proof.
\end{proof}
For the rest of the paper, we will restrict ourselves to finite abelian groups. In Lemma~\ref{lem:gram diagonalization}, we  establish the key property for the Gram matrix $\Gamma_W$ associated with $\cG$-covariant PSC when $\cG$ is an abelian group. More specifically, we show that the Gram matrix $\Gamma_W$ is diagonal in the basis formed by the 
character states $\{\ket{\chi}\}_{\chi\in \hcG}$. 
\begin{lem}\label{lem:gram diagonalization}
  Consider the Gram matrix  $\Gamma_W$ for a $\cG$-covariant PSC $W\colon g\rightarrow \ketbra{\psi_g}{\psi_g}$. When $\cG$ is abelian,  the character basis $\{\ket{\chi}\}_{\chi\in \hcG}$ diagonalizes $\Gamma_W$ and the eigenvalue corresponding to the basis element $\ket{\chi}$ is $\lambda_{\chi}$, which satisfies 
    \begin{align*}
        \lambda_{\chi}=\sum_{g\in \cG}\gamma_{g}\chi^{-1}(g).
    \end{align*}
\end{lem}
\begin{proof}
    For character states $\{\ket{\chi}\}_{\chi\in \hcG}$ and the Gram matrix $\Gamma_W$, the following relation holds 
\begin{align*}
    (\Gamma_W)_{g,:}\ket{\chi} & =\frac{1}{\sqrt{|\cG|}}\sum_{g'\in \cG}\gamma_{g'g^{-1}}\chi^{-1}(g')\\
    & =\frac{1}{\sqrt{|\cG|}}\sum_{g''\in \cG}\gamma_{g''}\chi^{-1}(g''g)\\
    & =\chi^{-1}(g)\frac{1}{\sqrt{|\cG|}}\sum_{g''\in \cG}\gamma_{g''}\chi^{-1}(g'')\\
    & =\chi^{-1}(g)\frac{\lambda_{\chi}}{\sqrt{|\cG|}}.
\end{align*}

This implies that $\ket{\chi}$ is an eigenvector of $\Gamma_W$, satisfying 
\begin{align}
    \Gamma_W\ket{\chi} & =\frac{\lambda_{\chi}}{\sqrt{|\cG|}}\sum_{g\in \cG}\chi^{-1}(g)\ket{g}\\
    & = \lambda_{\chi}\ket{\chi}. \qedhere
\end{align}
\end{proof}

    We denote the eigenvalues of the Gram matrix $\Gamma_W$ using  the ordered list $\blambda=\{\lambda_{\chi}\}_{\chi\in \hcG}$, which we refer to as its eigen list. Note that we use the specific ordering implied by the stated Fourier transform and its correspondence to eigenvectors. For the eigen list $\blambda=\{\lambda_{\chi}\}_{\chi\in \hcG}$, let $\bmu=\{\mu_{\chi}\}_{\chi\in \hcG}$ denote the normalized eigen list where $\mu_{\chi}=\frac{\lambda_{\chi}}{|\cG|}$.
We define the canonical choice of quantum state $\ket{\psi_g}$ $\forall g\in \cG$ as follows
\begin{align*}
    \ket{\psi_g}\coloneqq \frac{1}{\sqrt{|\cG|}}\sum_{\chi\in \hcG}\sqrt{\lambda_{\chi}}\chi(g)\ket{\chi}
\end{align*}
\begin{lem}\label{lem: eigenlist trace relation}
    For Gram matrix $\Gamma_W$ with eigen list $\blambda=\{\lambda_{\chi}\}_{\chi\in \hcG}$ and normalized eigen list $\bmu=\{\mu_{\chi}\}_{\chi\in \hcG}$, it holds
    \begin{align}
        \sum_{\chi\in \hcG}\lambda_{\chi} =|\cG|, \qquad
        \sum_{\chi\in \hcG}\mu_{\chi}  =1.
    \end{align}
   \end{lem}
\begin{proof}
    Observe that 
    \begin{align}
        \sum_{\chi\in \hcG}\lambda_{\chi} = \Tr(\Gamma_W)
        = \sum_{g\in \cG}\braket{\psi_g}{\psi_g}
         = |\cG|.
    \end{align}
    Hence, we get 
    \begin{align*}
        \sum_{\chi\in \hcG}\mu_{\chi}= \frac{1}{|\cG|}\sum_{\chi\in \hcG}\lambda_{\chi}=1. \qquad\qedhere
    \end{align*}
\end{proof}
Since the Gram matrix $\Gamma_{W}$ is positive semidefinite, we have $\lambda_{\chi}\geq 0$ for all $\chi\in \hcG$. Hence, the normalized eigen list $\bmu$ forms a probability distribution on $\hcG$.

\begin{lem}
     All $\cG$-covariant PSCs $W\colon g\rightarrow \ketbra{\psi_{g}}{\psi_g}$ with $g\in \cG$ are determined (up to isometric equivalence) by their eigen list $\blambda$.
\end{lem}
\begin{proof}
Recall the canonical choice of state $\ket{\psi_g}$ for Gram matrix $\Gamma_W$ with eigen list $\blambda$ as
\begin{align}
    \ket{\psi_g} &\coloneqq \frac{1}{\sqrt{|\cG|}}\sum_{\chi\in \hcG}\sqrt{\lambda_{\chi}}\chi(g)\ket{\chi}
\end{align}
The inner product $\braket{\psi_{g}}{\psi_{g'}}$ satisfies 
\begin{align*}
    \braket{\psi_{g}}{\psi_{g'}} &\ =\frac{1}{|\cG|}\sum_{\chi\in \hcG}\sum_{\chi'\in \hcG}\sqrt{\lambda_{\chi}\lambda_{\chi'}}\chi'^{-1}(g)\chi(g')\braket{\chi'}{\chi}\\
    &\ = \frac{1}{|\cG|}\sum_{\chi\in \hcG}\lambda_{\chi}\chi(g^{-1}g')\\
    &\ = \frac{1}{|\cG|}\sum_{\chi\in \hcG}\sum_{b\in \cG}\gamma_{b}\chi^{-1}(b)\chi(g^{-1}g')\\
    &\ =\frac{1}{|\cG|}\sum_{b\in \cG}\gamma_{b}\sum_{\chi\in \hcG}\chi(b^{-1}g^{-1}g')\\
    &\ =\frac{1}{|\cG|}\sum_{b\in \cG}\gamma_b |\cG|\delta_{b,g^{-1}g'}\\
    % &\ =\frac{1}{|\cG|}\sum_{\chi\in \hcG}\gamma_{g'g^{-1}}\\
    &\ =\gamma_{g^{-1}g'}= (\Gamma_W)_{g,g'}.
\end{align*}
Since two PSCs $W\colon g\rightarrow \ket{\psi_{g}}$ and $W'\colon g\rightarrow \ket{\psi_{g}'}$ with $g\in \cG$ are  isometrically (or unitarily) equivalent when there exists an isometry (or unitary) $V$ such that 
\begin{align*}
    \ket{\psi_{g}'}=V\ket{\psi_g}, \forall g\in\cG,
\end{align*}
thus all the isometrically (or unitarily) equivalent group-covariant PSCs with Gram matrix are well defined using eigenvalue spectrum $\blambda$ associated with Gram matrix $\Gamma_W$. 
\end{proof}
\subsection{Information Measures}
For a group-covariant PSC $W$, for the rest of the paper, we consider uniform input distribution on group $\cG$.
To characterize $W$ under the uniform input ensemble, we consider two information measures, namely the symmetric Holevo information, denoted by $I(W)$ and  the channel fidelity, denoted by $F(W)$.
The symmetric Holevo information measures the maximum amount of classical information that can be sent through the channel $W$ for uniform input distribution. More formally, for a CQ channel $W$, the symmetric Holevo information is computed as 
\begin{align*}
    I(W)=S\Bigg(\frac{1}{|\cG|}\sum_{g\in \cG}W(g)\Bigg)-\sum_{g\in \cG}\frac{1}{|\cG|}S\Bigg(W(g)\Bigg)
\end{align*}
where $S(\rho)=-\Tr(\rho\log\rho)$ is von Neumann entropy for the density matrix $\rho$.
When $W\colon g\rightarrow\ketbra{\psi_g}{\psi_g}$ $\forall g\in \cG$ is a $\cG$-covariant PSC, the symmetric Holevo information satisfies
\begin{align*}
    I(W) & =S\Bigg(\frac{1}{|\cG|}\sum_{g\in \cG}\ketbra{\psi_g}{\psi_g}\Bigg)-\frac{1}{|\cG|}\sum_{g\in \cG}S(\ketbra{\psi_g}{\psi_g})\\
   &  = S\Bigg(\frac{1}{|\cG|}\sum_{g\in \cG}\ketbra{\psi_g}{\psi_g}\Bigg),
\end{align*}
where we use the fact that $S(\ketbra{\psi_g}{\psi_g})=0$ $\forall g\in \cG$.
The channel fidelity characterizes expected overlap between the non-identical output quantum states of channel $W$. Formally, it is written as 
\begin{align*}
    F(W)
    & =\frac{1}{|\cG|(|\cG|-1)}\!\sum_{g,g'\in \cG:, g\neq g'}\!\!\!\!\!\!\!\Tr\left(\sqrt{\sqrt{W(g)}W(g')\sqrt{W(g)}}\right).
\end{align*}
For $\cG$-covariant PSC $W$, the expression reduces to
\begin{align*}
     F(W)=\frac{1}{|\cG|(|\cG|-1)}\sum_{g,g'\in\cG\colon g\neq g'}\left|\braket{\psi_g}{\psi_{g'}}\right|.
\end{align*}
\begin{lem}\label{lem:Holevo information and channel fidelity}
For abelian group $\cG$, consider the $\cG$-covariant
PSC $W\colon g\rightarrow \ketbra{\psi_{g}}{\psi_g}$. The Gram matrix $\Gamma_W$ is $\cG$-circulant and is characterized via eigen list $\blambda=\{\lambda_{\chi}\}_{\chi\in \hcG}$. Then, for the uniform distribution over input alphabets $g\in \cG$, the symmetric Holevo information $I(W)$ and the channel fidelity $F(W)$ satisfy
\begin{align*}
    I(W) &\ = H(\bmu)\\
  % F(W)  &\  =\frac{1}{q(q-1)}\left(\sum_{u\in [q]}\lambda_{u}^2-q\right)
  F(W) & =\frac{1}{|\cG|-1}\sum_{g\in \cG, g\neq e}|\gamma_{g}|,
\end{align*}
where $\bmu=\{\mu_{\chi}\}_{\chi\in \hcG}$ is the normalized eigen list with $\mu_{\chi}=\frac{\lambda_{\chi}}{|\cG|}$ $\forall \chi\in \hcG$ and $H(\bmu)$ is the Shannon entropy for $\bmu$ and $e$ is the identity element of $\cG$.
\end{lem}
\begin{proof}
    Consider the matrix $\Psi$ whose columns correspond to the quantum states $\{\ket{\psi_{g}}\}_{g\in \cG}$.
% \begin{align*}
%     \Psi\coloneqq \begin{bmatrix}
%         \ket{\psi_0}\dots \ket{\psi_{q-1}}
%     \end{bmatrix}
% \end{align*}
Then we have
\begin{align*}
    \Psi\Psi^{\dagger}= \sum_{g\in \cG}\ketbra{\psi_g}{\psi_g}.
\end{align*}
 The matrix $\Psi^{\dagger}\Psi=\Gamma_W$ as it satisfies 
\begin{align*}
    (\Psi^{\dagger}\Psi)_{g,g'}=\braket{\psi_g}{\psi_g'} =(\Gamma_W)_{g,g'},
\end{align*}
for all $g,g'\in \cG$. Since the eigenvalues of $\Psi^{\dagger}\Psi$ are the same as of $\Psi\Psi^{\dagger}$, we have 
\begin{align*}
    I(W) &\ = S\left(\frac{1}{|\cG|}\Psi\Psi^{\dagger}\right)\\
    &\ = -\sum_{\chi\in \hcG}\frac{\lambda_{\chi}}{|\cG|}\log \frac{\lambda_{\chi}}{|\cG|}\\
    & = -\sum_{\chi\in \hcG}\mu_{\chi}\log\mu_{\chi}\\
    & = H(\bmu).
\end{align*}

Writing the channel fidelity $F(W)$ in terms of the Gram matrix elements, we obtain 
\begin{align*}
 F(W) & =\frac{1}{|\cG|(|\cG|-1)}\sum_{g,g'\in\cG:g\neq g'}|\braket{\psi_g}{\psi_{g'}}|\\
 &\ =   \frac{1}{|\cG|(|\cG|-1)}\Bigg(\sum_{g,g'\in \cG:g\neq g'}|\gamma_{g^{-1}g'}|\Bigg)\\
 &\ =\frac{1}{|\cG|(|\cG|-1)}|\cG|\bigg(\sum_{g\in \cG,g\neq e}|\gamma_{g}|\bigg)\\
 & =\frac{1}{|\cG|-1}\sum_{g\in \cG,g\neq e}|\gamma_g| .\qedhere
\end{align*}
\end{proof}
\subsection{Pretty Good Measurement on Group-Covariant PSC}
A pretty good measurement (PGM) \cite{hausladen1994pretty,holevo1978asymptotically}, which is also known as the square root measurement, is defined as follows. Consider the  $\cG$-covariant CQ channel $W\colon g\rightarrow \rho_g$ with uniform input distribution on $\cG$. For the output state $\rho_g$ for $g\in \cG$, the measurement operator is $M_{g}=\frac{1}{|\cG|}\left(\bar{\rho}\right)^{-\frac{1}{2}}\rho_g\left(\bar{\rho}\right)^{-\frac{1}{2}}$ where $\bar{\rho}=\sum_{g\in \cG}\frac{1}{|\cG|}\rho_g$ and the inverse is taken on the support of $\bar{\rho}$.  This $|\cG|$-outcome measurement $\{M_g\}_{g\in \cG}$ is shown to be optimal \cite{eldar2002quantum} for  minimizing the probability of state-discrimination error for geometrically uniform quantum states. In this paper, to analyze the error rate in the quantum message passing framework,
we will use PGM to characterize the group symbol error probability. The following lemma establishes the relation between the error probability for discriminating the output states of $\cG$-covariant PSC $W$ in terms of the eigen list of the Gram matrix. 
\begin{lem}\label{lem:pgm}
Consider group-covariant PSC $W\colon g\rightarrow \ketbra{\psi_{g}}{\psi_g}$ with $g\in \cG$ with group-circulant Gram matrix $\Gamma_W$ and eigen list $\blambda=\{\lambda_{\chi}\}_{\chi\in \hcG}$. Then, for the uniform distribution over input alphabets $g\in \cG$, the probability of error $\perr(W)$ for distinguishing the states optimally using pretty good measurement satisfies 
\begin{align*}
    \perr(W)=1-\left(\frac{1}{|\cG|}\sum_{\chi\in \hcG}\sqrt{\lambda_{\chi}}\right)^{2}.
\end{align*}
\end{lem}

\begin{proof}
    For the states $\{\ket{\psi_g}\}_{g\in \cG}$, the pretty good measurement is constructed using measurement operators $\{\Pi_g\}_{g\in \cG}$ such that 
    \begin{align*}
        \Pi_g=\ketbra{\tau_g}{\tau_g},
    \end{align*}
    where $\ket{\tau_g}$ = $\sqrt{\frac{1}{|\cG|}}\bar{\rho}^{-\frac{1}{2}}\ket{\psi_g}$ and the average density matrix $\bar{\rho}=\sum_{g\in \cG}\frac{1}{|\cG|}\ketbra{\psi_g}{\psi_g}$. For the uniform input distribution, we have 
    \begin{align*}
    \bar{\rho}=\frac{1}{|\cG|}\Psi\Psi^{\dagger}=\sum_{\chi\in \hcG}\frac{\lambda_{\chi}}{|\cG|}\ketbra{\chi}{\chi}.
    \end{align*}
    Thus, $\ket{\tau_g}$ can be written as 
    \begin{align*}
        \ket{\tau_g}
        & =\frac{1}{\sqrt{|\cG|}}(\sum_{\chi'\in \hcG}\sqrt{\frac{|\cG|}{\lambda_{\chi'}}}\ketbra{\chi'}{\chi'})(\frac{1}{\sqrt{|\cG|}}\sum_{\chi\in \hcG}\sqrt{\lambda_{\chi}}\chi(g)\ket{\chi})\\
        & = \frac{1}{\sqrt{|\cG|}}\sum_{\chi\in \hcG}\chi(g)\ket{\chi}.
    \end{align*}
Hence, we compute $\braket{\tau_g}{\psi_g}$ as
\begin{align*}
    \braket{\tau_g}{\psi_g} & =\frac{1}{|\cG|}\sum_{\chi\in \hcG}\sum_{\chi'\in\hcG}\sqrt{\lambda_{\chi'}}\chi'\chi^{-1}(g)\braket{\chi}{\chi'}\\
    & =\frac{1}{|\cG|}\sum_{\chi\in \hcG}\sqrt{\lambda_{\chi}}.
\end{align*}
Finally, the average success probability of the measurement satisfies 
\begin{align*}
    \psuc(W)& =\frac{1}{|\cG|}\sum_{g\in \cG}\Tr(\Pi_g\ketbra{\psi_g}{\psi_g})\\
    & = \frac{1}{|\cG|}\sum_{g\in \cG}|\braket{\tau_g}{\psi_g}|^{2}\\
    & =\frac{1}{|\cG|}\sum_{g\in \cG}\left(\frac{1}{|\cG|}\sum_{\chi\in \hcG}\sqrt{\lambda_{\chi}}\right)^{2}\\
    & = \left(\frac{1}{|\cG|}\sum_{\chi\in \hcG}\sqrt{\lambda_{\chi}}\right)^{2} .
\end{align*}
Thus, the error probability is given by
\begin{align*}
    \perr(W)=1-\psuc(W)=1-\left(\frac{1}{|\cG|}\sum_{\chi\in \hcG}\sqrt{\lambda_{\chi}}\right)^{2}.
\end{align*}
\end{proof}
% \subsection{Heralded Mixture of PSCs}
% \begin{defn}
%         A channel $W_{H}$ is the  heralded mixtures of group covariant PSCs, if the channel output can be characterized via an ensemble of group covariant PSCs. More generally, the output of channel $W_H$ $\forall g \in \cG$, can be written as 
%         \begin{align*}
%             W_H(g)=\sum_{x\in \cX}p_{x}W_{x}(g)\otimes \ketbra{x}{x}
%         \end{align*}
%     where $W_{x}\colon g\rightarrow \ketbra{\psi^{x}_g}{\psi^{x}_g}$ is a symmetric $\cG$-covariant PSC with group-circulant Gram matrix $\Gamma_{W_x}$ and the eigen list $\blambda_x$ $\forall x \in \cX$, $\{\ket{x}\}_{x\in \cX}$ is an orthonormal basis and $\{p_{x}\}_{x\in \cX}$ is some probability distribution in finite sized alphabet $\cX$.
%     \end{defn}

\subsection{Heralded Mixture of Group-Covariant PSCs}

In general, the local quantum message-passing rules studied in Section~\ref{sec:quantum message passing factor graph lemmas}
can produce outputs including an additional classical side register that distinguishes between multiple $\cG$-covariant PSCs. Such outputs are known as heralded mixtures.

\begin{defn}
Let $\cG$ be a finite abelian group. A CQ channel $W_H\colon g\to W_H(g)$, $g\in\cG$, is called a
\emph{heralded mixture of $\cG$-covariant PSCs} if there exist
\begin{itemize}
    \item a finite set $\cX$,
    \item a probability distribution $\{p_x\}_{x\in\cX}$,
    \item an orthonormal basis $\{\ket{x}\}_{x\in\cX}$, and
    \item for each $x\in\cX$, a $\cG$-covariant PSC
    $W_x\colon g\to \ketbra{\psi_g^x}{\psi_g^x}$
\end{itemize}
such that
\begin{align*}
    W_H(g)=\sum_{x\in\cX} p_x\, W_x(g)\otimes \ketbra{x}{x},
    \qquad g\in\cG .
\end{align*}
We denote by $\cM(\cG)$ the class of all heralded mixtures of $\cG$-covariant PSCs.
\end{defn}

Conditioned on the classical label $x$, the channel $W_H$ reduces to the single $\cG$-covariant PSC $W_x$.
By the preceding discussion, each component $W_x$ is determined up to isometric equivalence by its eigen list.
Hence a channel in $\cM(\cG)$ may be specified equivalently by a finite ensemble
$\{p_x,\blambda_x\}_{x\in\cX}$, where $\blambda_x$ denotes the eigen list of $W_x$.
\section{Message Passing on Factor Graphs}\label{sec:quantum message passing factor graph lemmas}
In this section, we study how local factor operations act on channels whose input alphabets are finite abelian groups. Our analysis proceeds in two steps. We first define, for each primitive factor, the induced CQ channel obtained when the incoming messages are arbitrary group-covariant CQ channels, and we show that these induced channels remain group-covariant. We then specialize to the case where the incoming messages are group-covariant PSCs, or more generally elements of $\cM(\cG)$ on the appropriate input alphabets, and derive explicit quantum message-passing update rules in terms of the corresponding eigen lists and herald registers. This two-level description separates the structural covariance of the induced channels from the more refined character-domain parametrization that will be used throughout the remainder of the paper.
We begin by defining the induced CQ channel associated with a factor whose incident edges carry group-covariant CQ channels.
% In this section, we develop quantum message passing update rules based on classical factors acting on elements from abelian groups. To achieve that, we first describe the notion of induced CQ channels from a factor in which each input edge is associated with a group covariant CQ channel. 
More specifically, consider the abelian groups $\cG_1$, $\cG_2$, $\dots$, $\cG_m$ and a factor node $\phi$ such that it takes group elements  $g_1\in \cG_1$ and $g_2\in \cG_2$, $\dots$, $g_m\in \cG_m$ as inputs and outputs an element $h$ in some abelian group $\cH$.  For $g_1\in \cG_1$ and $g_2\in \cG_2$, $\dots$, $g_m\in \cG_m$ and $h\in \cH$, define the factor function as 
\begin{align*}
    f(g_1,\dots,g_m;h)=\mone\{h=\phi(g_1,\dots,g_m)\}.
\end{align*}
Let $W_i$ denote a group covariant CQ channel with input $g_{i}\in \cG_i$.
% Let each input $g_i$ correspond to a group-covariant CQ channel $W_{i}\colon g_i\rightarrow \rho^{(i)}_{g_i}$ with $g_i\in \cG_i$. 
Let each channel $W_{i}\colon g_{i}\rightarrow \rho^{(i)}_{g_i}$ be a $\cG_i$-covariant CQ channel with respect to a unitary representation $\{U^{(i)}_{g}\}_{g\in \cG_i}$. 

The effective CQ channel induced by the factor function $f$ is denoted by $W_{\phi}\colon h\rightarrow \rho_h$ such that each output state $\rho_h$ is given by
\begin{align*}
    \rho_{h} & =\frac{1}{N_{h}}\sum_{(g_1,\dots, g_m)\in \cG_{1}\times\dots\times \cG_m\colon h=\phi(g_1,\dots, g_m)}\bigotimes_{i=1}^{m}W_{i}(g_i)\\
    & = \frac{1}{N_{h}}\sum_{(g_1,\dots, g_m)\in \cG_{1}\times\dots\times \cG_m\colon h=\phi(g_1,\dots, g_m)}\bigotimes_{i=1}^{m}\rho^{(i)}_{g_i},
\end{align*}
where $N_h>0$ is the number of satisfied tuples $(g_1,\dots, g_m)\in \cG_{1}\times\dots\times \cG_m$ such that $h=\phi(g_1,\dots, g_m)$. Next, we describe some factors and the corresponding induced CQ channels.
\begin{itemize}
    \item \textbf{Check Factor:} The check factor with two input edges corresponds to the parity equation in the group $\cG$, such that for $g_1,g_2,h\in \cG$, the factor function $f_{\cnop}$ satisfies 
    \begin{align*}
        f_{\cnop}(g_1,g_2;h) 
        % &= \mone\{g_{1}g_{2}h=0\}\\
        & =\mone\{h=g_{1}g_{2}\}.
    \end{align*}
    For two $\cG$-covariant CQ channels  $W_1:g\rightarrow \rho^{(1)}_{g}$ and $W_2:g\rightarrow \rho^{(2)}_{g}$ where $g\in \cG$, the induced CQ channel $W_1\cnop W_2\colon h\rightarrow \rho_{h}^{\cnop}$ outputs the state $\rho_{h}^{\cnop}$ $\forall h\in \cG$ such that the following holds 
    \begin{align*}
       \rho_{h}^{\cnop}
       & = \frac{1}{|\cG|}\sum_{g_{1}\in \cG}W_{1}(g_1)\otimes W_{2}(g_{1}^{-1}h)\\
       & =\frac{1}{|\cG|}\sum_{g_1\in \cG}\rho^{(1)}_{g_1}\otimes \rho^{(2)}_{g_{1}^{-1}h}.
    \end{align*}
Similarly, when the check factor contains three input edges, the factor function satisfies 
\begin{align*}
    f_{\cnop(2)}(g_1,g_2,g_3;h)=\mone\{h=g_1g_2g_3\}.
\end{align*}
Considering each channel $W_{i}\colon g\rightarrow\rho^{(i)}_g$ to be $\cG$-covariant CQ channel, the output states of the induced channel $W_1\cnop W_2\cnop W_3 \colon h\rightarrow \rho_h^{\cnop(2)}$  $h\in \cG$ can be written as 
\begin{align*}
    \rho^{\cnop(2)}_h & =\frac{1}{|\cG|^2}\sum_{g_1,g_2\in \cG}\rho^{(1)}_{g_1}\otimes \rho^{(2)}_{g_{2}}\otimes \rho^{(3)}_{g_{1}^{-1}g_{2}^{-1}h}\\
    & = \frac{1}{|\cG|}\sum_{u\in \cG}\left(\frac{1}{|\cG|}\sum_{g_1\in \cG}\rho^{(1)}_{g_1}\otimes \rho^{(2)}_{g_{1}^{-1}u}\right)\otimes \rho^{(3)}_{u^{-1}h}\\
    & = \frac{1}{|\cG|}\sum_{u\in \cG}\rho^{\cnop}_{u}\otimes \rho^{(3)}_{u^{-1}h}.
\end{align*}
% Iteratively, this can be generalized 
Iterating the same argument yields to the check factor with $d$ input edges such that the factor function satisfies 
\begin{align*}
    f_{\cnop(d-1)}(g_1,\dots,g_d;h)=\mone\{h=g_1\dots g_d\},
\end{align*}
and the output of induced CQ channel $\cnop_{i=1}^{d}W_{i} \colon h\rightarrow \rho_h^{\cnop(d-1)}$ satisfies the following recursive relation
\begin{align*}
    \rho_{h}^{\cnop(d-1)}=\frac{1}{|\cG|}\sum_{g\in 
    \cG}\rho^{\cnop(d-2)}_{g}\otimes \rho^{(d)}_{g^{-1}h},
\end{align*}
where we use the notation $\cnop_{i=1}^{d}W_{i}$ to denote the channel $W_1\cnop\dots\cnop W_d$. Define unitary $V^{\cnop (d-1)}_{g}$ for $g\in \cG$ as follows 
\begin{align*}
   V^{\cnop (d-1)}_{g}=U^{(1)}_{g}\otimes \mI\otimes \dots\otimes \mI.
\end{align*}
Then applying $V^{\cnop (d-1)}_{g}$ on the state $\rho_{h}^{\cnop(d-1)}$ we get 
\begin{align*}
    V^{\cnop (d-1)}_{g}\rho_{h}^{\cnop(d-1)} (V^{\cnop (d-1)}_{g})^{\dagger} & = V^{\cnop (d-1)}_{g} \left(\frac{1}{|\cG|^{d-1}}\sum_{g_1,\dots g_{d-1}\in \cG}\rho_{g_{1}}^{(1)}\otimes \rho^{(d-1)}_{g_{d-1}}\otimes \rho^{(d)}_{(g_1\dots g_{d-1})^{-1}h}\right) (V^{\cnop (d-1)}_{g})^{\dagger}\\
    & = \frac{1}{|\cG|^{d-1}}\sum_{g_1,\dots g_{d-1}\in \cG}U^{(1)}_{g}\rho_{g_{1}}^{(1)}(U^{(1)}_{g})^{\dagger}\otimes \rho^{(d-1)}_{g_{d-1}}\otimes \rho^{(d)}_{(g_1\dots g_{d-1})^{-1}h}\\
    & = \frac{1}{|\cG|^{d-1}}\sum_{g_1,\dots g_{d-1}\in \cG}\rho_{gg_{1}}^{(1)}\otimes \rho^{(d-1)}_{g_{d-1}}\otimes \rho^{(d)}_{(g_1\dots g_{d-1})^{-1}h}.
\end{align*}
By re-indexing $g_1'=gg_1$ we get 
\begin{align*}
    V^{\cnop (d-1)}_{g}\rho_{h}^{\cnop(d-1)} (V^{\cnop (d-1)}_{g})^{\dagger}& =\frac{1}{|\cG|^{d-1}}\sum_{g_1',\dots g_{d-1}\in \cG}\rho_{g_{1}'}^{(1)}\otimes \rho^{(d-1)}_{g_{d-1}}\otimes \rho^{(d)}_{(g_1'\dots g_{d-1})^{-1}gh} \\
    & = \rho^{\cnop (d-1)}_{gh}.
\end{align*}
Thus, this proves the induced CQ channel $\cnop_{i=1}^{d}W_{i}$ is $\cG$-covariant.
\item \textbf{Equality Factor:} The equality factor defines the constraint that each edge carry the same identical group element of a group $\cG$. More specifically, consider the equality factor with two input edges. Then for $g_1,g_2,h\in \cG$, the factor function $f_{\vnop}$ satisfies 
\begin{align*}
    f_{\vnop}(g_1,g_2;h)=\mone\{h=g_1=g_2\}.
\end{align*}
Then, for $\cG$-covariant CQ channels $W_{i}\colon g\rightarrow \rho^{(i)}_{g}$, the induced CQ channel $W_1\vnop W_2\colon h\rightarrow \rho_{h}^{\vnop} $ outputs state $\rho_{h}^{\vnop}$ $\forall h\in \cG$, such that 
\begin{align*}
    \rho_{h}^{\vnop} & = W_{1}(h)\otimes W_{2}(h)\\
    & = \rho^{(1)}_h\otimes \rho^{(2)}_h.
\end{align*}
Similarly, the equality factor function with $d$ input edges satisfies 
\begin{align*}
    f_{\vnop(d-1)}(g_1,\dots,g_d;h)=\mone\{h=g_1=\dots=g_d\}.
\end{align*}
The output of the induced CQ channel $\vnop_{i=1}^{d}W_i$ satisfies the following recursive relation $\forall h\in \cG$
\begin{align*}
    \rho^{\vnop(d-1)}_{h}=\rho_{h}^{\vnop(d-2)}\otimes \rho^{(d)}_h,
\end{align*}
where $\vnop_{i=1}^{d}W_i$ refers to the channel $W_{1}\vnop\dots \vnop W_d$. 
Consider the unitary $V^{\vnop(d-1)}_{g}$ defined for $g\in\cG$, as follows 
\begin{align*}
    V^{\vnop(d-1)}_{g}\coloneqq \bigotimes_{i=1}^{d}U^{(i)}_{g}.
\end{align*}
Then we immediately get 
\begin{align*}
     V^{\vnop(d-1)}_{g}\rho^{\vnop(d-1)}_{h}(V^{\vnop(d-1)}_{g})^{\dagger} & =\bigotimes_{i=1}^d U^{(i)}_{g}\rho_{h}^{(i)} (U^{(i)}_{g})^{\dagger}\\
     & = \bigotimes_{i=1}^{d}\rho_{gh}^{(i)}\\
     & = \rho^{\vnop (d-1)}_{gh}.
\end{align*}
This means that the induced CQ channel $\vnop_{i=1}^{d}W_i$ is also $\cG$-covariant.
\item \textbf{Homomorphism Factor:} 
% The homomorphism factor refers to the constraint being a group homomorphism. 
We next consider a factor defined by a group homomorphism.
Let $\cG$ and $\cH$ be two abelian groups and $\phi \colon \cG\rightarrow \cH$ be a homomorphism. Then, for a single input, the factor function $f_{\phi}$ satisfies 
\begin{align*}
    f_{\phi}(g;h)=\mone\{h=\phi(g)\},
\end{align*}
for $g\in \cG$ and some $h\in \cH$. Let $W\colon g\rightarrow \rho_g$ be a $\cG$-covariant CQ channel, then the induced CQ channel $W_{\phi}\colon h\rightarrow \rho^{\phi}_h$ with $h\in \cH$, satisfies the following form
\begin{align*}
     \rho^{\phi}_h & = \frac{1}{|\ker\phi|}\sum_{g\in \cG\colon \phi(g)=h }W(g)\\
     & = \frac{1}{|\ker\phi|}\sum_{g\in \cG\colon \phi(g)=h}\rho_g.
\end{align*}
This extends directly to a factor with $d$ inputs, where we consider inputs from abelian groups $\cG_1,\dots, \cG_d$ and homomorphism $\phi\colon \cG_1\times \dots \times\cG_d\rightarrow \cH$. Let $W_{i}\colon g_{i}\rightarrow \rho^{(i)}_{g_i}$ be $\cG_i$-covariant CQ channels for $g_{i}\in \cG_i$, then the induced CQ channel $W_{\phi}$ has the output state of the form 
\begin{align*}
     \rho^{\phi}_h & = \frac{1}{|\ker\phi|}\sum_{(g_1,\dots,g_d)\in \cG_1\times\dots\times\cG_d\colon \phi(g_1,\dots,g_d)=h }\bigotimes_{i=1}^{d} W_{i}(g_i)\\
     & = \frac{1}{|\ker\phi|}\sum_{(g_1,\dots,g_d)\in \cG_1\times\dots\times\cG_d\colon \phi(g_1,\dots,g_d)=h }\bigotimes_{i=1}^{d} \rho^{(i)}_{g_i}.
\end{align*}
Now, consider element $(g_1',\dots,g_d')\in \cG_1\times \dots \times \cG_{d}$ such that $\phi(g_1',\dots,g_d')=b$ with $b\in \image\phi$. Then the following relation holds
\begin{align*}
    U^{(1)}_{g_1'}\otimes \dots \otimes U^{(d)}_{g_d'}\rho^{\phi}_h 
    (U^{(1)}_{g_1'})^{\dagger}\otimes \dots \otimes (U^{(d)}_{g_d'})^{\dagger}  & = \frac{1}{|\ker\phi|}\sum_{(g_1,\dots,g_d)\in \cG_1\times\dots\times\cG_d\colon \phi(g_1,\dots,g_d)=h }\bigotimes_{i=1}^{d} U^{(i)}_{g_{i}'}\rho^{(i)}_{g_i}(U^{(i)}_{g_i'})^{\dagger}\\
    & = \frac{1}{|\ker\phi|}\sum_{(g_1,\dots,g_d)\in \cG_1\times\dots\times\cG_d\colon \phi(g_1,\dots,g_d)=h }\bigotimes_{i=1}^{d}\rho^{(i)}_{g_i'g_i}\\
    & = \frac{1}{|\ker\phi|}\sum_{(g_1'',\dots,g_d'')\in \cG_1\times\dots\times\cG_d\colon \phi(g_1'',\dots,g_d'')=bh }\bigotimes_{i=1}^{d}\rho^{(i)}_{g_i''}\\
    & = \rho^{\phi}_{bh}.
\end{align*}
This implies that the induced CQ channel $W_{\phi}$ inherits a natural group-covariant structure from the incoming channels. In particular, if $\phi$ is surjective, then $W_{\phi}$ becomes $\cH$-covariant.
\item \textbf{Marginalization Factor:} Marginalization factor is one special case of homomorphism factor that takes inputs from abelian groups $\cG_1,\dots, \cG_d$ and drops an element from some group $\cG_{i}$ where $1\leq i\leq d$. The factor function $f_{\phi}$ satisfies 
\begin{align*}
    % f_{\phi_{\text{marginal}}}\left((g_1,\dots,g_d); (g_1,\dots,g_{i-1},g_{i+1},\dots,g_d)\right)=\mone\{(g_1,\dots,g_{i-1},g_{i+1},\dots,g_d)=\phi(g_1,\dots,g_d)\}.
    \phi(g_1,\dots,g_d)=  (g_1,\dots,g_{i-1},g_{i+1},\dots,g_d)
\end{align*}
The induced CQ channel $W_{\phi}$ outputs a marginalized state of the form 
\begin{align*}
    \rho_{(g_1,\dots,g_{i-1},g_{i+1},\dots,g_d)}^{\phi}=\frac{1}{|\cG_i|}\sum_{g\in \cG_i}\bigotimes_{j=1}^{i-1} W_{j}(g_j)\otimes W_{i}(g)\bigotimes_{k=i+1}^{d}W_{k}(g_k).
\end{align*}
Consider the unitary $V^{\phi}_{(g'_1,\dots,g'_{i-1},g'_{i+1},\dots,g'_d)}$ defined by
\begin{align*}
    V^{\phi}_{(g'_1,\dots,g'_{i-1},g'_{i+1},\dots,g'_d)}= U^{(1)}_{g'_{1}}\otimes \dots U^{(i-1)}_{g_{i-1}'}\otimes \mI\otimes U^{(i+1)}_{g_{i+1}'}\otimes \dots \otimes U^{(d)}_{g'_d}.
\end{align*}
Thus, applying $V^{\phi}_{(g'_1,\dots,g'_{i-1},g'_{i+1},\dots,g'_d)}$ on the state $\rho_{(g_1,\dots,g_{i-1},g_{i+1},\dots,g_d)}^{\phi}$ we get $\rho_{(g'_1g_1,\dots,g'_{i-1}g_{i-1},g'_{i+1}g_{i+1},\dots,g'_dg_d)}^{\phi}$ which implies the channel $W_{\phi}$ is $\cG_1\times \dots \cG_{i-1}\times\cG_{i+1}\times\dots\times\cG_d$-covariant. 
\item \textbf{Automorphism Factor:} The automorphism factor is a special case of homomorphism factor. In this case, the factor function $f_{\phi}$ is defined using automorphism $\phi\colon\cG \rightarrow \cG$ for some abelian group $\cG$.  Since $\phi$ is an automorphism, every output element $h$ has a unique preimage $g=\phi^{-1}(h)$. Thus, the induced CQ channel $W_{\phi}$ has output states of the form 
\begin{align*}
    \rho^{\phi}_h= W(\phi^{-1}(h))= \rho_{\phi^{-1}(h)}.
\end{align*}
Moreover, the induced channel $W_{\phi}$ is also $\cG$-covariant. If $W$ is $\cG$-covariant with unitary representation $\{U_{a}\}_{a\in \cG}$, then the output of the induced channel $\forall h_1,h_2\in \cG$ satisfies
\begin{align*}
    \rho^{\phi}_{h_1h_2} & =\rho_{\phi^{-1}(h_1h_2)}\\
    & = \rho_{\phi^{-1}(h_1)\phi^{-1}(h_2)}\\
    & = U_{\phi^{-1}(h_1)}\rho_{\phi^{-1}(h_2)}U_{\phi^{-1}(h_1)}^{\dagger}.
\end{align*}
This implies that the CQ channel $W_{\phi}$ is $\cG$-covariant.
\end{itemize}
% The above description of induced CQ channels for different factors is important for many reasons. This establishes recursive relations for multiple factors and each of the induced CQ channels for the above described factors remain group covariant. Our goal is to leverage the group symmetry for each of the induced CQ channels and to find an efficient characterization of the output states of those channels. This means that if we can find an efficient characterization of the induced CQ channels for these factors, any combination of these factors can also be efficiently  characterized as long as the overall factor graph remains a tree. 

% For the rest of the paper we will consider group covariant PSC and group covariant heralded mixture of PSCs. Using quantum message passing, we will establish that for the factors described above, the induced CQ channels remain group covariant heralded mixture of PSCs.  Consider group covariant PSCs for each group $\cG_1$, $\cG_2$, $\dots$, $\cG_m$. The quantum message passing utilizes the eigen lists $\blambda_k=\{\lambda_{\chi}^{(k)}\}_{\chi\in \hcG_k}$  $\forall k\in \{1,\dots,m\}$ and computes the sufficient statistics to characterize the output CQ channel induced by local factors. In subsequent sections, we will establish this for each of the factors. 

The preceding discussion identifies, for each primitive factor, the induced CQ channel obtained from incoming group-covariant CQ channels and shows that these induced channels remain group-covariant. This provides the structural layer of the quantum message-passing framework. We now specialize to the more structured setting of group-covariant PSCs, and more generally finite heralded mixtures of such PSCs, for which each local factor admits an explicit characterization in the character domain. 

For each primitive factor considered below, we proceed in the same way. We first identify the induced channel on the relevant output group, then exhibit an isometry or unitary that reveals the natural sufficient statistic, and finally conclude whether the output is a single group-covariant PSC or a finite heralded mixture of such PSCs. In this setting, the relevant sufficient statistics are the eigen lists $\blambda_k=\{\lambda^{(k)}_{\chi}\}_{\chi\in \widehat{\cG}_k}$ associated with the incoming messages.

\subsection{Quantum Message Passing for Check Factor}
% In this section, we develop the quantum message passing update rule for the check factor combining of two group covariant PSCs. More specifically, in Lemma~\ref{lem:check node abelian}, we establish that the induced CQ channel of the check factor is characterized via the heralded mixture of PSCs, where the probability and the eigen list for each PSC are calculated using the eigen list of the original group covariant PSCs associated with the incoming edges of the factor.

In this subsection, we analyze the check factor for two incoming $\cG$-covariant PSCs. The main point is that the check-factor output need not remain a single PSC. Rather, after a suitable change of basis, it decomposes into a finite heralded mixture of $\cG$-covariant PSCs. The next lemma gives the corresponding ensemble explicitly in terms of the incoming eigen lists.

\begin{lem}\label{lem:check node abelian}
For abelian group $\cG$, consider $\cG$-covariant PSCs $W_{1}$ and $W_{2}$, with Gram matrices $\Gamma_{W_1}$ and $\Gamma_{W_2}$ having eigen lists $\blambda_{1}=\{\lambdaa_{\chi}\}_{\chi\in \hcG}$ and $\blambda_{2}=\{\lambdab_{\chi}\}_{\chi\in \hcG}$ respectively. Then, the induced CQ channel of the check factor $W_{1}\cnop W_{2}$ can be decomposed into a heralded mixture of PSCs  characterized using the ensemble $\{p^{\cnop}_{\chi},W^{\cnop}_{\chi}\}_{\chi\in \hcG}$ where  $p^{\cnop}_{\chi}$ is the probability of the abelian PSC $W^{\cnop}_{\chi}$. Let $\blambda^{\cnop}_{\chi}=\{\lambda^{\cnop,\chi}\}_{\chi'\in\hcG}$ be the eigen list for $W^{\cnop}_{\chi}$ for $\chi\in \hcG$. Then, for  $\chi\in\hcG$, $p_{\chi}^{\cnop}$ and $\blambda^{\cnop}_{\chi}$ are computed using $\blambda_1$ and $\blambda_2$ as below
    \begin{align*}
        p_{\chi}^{\cnop} & =\frac{1}{|\cG|^{2}}\sum_{\chi'\in \hcG}\lambda_{\chi\chi'}^{(1)}\lambda_{\chi'}^{(2)},\\
        \lambda_{\chi'}^{(\cnop,\chi)} & = \frac{1}{|\cG|p_{\chi}^{\cnop}}\lambda_{\chi\chi'}^{(1)}\lambda_{\chi'}^{(2)}.
    \end{align*}
    \end{lem}
\begin{proof}
               Consider the states $\bigg\{\ket{\psia_{g}}\bigg\}_{g\in\cG}$ that are the outputs of channel $W_1$ satisfying
        \begin{align*}
            \ket{\psia_{g}}\coloneqq \frac{1}{\sqrt{|\cG|}}\sum_{\chi\in \hcG}\sqrt{\lambdaa_{\chi}}\chi(g)\ket{\chi}.
        \end{align*}
  Similarly,  consider the states $\bigg\{\ket{\psib_{g}}\bigg\}_{g\in\cG}$ which are outputs of channel $W_2$ satisfying
        \begin{align*}
            \ket{\psib_{g}}\coloneqq \frac{1}{\sqrt{|\cG|}}\sum_{\chi\in \hcG}\sqrt{\lambdab_{\chi}}\chi(g)\ket{\chi}.
        \end{align*}

      Then, the outputs of check-factor-combined channel $W_{1}\cnop W_{2}$ can be written  $\forall g\in \cG$ as 
    \begin{align*}
       & [W_{1}\cnop W_{2}](g)\\
       & =\frac{1}{|\cG|}\sum_{g'\in \cG}\ketbra{\psia_{g'}}{\psia_{g'}}\otimes \ketbra{\psib_{g'^{-1}g}}{\psib_{g'^{-1}g}}
    \end{align*}
    Expanding $\ket{\psia_{g'}}\otimes \ket{\psib_{g'^{-1}g}}$ we get 
    \begin{align*}
      &  \ket{\psia_{g'}}\otimes \ket{\psib_{g'^{-1}g}}\\
      & =\frac{1}{|\cG|}\sum_{\chi\in \hcG}\sum_{\chi'\in \hcG}\sqrt{\lambdaa_{\chi}\lambdab_{\chi'}}\chi(g')\chi'(g'^{-1}g)\ket{\chi}\otimes\ket{\chi'} .
    \end{align*}
Consider the unitary $\tilde{U}^{\cnop}$ which acts on character vectors according to
\begin{align*}
    \tilde{U}^{\cnop}\ket{\chi}\otimes\ket{\chi'}=\ket{\chi\chi'^{-1}}\otimes\ket{\chi'}.
\end{align*}
Then we get
\begin{align*}
   &  \tilde{U}^{\cnop}\ket{\psia_{g'}}\otimes \ket{\psib_{g'^{-1}g}}\\
   & =  \frac{1}{|\cG|}\sum_{\chi\in \hcG}\sum_{\chi'\in \hcG}\sqrt{\lambdaa_{\chi}\lambdab_{\chi'}}\chi(g')\chi'(g'^{-1}g)\ket{\chi\chi'^{-1}}\otimes\ket{\chi'}.
\end{align*}
Hence for all $g\in \cG$, we can write
\begin{align*}
   & \tilde{U}^{\cnop}[W_{1}\cnop W_{2}](g)(\tilde{U}^{\cnop})^{\dagger}\\
   & = \frac{1}{|\cG|^3}\sum_{g'\in \cG}\left(\sum_{\chi\in \hcG}\sum_{\chi'\in \hcG}\sqrt{\lambdaa_{\chi}\lambdab_{\chi'}}\chi(g')\chi'(g'^{-1}g)\ket{\chi\chi'^{-1}}\otimes\ket{\chi'}\right) \left(\sum_{\xi\in \hcG}\sum_{\xi'\in \hcG}\sqrt{\lambdaa_{\xi}\lambdab_{\xi'}}\xi^{-1}(g')\xi'^{-1}(g'^{-1}g)\bra{\xi\xi'^{-1}}\otimes\bra{\xi'}\right)\\
   & = \frac{1}{|\cG|^3}\sum_{g'\in \cG}\bigg{(}\sum_{\chi,\chi',\xi,\xi'\in \hcG}\sqrt{\lambdaa_{\chi}\lambdaa_{\xi}\lambdab_{\chi'}\lambdab_{\xi'}}\chi\xi^{-1}(g')\chi'^{-1}\xi'(g')\chi'\xi'^{-1}(g)
  \ketbra{\chi\chi'^{-1}}{\xi\xi'^{-1}}\otimes \ketbra{\chi'}{\xi'}\bigg{)}\\
   & =\frac{1}{|\cG|^3}\sum_{\chi,\chi',\xi,\xi'\in \hcG}\sqrt{\lambdaa_{\chi}\lambdaa_{\xi}\lambdab_{\chi'}\lambdab_{\xi'}}\left(\sum_{g'\in \cG}\chi\xi^{-1}(g')\chi'^{-1}\xi'(g')\right) \chi'\xi'^{-1}(g)\ketbra{\chi\chi'^{-1}}{\xi\xi'^{-1}}\otimes \ketbra{\chi'}{\xi'}\\
   & =\frac{1}{|\cG|^3}\sum_{\chi,\chi',\xi,\xi'\in \hcG}\sqrt{\lambdaa_{\chi}\lambdaa_{\xi}\lambdab_{\chi'}\lambdab_{\xi'}}\left(|\cG|\delta_{\chi\xi^{-1},\chi'\xi'^{-1}}\right) \chi'\xi'^{-1}(g)\ketbra{\chi\chi'^{-1}}{\xi\xi'^{-1}}\otimes \ketbra{\chi'}{\xi'}\\
   & = \frac{1}{|\cG|^2}\sum_{\chi'',\chi',\xi'\in \hcG}\sqrt{\lambdaa_{\chi''\chi'}\lambdaa_{\chi''\xi'}\lambdab_{\chi'}\lambdab_{\xi'}}\chi'\xi'^{-1}(g)\ketbra{\chi''}{\chi''}\otimes  \ketbra{\chi'}{\xi'} \\
   & =\frac{1}{|\cG|^2}\sum_{\chi''\in \hcG}\ketbra{\chi''}{\chi''}\otimes  \left(\sum_{\chi',\xi'\in \hcG}\sqrt{\lambdaa_{\chi''\chi'}\lambdaa_{\chi''\xi'}\lambdab_{\chi'}\lambdab_{\xi'}}\chi'\xi'^{-1}(g)\ketbra{\chi'}{\xi'}\right).
\end{align*}
Hence we can compute $p^{\cnop}_{\chi}$ as 
\begin{align*}
    p^{\cnop}_{\chi} & =\Tr\left(\bra{\chi}\otimes \mI \left(\tilde{U}^{\cnop}[W_{1}\cnop W_{2}](g)(\tilde{U}^{\cnop})^{\dagger}\right)\ket{\chi}\otimes \mI\right)\\
    & = \Tr\left(\frac{1}{|\cG|^2}\sum_{\chi',\xi'\in \hcG}\sqrt{\lambdaa_{\chi\chi'}\lambdaa_{\chi\xi'}\lambdab_{\chi'}\lambdab_{\xi'}}\chi'\xi'^{-1}(g)\ketbra{\chi'}{\xi'}\right)\\
    & =\frac{1}{|\cG|^2}\sum_{\chi''\in \hcG}\bigg{(}  \sum_{\chi',\xi'\in \hcG}\sqrt{\lambdaa_{\chi\chi'}\lambdaa_{\chi\xi'}\lambdab_{\chi'}\lambdab_{\xi'}}\chi'\xi'^{-1}(g)\bra{\chi''}\ketbra{\chi'}{\xi'}\ket{\chi''}\bigg{)}\\
 & = \frac{1}{|\cG|^2}\sum_{\chi'\in \hcG}\lambdaa_{\chi\chi'}\lambdab_{\chi'}.
\end{align*}

Consider the state $\bigg\{\ket{\psi^{(\cnop,\chi)}_{g}}\bigg\}_{g\in \cG}$ such that 
\begin{align*}
    \ket{\psi^{(\cnop,\chi)}_{g}}\coloneqq\frac{1}{\sqrt{|\cG|}}\sum_{\xi\in \hcG}\sqrt{\lambda_{\xi}^{(\cnop,\chi)}}\xi(g)\ket{\xi},
\end{align*}
where $\lambda_{\chi'}^{(\cnop,\chi)}  = \frac{1}{|\cG|p_{\chi}^{\cnop}}\lambda_{\chi\chi'}^{(1)}\lambda_{\chi'}^{(2)}$. Then we get 
\begin{align}
   & \ketbra{\psi^{(\cnop,\chi)}_g}{\psi^{(\cnop,\chi)}_g}\\
   & =\frac{1}{|\cG|}\sum_{\chi',\xi'\in \hcG}\sqrt{\lambda_{\chi'}^{(\cnop,\chi)}\lambda_{\xi'}^{(\cnop,\chi)}}\chi'\xi'^{-1}(g)\ketbra{\chi'}{\xi'}\\
    & = \frac{1}{|\cG|^2p^{\cnop}_{\chi}}\sum_{\chi',\xi'\in \hcG}\sqrt{\lambdaa_{\chi\chi'}\lambdaa_{\chi\xi'}\lambdab_{\chi'}\lambdab_{\xi'}}\chi'\xi'^{-1}(g)\ketbra{\chi'}{\xi'}.
\end{align}
The following relation holds for the check node combined output
\begin{align*}
 &  \frac{ \bra{\chi}\otimes \mI \left(\tilde{U}^{\cnop}[W_{1}\cnop W_{2}](g)(\tilde{U}^{\cnop})^{\dagger}\right)\ket{\chi}\otimes \mI}{\Tr\left(\bra{\chi}\otimes \mI \left(\tilde{U}^{\cnop}[W_{1}\cnop W_{2}](g)(\tilde{U}^{\cnop})^{\dagger}\right)\ket{\chi}\otimes \mI\right)}\\
   & =\frac{ \bra{\chi}\otimes \mI \left(\tilde{U}^{\cnop}[W_{1}\cnop W_{2}](g)(\tilde{U}^{\cnop})^{\dagger}\right)\ket{\chi}\otimes \mI}{p^{\cnop}_{\chi}}\\
   & = \frac{1}{|\cG|^2p^{\cnop}_{\chi}}\sum_{\chi',\xi'\in \hcG}\sqrt{\lambdaa_{\chi\chi'}\lambdaa_{\chi\xi'}\lambdab_{\chi'}\lambdab_{\xi'}}\chi'\xi'^{-1}(g)\ketbra{\chi'}{\xi'}\\
   & = \ketbra{\psi^{(\cnop,\chi)}_g}{\psi^{(\cnop,\chi)}_g}.\qedhere
\end{align*}
\end{proof}

Therefore, the check-factor update maps two incoming $\cG$-covariant PSCs to a finite heralded mixture of $\cG$-covariant PSCs. Thus, the outgoing message is completely described by the ensemble $\{p_{\chi}^{\cnop},\blambda^{(\cnop,\chi)}\}_{\chi\in \widehat{\cG}}$.

% The natural question to ask is what the operational unitary is that converts the induced CQ channel output to a heralded mixture of PSCs. This is achieved by the check factor unitary $U^{\cnop}$ acting on two group covariant PSC output states. More specifically, the check factor unitary satisfies the following relation 
% \begin{align*}
%     U^{\cnop}  = (\mI\otimes \cF^{\dagger})(\text{SWAP})\tilde{U}^{\cnop}
%    \end{align*}
%    where unitary $\cF$ transforms the state $\ket{g}$ for a fixed $g\in\cG$ to the character state $\ket{\chi}$ for fixed $\chi\in\hcG$. Notice that the check factor unitary $U^{\cnop}$ does not depend on the eigen lists of the input channels.

It is also useful to describe the operational unitary underlying the above decomposition. The map $\tilde{U}^{\cnop}$ reorganizes the two incoming character labels into a herald label and a residual character label. After this change of basis, the induced output becomes block diagonal in the herald register, and each block is a $\cG$-covariant PSC. An operational realization is given by
\begin{align*}
    U^{\cnop} = (\mI\otimes \cF^{\dagger})(\mathrm{SWAP})\tilde{U}^{\cnop},
\end{align*}
where $\cF$ denotes the Fourier transform between the canonical basis indexed by $\cG$ and the character basis indexed by $\widehat{\cG}$. Notice that $U^{\cnop}$ depends only on the group structure and not on the eigen lists of the incoming channels.
In \ref{appendix:check factor example}, we provide an example for check factor update rule for group $\cG=\mathbb{Z}_3\times\mathbb{Z}_2$.

\subsection{Quantum Message Passing for Equality Factor}
% In Lemma~\ref{lem:bit node abelian}, we establish the characterization of the induced CQ channel of the equality factor with two input edges where each edge corresponds to a $\cG$-covariant PSC. This characterization is achieved using the eigen list of the PSCs.  

We next consider the equality factor with two incoming $\cG$-covariant PSCs. In contrast to the check factor, the equality update remains within the class of single $\cG$-covariant PSCs after isometric compression. Thus, the outgoing message is again characterized by one eigen list rather than a nontrivial heralded ensemble.
\begin{lem}\label{lem:bit node abelian}
Consider $\cG$-covariant PSCs $W_{1}$ and $W_{2}$ for abelian group $\cG$ with Gram matrices $\Gamma_{W_1}$ and $\Gamma_{W_2}$ with eigen lists $\blambda_{1}=\{\lambdaa_{\chi}\}_{\chi\in \hcG}$ and $\blambda_{2}=\{\lambdab_{\chi}\}_{\chi\in \hcG}$ respectively.
Then, the induced channel of the equality factor $W_{1}\vnop W_{2}$ is isometrically equivalent to  a $\cG$-covariant PSC $W^{\vnop}$ with Gram matrix eigen list $\blambda^{\vnop}=\{\lambda^{\vnop}_{\chi}\}_{\chi\in \hcG}$ such that 
        \begin{align*}
            \lambda^{\vnop}_{\chi}=\frac{1}{|\cG|}\sum_{\chi'\in \hcG}\lambdaa_{\chi'}\lambdab_{\chi\chi'^{-1}}.
        \end{align*}
 \end{lem}  
    \begin{proof}
Consider the states $\bigg\{\ket{\psia_{g}}\bigg\}_{g\in\cG}$ and $\bigg\{\ket{\psib_{g}}\bigg\}_{g\in\cG}$  which are outputs of channel $W_1$ and $W_{2}$ respectively. 
Then the Gram matrix entries for $\Gamma_{W_1}$ and $\Gamma_{W_2}$ $\forall g\in \cG$, satisfy
\begin{align*}
    \gamma^{(1)}_{g} & =\frac{1}{|\cG|}\sum_{\chi\in \hcG}\lambdaa_{\chi}\chi(g),\\
    \gamma^{(2)}_{g} & =\frac{1}{|\cG|}\sum_{\chi\in \hcG}\lambdab_{\chi}\chi(g).
\end{align*}

For the equality-factor combined channel $W_{1}\vnop W_{2}$, for all inputs $g\in \cG$, the outputs are 
    \begin{align*}
        [W_{1}\vnop W_{2}](g)= \ketbra{\psia_{g}}{\psia_g}\otimes \ketbra{\psib_g}{\psib_g}.
    \end{align*}
    Thus, the Gram matrix $\Gamma_{W_1\vnop W_2}$ for the combined channel $W_{1}\vnop W_{2}$ satisfies 
    \begin{align*}
        (\Gamma_{W_1\vnop W_2})_{g,g'} & =\braket{\psia_g}{\psia_{g'}}\braket{\psib_g}{\psib_{g'}}\\
        & = \gamma^{(1)}_{g^{-1}g'}\gamma^{(2)}_{g^{-1}g'}\\
        & =\left(\frac{1}{|\cG|}\sum_{\xi\in \hcG}\lambdaa_{\xi}\xi(g^{-1}g')\right)\left(\frac{1}{|\cG|}\sum_{\xi'\in \hcG}\lambdab_{\xi'}\xi'(g^{-1}g')\right)\\
        & =\frac{1}{|\cG|^2}\sum_{\xi,\xi'\in \hcG}\lambdaa_{\xi}\lambdab_{\xi'}\xi\xi'(g^{-1}g')\\\
        & =\frac{1}{|\cG|^{2}}\sum_{\chi\in \hcG} \left(\sum_{\xi\in \hcG}\lambdaa_{\xi}\lambdab_{\chi\xi^{-1}}\right)\chi(g^{-1}g').
    \end{align*}
    Denoting $\lambda^{\vnop}_{\chi}=\frac{1}{|\cG|}\sum_{\xi\in \hcG}\lambdaa_{\xi}\lambdab_{\chi\xi^{-1}}$, we get 
    \begin{align*}
       (\Gamma_{W_1\vnop W_2})_{g,g'}= \frac{1}{|\cG|}\sum_{\chi\in \hcG}\lambda^{\vnop}_{\chi}\chi(g^{-1}g').
    \end{align*}
    Thus $\Gamma_{W_1\vnop W_2}$ is $\cG$-circulant with eigen list $\blambda^{\vnop}=\{\lambda_{\chi}^{\vnop}\}_{\chi\in \hcG}$ which characterizes channel $W^{\vnop}$.
    \end{proof}

    Hence, the equality-factor update preserves the class of $\cG$-covariant PSCs. Equivalently, it may be viewed as a degenerate heralded mixture with a single herald value. 
    
% Since, in this case the output is a $\cG$-covariant PSC characterized via single eigen list, we often use the notation $\blambda_1\vnop \blambda_2$ to denote the resultant eigen list $\blambda^{\vnop}$.
% To describe the operational unitary for the equality factor, consider the unitary $U^{\vnop}_{\blambda_1,\blambda_2}$ acting on  the output states for two $\cG$-covariant  PSCs with eigen lists $\blambda_1$ and $\blambda_2$.
% The equality factor unitary $U^{\vnop}_{\blambda_1,\blambda_2}$ satisfies the following relation $\forall g\in \cG$ and some fixed $b\in\cG$
% \begin{align*}
% U^{\vnop}_{\blambda_1,\blambda_2}\ket{\psi^{(1)}_{g}}\otimes \ket{\psi^{(2)}_{g}}=\ket{\psi^{\vnop}_{g}}\otimes \ket{b}.
% \end{align*}

Since the output of the equality factor is again a single $\cG$-covariant PSC, we often write $\blambda_1\vnop \blambda_2$ for the resulting eigen list $\blambda^{\vnop}$. It is also useful to describe an operational unitary for this combination. The equality-factor unitary $U^{\vnop}_{\blambda_1,\blambda_2}$ compresses the two aligned input states into one outgoing PSC together with an ancilla state independent of the channel input. Formally, for some fixed $\eta\in \hcG$ and all $g\in \cG$, we have
\begin{align*}
    U^{\vnop}_{\blambda_1,\blambda_2}\ket{\psi^{(1)}_{g}}\otimes \ket{\psi^{(2)}_{g}}
    =
    \ket{\psi^{\vnop}_{g}}\otimes \ket{\eta}.
\end{align*}
Thus, the equality factor does not create a nontrivial herald register; rather, it performs an isometric compression of the two incoming PSC outputs into a single $\cG$-covariant PSC.
To construct the unitary $U^{\vnop}_{\blambda_1,\blambda_2}$, we let
\begin{align*}
     U^{\vnop}_{\blambda_1,\blambda_2} =U^{\text{control}}_{\blambda_1,\blambda_2}U_{+},
\end{align*}
where $U_{+}$ satisfies 
\begin{align*}
    U_{+}\ket{\chi}\otimes \ket{\chi'}=\ket{\chi\chi'}\otimes \ket{\chi'}
\end{align*}
for all $\chi,\chi'\in \hcG$, and the unitary $U^{\text{control}}_{\blambda_1,\blambda_2}$ is given by
\begin{align*}
    U^{\text{control}}_{\blambda_1,\blambda_2}=\sum_{\chi\in \hcG} \ketbra{\chi}{\chi}\otimes U^{\chi}_{\blambda_1,\blambda_2} ,
\end{align*}
such that the unitary $U^{\chi}_{\blambda_1,\blambda_2}$ has the following form
\begin{align*}
    U^{\chi}_{\blambda_1,\blambda_2}=\begin{cases}
        \mI, \qquad\text{if $\ket{\zeta_{\chi}}=\ket{\eta}$}\\
        \mI -2\ketbra{\zeta_{\chi}'}{\zeta_{\chi}'},\qquad\text{otherwise}
    \end{cases}
\end{align*}
with $\ket{\zeta_{\chi}'}=\frac{\ket{\zeta_{\chi}}-\ket{\eta}}{||\ket{\zeta_{\chi}}-\ket{\eta}||}$,  $\ket{\zeta_{\chi}}=\frac{\ket{\tilde{\zeta}_{\chi}}}{||\ket{\tilde{\zeta}_{\chi}}||}$, and $\ket{\tilde{\zeta}_{\chi}}=\sum_{\xi\in\hcG}\sqrt{\lambdaa_{\chi\xi^{-1}}\lambdab_{\xi}}\ket{\xi}$.

In \ref{appendix:equality factor example}, we provide an example for equality factor update rule for group $\cG=\mathbb{Z}_3\times\mathbb{Z}_2$.
\subsection{Quantum Message Passing for Homomorphism Factor}

% In this section, establish the quantum message passing for homomorphism factors.
% Consider the surjective homomorphism $\phi: \cG_1\rightarrow \cG_2$ where $\cG_1$ and $\cG_2$ are abelian groups. Consider the dual map $\hat{\phi}:\hcG_{2}\rightarrow \hcG_1$ defined by the composite operation $\hat{\phi}(\xi)=\hat{\phi}\xi$ $\forall \xi\in \hcG_2$. Observe that the following equality holds $\forall g\in \cG_1$
% \begin{align*}
%     \hat{\phi}\xi(g)=\xi(\phi(g)).
% \end{align*}
% The image of map $\hat{\phi}$ satisfies 
% \begin{align*}
%     \image\hat{\phi}\coloneqq \{\chi\in \hcG_1:\chi(k)=1, \forall k\in \ker{\phi}\}
% \end{align*}
% For each $h\in \cG_2$, the output of induced CQ channel $\rho_h$ $\forall h\in \cG_2$ is
% \begin{align*}
%     \rho_h= \frac{1}{|\ker{\phi}|}\sum_{g\in \cG_1: \phi(g)=h}\ketbra{\psi_g}{\psi_g}
% \end{align*}
% When $\phi$ is surjective, we have $|\cG_1|=|\ker\phi||\cG_2|$. Also $\hat{\phi}$ is injective, so $|\image\hat{\phi}|=|\hcG_2|=|\cG_2|$. Thus, we get 
% \begin{align*}
%     \frac{|\hcG_1|}{|\image\hat{\phi}|}=\frac{|\cG_1|}{|\cG_2|}=|\ker\phi|
% \end{align*}
% Next, we define a set of coset representatives  of characters $\cT$ such that $\hcG_1$ can be written as a disjoint union as below 
% \begin{align}\label{eq:group disjoint union}
%     \hcG_1=\bigsqcup_{\eta\in \cT}\eta\image\hat{\phi}
% \end{align}

We now consider the quantum message passing rule associated with a surjective homomorphism factor. Let $\phi:\cG_1\rightarrow \cG_2$ be a surjective homomorphism between finite abelian groups. On the dual side, $\phi$ induces the injective map $\hat{\phi}:\widehat{\cG}_2\rightarrow \widehat{\cG}_1$ defined by
\begin{align*}
    \hat{\phi}(\xi)(g)=\xi(\phi(g)),
    \qquad \xi\in \widehat{\cG}_2,\ g\in \cG_1.
\end{align*}
The image of $\hat{\phi}$ is precisely the subgroup of characters on $\cG_1$ that are trivial on $\ker \phi$, namely
\begin{align*}
    \image \hat{\phi}
    \coloneqq
    \{\chi\in \widehat{\cG}_1:\chi(k)=1,\ \forall k\in \ker \phi\}.
\end{align*}
Since $\phi$ is surjective, $\hat{\phi}$ is injective and $|\cG_1|=|\ker \phi|\,|\cG_2|$. Hence, we have
\begin{align*}
    \frac{|\widehat{\cG}_1|}{|\image \hat{\phi}|}
    =
    \frac{|\cG_1|}{|\cG_2|}
    =
    |\ker \phi|.
\end{align*}
Choose a set of coset representatives $\cT$ such that dual group $\hcG_1$ is expressed using following disjoint union
\begin{align}\label{eq:group disjoint union}
    \widehat{\cG}_1=\bigsqcup_{\eta\in \cT}\eta\,\image \hat{\phi}.
\end{align}
This coset decomposition is the source of the herald register in the homomorphism update where each coset label $\eta$ will index one outgoing $\cG_2$-covariant PSC.
Now we are ready to state the general lemma about surjective homomorphism $\phi$. 
\begin{lem}\label{lem:surjective homomorphism}
   % Consider $\cG_1$-covariant PSC $W_{1}$ with Gram matrix $\Gamma_{W_1}$ characterized by the eigen list $\blambda_{1}=\{\lambdaa_{\chi}\}_{\chi\in \hcG_1}$. Let  $\phi:\cG_1\rightarrow \cG_2$ be a surjective homomorphism. Then, the output is characterized via a heralded mixture of $\cG_2$-covariant PSCs. Specifically, the ensemble $\{p_{\eta},\blambda^{\eta}\}_{\eta\in\cT}$ characterizes the output channel, where $p_{\eta}$ corresponds to the probability of $\cG_2$-covariant PSC satisfying
   % \begin{align*}
   %     p_{\eta} & = \frac{1}{|\cG_1|}\sum_{\chi\in \eta\image\hat{\phi}}\lambda_{\chi}^{(1)}
   % \end{align*}
   % with eigen list $\blambda^{\eta}\coloneqq\{\lambda^{\eta}_{\xi}\}_{\xi\in \hcG_2}$  satisfying 
   % \begin{align*}
   %     % p_{\eta} & = \frac{1}{|\cG_1|}\sum_{\chi\in \eta\image\hat{\phi}}\lambda_{\chi}^{(1)}\\
   %     \lambda_{\xi}^{\eta}  = \frac{|\cG_2|}{|\cG_1|p_{\eta}}\lambda^{(1)}_{\eta\hat{\phi}(\xi)}
   % \end{align*}
   % $\forall \eta\in \cT$, $\xi\in \hcG_2$.

   Consider a $\cG_1$-covariant PSC $W_1$ with Gram matrix $\Gamma_{W_1}$ and eigen list $\blambda_1=\{\lambda^{(1)}_{\chi}\}_{\chi\in \hcG_1}$. Let $\phi:\cG_1\rightarrow \cG_2$ be a surjective homomorphism, and let $\cT\subseteq \hcG_1$ be a set of representatives for the cosets of $\image \hat{\phi}$ in $\hcG_1$. Then, the output channel is characterized by a heralded mixture of $\cG_2$-covariant PSCs. More specifically, the output channel is characterized by the ensemble $\{p_{\eta},\blambda^{\eta}\}_{\eta\in\cT}$, where $p_{\eta}$ is given by
\begin{align*}
    p_{\eta} = \frac{1}{|\cG_1|}\sum_{\chi\in \eta\,\image\hat{\phi}}\lambda_{\chi}^{(1)}
\end{align*}
and $\blambda^{\eta}=\{\lambda^{\eta}_{\xi}\}_{\xi\in \hcG_2}$ is the eigen list of the corresponding $\cG_2$-covariant PSC, where
\begin{align*}
    \lambda_{\xi}^{\eta} = \frac{|\cG_2|}{|\cG_1|\,p_{\eta}}\lambda^{(1)}_{\eta\hat{\phi}(\xi)}
\end{align*}
for all $\eta\in \cT$ and $\xi\in \hcG_2$.
\end{lem}
\begin{proof}
    Recall the canonical state for channel $W_1$ is
    \begin{align*}
        \ket{\psia_g}\coloneqq\frac{1}{\sqrt{|\cG_1|}}\sum_{\chi\in \hcG_1}\sqrt{\lambdaa_{\chi}}\chi(g)\ket{\chi}.
    \end{align*}
Then, from the disjoint union in \eqref{eq:group disjoint union}, for every character $\chi\in \hcG_1$, there exists a unique pair $(\eta, \xi)$ with $\eta\in \cT$ and $\xi \in \hcG_2$ such that 
\begin{align*}
    \chi=\eta\hat{\phi}(\xi),
\end{align*}
where we use the injectivity of $\hat{\phi}$.
Next, we define the operator $V$, which acts on the state $\ket{\chi}$, using
\begin{align*}
    V\ket{\eta\hat{\phi}(\xi)}=\ket{\xi}\otimes \ket{\eta}.
\end{align*}
Equivalently, we can write 
\begin{align*}
    V=\sum_{\eta\in \cT}\sum_{\xi\in \hcG_2}\ket{\xi}\otimes\ket{\eta}\bra{\eta\hat{\phi}(\xi)}.
\end{align*}
Thus, we compute 
\begin{align*}
    V^{\dagger}V & =\sum_{\eta\in \cT}\sum_{\xi\in \hcG_2}\sum_{\eta'\in \cT}\sum_{\xi'\in \hcG_2}\ket{\eta\hat{\phi}(\xi)}\bra{\xi}\otimes \bra{\eta}\ket{\xi'}\otimes\ket{\eta'}\bra{\eta'\hat{\phi}(\xi')}\\
    & = \sum_{\eta\in \cT}\sum_{\xi\in \hcG_2}\ketbra{\eta\hat{\phi}(\xi)}{\eta\hat{\phi}(\xi)}\\
    & =\mI.
\end{align*}
This implies  that $V$ is an isometry. Applying $V$ on the state $\ket{\psia_g}$ with $\phi(g)=h$, we get 
\begin{align*}
    V\ket{\psia_g} & =\frac{1}{\sqrt{|\cG_1|}}\sum_{\eta\in \cT}\sum_{\xi\in \hcG_2}\sqrt{\lambdaa_{\eta\hat{\phi}(\xi)}}\eta\hat{\phi}(\xi)(g)\ket{\xi}\otimes \ket{\eta}\\
    & = \frac{1}{\sqrt{|\cG_1|}}\sum_{\eta\in \cT}\sum_{\xi\in \hcG_2}\sqrt{\lambdaa_{\eta\hat{\phi}(\xi)}}\eta(g)\hat{\phi}(\xi)(g)\ket{\xi}\otimes \ket{\eta}\\
    & = \frac{1}{\sqrt{|\cG_1|}}\sum_{\eta\in \cT}\sum_{\xi\in \hcG_2}\sqrt{\lambdaa_{\eta\hat{\phi}(\xi)}}\eta(g)\xi(\phi(g))\ket{\xi}\otimes \ket{\eta}\\
    & = \frac{1}{\sqrt{|\cG_1|}}\sum_{\eta\in \cT}\sum_{\xi\in \hcG_2}\sqrt{\lambdaa_{\eta\hat{\phi}(\xi)}}\eta(g)\xi(h)\ket{\xi}\otimes \ket{\eta}.
\end{align*}
On the other hand, for $h\in \cG_2$, $\exists$ $g_0\in \cG_1$ such that $g_0\ker\phi\in \cG/\ker\phi$ and $\phi(g)=h$ $\forall g\in g_0\ker\phi$. Hence, we can rewrite the state $\rho_h$ as 
\begin{align*}
    \rho_h=\frac{1}{|\ker\phi|}\sum_{k\in \ker\phi}\ketbra{\psia_{g_0k}}{\psia_{g_0k}}.
\end{align*}
Applying isometry $V$ to state $\rho_h$, we get 
\begin{align*}
    V\rho_hV^{\dagger} & = \frac{1}{|\ker\phi|}\sum_{k\in \ker\phi}V\ketbra{\psia_{g_0k}}{\psia_{g_0k}}V^{\dagger}\\
    & = \frac{1}{|\ker\phi|}\sum_{k\in \ker\phi}\frac{1}{|\cG_1|}\sum_{\eta,\eta'\in \cT,\xi,\xi'\in \hcG_2}\sqrt{\lambdaa_{\eta\hat{\phi}(\xi)}\lambdaa_{\eta'\hat{\phi}(\xi')}}\eta\eta'^{-1}(g_0)\eta\eta'^{-1}(k)\xi\xi'^{-1}(h)\ketbra{\xi}{\xi'}\otimes\ketbra{\eta}{\eta'}.
\end{align*}
Using the fact that $\sum_{k\in \ker\phi}\eta(\eta')^{-1}(k)=|\ker\phi|\delta_{\eta,\eta'}$, $\forall \eta,\eta'\in \cT$, we get 
\begin{align*}
     V\rho_hV^{\dagger} & =\frac{1}{|\cG_1|}\sum_{\eta\in \cT,\xi,\xi'\in \hcG_2}\sqrt{\lambdaa_{\eta\hat{\phi}(\xi)}\lambdaa_{\eta\hat{\phi}(\xi')}}\xi\xi'^{-1}(h)\ketbra{\xi}{\xi'}\otimes\ketbra{\eta}{\eta}.
\end{align*}
Hence, we compute the probability $p_{\eta}$ as follows 
\begin{align*}
    p_{\eta} & = \Tr\left(\mI\otimes \bra{\eta}V\rho_hV^{\dagger}\mI\otimes \ket{\eta}\right)\\
    & = \Tr\left(\frac{1}{|\cG_1|}\sum_{\xi,\xi'\in \hcG_2}\sqrt{\lambdaa_{\eta\hat{\phi}(\xi)}\lambdaa_{\eta\hat{\phi}(\xi')}}\xi\xi'^{-1}(h)\ketbra{\xi}{\xi'}\right)\\
    & = \frac{1}{|\cG_1|}\sum_{\xi\in \hcG_2}\lambdaa_{\eta\hat{\phi}(\xi)}.
\end{align*}
Define $\cG_2$-covariant state $\ket{\psi_{h}^{\eta}}$ as below 
\begin{align*}
    \ket{\psi_{h}^{\eta}}\coloneqq \frac{1}{\sqrt{|\cG_2|}}\sum_{\xi\in \hcG_2}\sqrt{\lambda^{\eta}_{\xi}}\xi(h)\ket{\xi},
\end{align*}
where $\lambda^{\eta}_{\xi}= \frac{|\cG_2|}{|\cG_1|p_{\eta}}\lambda^{(1)}_{\eta\hat{\phi}(\xi)}$. Then, we immediately get 
\begin{align*}
    \ketbra{\psi_{h}^{\eta}}{\psi_{h}^{\eta}} & = \frac{1}{|\cG_1|p_{\eta}}\sum_{\xi,\xi'\in \hcG_2}\sqrt{\lambdaa_{\eta\hat{\phi}(\xi)}\lambdaa_{\eta\hat{\phi}(\xi')}}\xi\xi'^{-1}(h)\ketbra{\xi}{\xi'}.
\end{align*}
Thus, for all $h\in \cG_2$, we have 
% $V\rho_hV^{\dagger}$ satisfies $\forall h\in\cG_2$
\begin{align*}
    V\rho_hV^{\dagger}= \sum_{\eta\in\cT}p_{\eta}\ketbra{\psi_{h}^{\eta}}{\psi_{h}^{\eta}}\otimes \ketbra{\eta}{\eta}.
\end{align*}
\end{proof}

Therefore, the homomorphism-factor update maps an incoming $\cG_1$-covariant PSC to a finite heralded mixture of $\cG_2$-covariant PSCs, with herald register indexed by the coset representatives $\eta\in \cT$. The outgoing message is thus completely described by the ensemble $\{p_{\eta},\blambda^{\eta}\}_{\eta\in \cT}$. 

The isometry $V$ also has a direct operational interpretation. It coherently separates each input character $\chi\in \widehat{\cG}_1$ into its coset representative $\eta\in \cT$ and its pulled-back output character $\hat{\phi}(\xi)\in \image \hat{\phi}$. After averaging over each fiber of $\phi$, the output becomes block diagonal in the $\eta$ register. Thus, $\eta$ serves as the classical herald, and each block is a $\cG_2$-covariant PSC.

% The Lemma~\ref{lem:surjective homomorphism}, even generalized to homomorphism $\phi$ which is not surjective. In this case, we consider the subgroup $\image\phi\le \cG_2$ and analyze output states $\rho_h$ $\forall h\in \image\phi$ while the overall derivation remain the same. Since $\image\phi$ is a subgroup, this also implies that it is sufficient to consider surjective homomorphisms without loss of the generality.

Lemma~\ref{lem:surjective homomorphism} extends immediately to an arbitrary homomorphism $\phi$ that is not necessarily surjective by replacing $\cG_2$ with the subgroup $\image \phi\le \cG_2$. Since $\image \phi$ is again a finite abelian group, it is sufficient to analyze the surjective case in what follows. In \ref{appendix:homomorphism factor example}, we provide a few examples for homomorphism factor update rules for several groups and homomorphisms.

Next we state a special case of Lemma~\ref{lem:surjective homomorphism}. In Lemma~\ref{lem:surjective homomorphism special}, we identify a support condition on the eigen list $\{\lambda_{\chi}\}_{\chi\in \hcG_1}$ under which a $\cG_1$-covariant input PSC induces a single PSC through the homomorphism $\phi$ rather than a heralded mixture.

\begin{lem}\label{lem:surjective homomorphism special}
    Consider $\cG_1$-covariant PSC $W_{1}$ with Gram matrix $\Gamma_{W_1}$ characterized by the eigen list $\blambda_{1}=\{\lambdaa_{\chi}\}_{\chi\in \hcG_1}$. Let  $\phi:\cG_1\rightarrow \cG_2$ be a surjective homomorphism. Also assume that  $\lambdaa_{\chi}=0$, $\forall \chi\notin \image(\hat{\phi})$. Then the induced channel $W_2$ defined by 
    \begin{align*}
        \ket{\psib_{h}}\coloneqq \ket{\psia_{g}},\quad \text{$\forall g\in \cG_1$ such that $\phi(g)=h$}
    \end{align*}
    is a $\cG_2$-covariant PSC. Moreover the channel $W_2$ is characterized by Gram matrix $\Gamma_{W_2}$ and eigen list $\blambda_{2}=\{\lambdab_{\xi}\}_{\xi\in \hcG_2}$, where 
    \begin{align*}
        \lambdab_{\xi}=\frac{|\cG_2|}{|\cG_1|}\lambdaa_{\hat{\phi}(\xi)}
    \end{align*}
    for all $\xi\in \hcG_2$ and $\hat{\phi}:\hcG_2\rightarrow\hcG_1$ is the dual map of $\phi$.
\end{lem}
\begin{proof}
    % Recall the canonical state for channel $W_1$
    % \begin{align*}
    %     \ket{\psia_a}\coloneqq\frac{1}{\sqrt{|\cG_1|}}\sum_{\chi\in \hcG_1}\sqrt{\lambdaa_{\chi}}\chi(a)\ket{\chi}
    % \end{align*}
    First, observe that since $\lambdaa_{\chi}=0$  $\forall \chi\notin \image(\hat{\phi})$, we get 
    \begin{align*}
        \ket{\psia_{kg}}=\ket{\psia_g}, \forall k\in \ker{\phi}.
    \end{align*}
    Thus, the state $\ket{\psib_h}$ is well defined.
    Define the operator $V_{\phi}$ based on the map $\phi$ as follows,
    \begin{align*}
        V_{\phi}=\sum_{\xi\in \hcG_2}\ketbra{\xi}{\hat{\phi}(\xi)}.
    \end{align*}
    Computing $V_{\phi}^{\dagger}V_{\phi}$ we get 
    \begin{align*}
        V_{\phi}^{\dagger}V_{\phi}=\sum_{\xi\in \hcG_2}\ketbra{\hat{\phi}(\xi)}{\hat{\phi}(\xi)}=\Pi_{\image(\hat{\phi})},
    \end{align*}
    where $\Pi_{\image(\hat{\phi})}$ is the projection on the image of $\hat{\phi}$. From the support assumption ($\lambdaa_{\chi}=0$  $\forall \chi\notin \image(\hat{\phi})$), we get 
    \begin{align*}
\Pi_{\image(\hat{\phi})}\ket{\psia_g}=\ket{\psia_g}, \forall g\in \cG_1.
    \end{align*}
Thus, $V_{\phi}$ is a valid isometry for the subspace spanned by the states $\bigg\{\ket{\psia_g}\bigg\}_{g\in \cG_1}$.
    Applying $V_{\phi}$ on the state $\ket{\psia_g}$, we get 
    \begin{align*}
        V_{\phi}\ket{\psia_g} & = \frac{1}{\sqrt{|\cG_1|}}\sum_{\chi\in \hcG_1}\sqrt{\lambdaa_{\chi}}\chi(g)\sum_{\xi\in\hcG_2}\ketbra{\xi}{\hat{\phi}\xi}\ket{\chi}\\
        & \stackrel{(1)}{=} \frac{1}{\sqrt{|\cG_1|}}\sum_{\xi\in \hcG_2}\sqrt{\lambdaa_{\hat{\phi}\xi}}\hat{\phi}\xi(g)\ket{\xi}\\
        & = \frac{1}{\sqrt{|\cG_1|}}\sum_{\xi\in \hcG_2}\sqrt{\lambdaa_{\hat{\phi}\xi}}\xi(\phi(g))\ket{\xi},
    \end{align*}
    where $(1)$ follows from the fact that $\lambdaa_{\chi}=0$ $\forall\chi\notin \image(\hat{\phi})$, so the only surviving terms correspond to $\chi=\hat{\phi}(\xi)$.
Finally, setting $\lambdab_{\xi}=\frac{|\cG_2|}{|\cG_1|}\lambdaa_{\hat{\phi}\xi}$ and defining the canonical state for $W_2$ to be
\begin{align*}
    \ket{\psib_h}=\frac{1}{\sqrt{|\cG_2|}}\sum_{\xi\in \hcG_2}\sqrt{\lambdab_{\xi}}\xi(h)\ket{\xi}
\end{align*}
for all $h\in\cG_2$, we get 
\begin{align*}
     V_{\phi}\ket{\psia_g} & =\frac{1}{\sqrt{|\cG_1|}}\sum_{\xi\in \hcG_2}\sqrt{\frac{|\cG_1|}{|\cG_2|}\lambdab_{\xi}}\xi(\phi(g))\ket{\xi}\\
     & =\ket{\psib_{\phi(g)}}.\qedhere
\end{align*}
\end{proof}
In this special support regime, the herald becomes trivial and the homomorphism update reduces to a single $\cG_2$-covariant PSC. Thus, Lemma~\ref{lem:surjective homomorphism special} identifies precisely when the general heralded decomposition of Lemma~\ref{lem:surjective homomorphism} collapses to one PSC.

\subsection{Quantum Message Passing for Marginalization Factor}

We next consider marginalization over one component of a product alphabet. In the character domain, averaging over the marginalized variable enforces equality of the corresponding dual indices by character orthogonality. As a result, the discarded dual coordinate becomes a classical herald, while the retained coordinate labels the outgoing PSC on the remaining group.
\begin{lem}\label{lem:marginalization factor}
  Let $\cG_1$ and $\cG_2$ be abelian groups such that $\cU=\cG_1\times \cG_2$.   Consider a $\cU$-covariant PSC $W\colon (g_1,g_2)\rightarrow \ketbra{\psi_{(g_1,g_2)}}{\psi_{(g_1,g_2)}}$ for $(g_1,g_2)\in \cU$ characterized via the eigen list $\blambda=\{\lambda_{(\chi,\eta)}\}_{(\chi,\eta)\in \hcG_1\times \hcG_2}$. Let $\phi\colon \cU\rightarrow \cG_1$ be a marginalization map such that $\phi(g_1,g_2)=g_1$ $\forall (g_1,g_2)\in \cU$. Then the induced channel $W_{\phi}$ is isometrically equivalent to a heralded mixture of $\cG_1$-covariant PSCs characterized via the ensemble $\{p_{\eta},\blambda^{\eta}\}_{\eta\in \hcG_2}$ where the probability $p_{\eta}$ is
  \begin{align*}
      p_{\eta}=\frac{1}{|\cU|}\sum_{\chi\in \hcG_1}\lambda_{(\chi,\eta)}
  \end{align*}
  and the elements of eigen list $\blambda^{\eta}=\{\lambda^{\eta}_{\chi}\}_{\chi\in \hcG_1}$ satisfy the following relation 
  \begin{align*}
      \lambda^{\eta}_{\chi}=\frac{1}{|\cG_2|p_{\eta}}\lambda_{(\chi,\eta)}.
  \end{align*}
\end{lem}
\begin{proof}
For a group element $(g_1,g_2)\in \cU$, we can write the canonical form for state $\ket{\psi_{(g_1,g_2)}}$ as follows 
\begin{align*}
     \ket{\psi_{(g_1,g_2)}}=\frac{1}{\sqrt{|\cU|}}\sum_{\chi\in \hcG_{1},\eta\in \hcG_2}\sqrt{\lambda_{(\chi,\eta)}}\chi(g_1)\eta(g_2)\ket{(\chi,\eta)}.
\end{align*}
The marginalization map $\phi_{\text{marginal}}$ induces the CQ channel $W_{\phi_{\text{marginal}}}\colon g_1\rightarrow \rho_{g_1}$ $\forall g_1\in \cG_1$ such that 
\begin{align*}
     \rho_{g_1}= \frac{1}{|\cG_2|}\sum_{g_2\in \cG_2}\ketbra{\psi_{(g_1,g_2)}}{\psi_{(g_1,g_2)}}.
\end{align*}
Expanding the state $\ket{\psi_{(g_1,g_2)}}$, we get 
\begin{align*}
    \rho_{g_1} & =\frac{1}{|\cG_2|}\sum_{g_2\in \cG_2}\frac{1}{|\cU|}\sum_{\chi\in \hcG_{1},\eta\in \hcG_2}\sum_{\chi'\in \hcG_{1},\eta'\in \hcG_2}\sqrt{\lambda_{(\chi,\eta)}}\sqrt{\lambda_{(\chi',\eta')}}\chi(g_1)\chi'^{-1}(g_1)\eta(g_2)\eta'^{-1}(g_2)\ketbra{(\chi,\eta)}{(\chi',\eta')}.
\end{align*}
Using the character orthogonality, we get 
\begin{align*}
    \frac{1}{|\cG_2|}\sum_{g_2\in \cG_2}\eta(g_2)\eta'^{-1}(g_2)=\delta_{\eta,\eta'}.
\end{align*}
Thus, we can rewrite the state $\rho_{g_1}$ for $g_1\in \cG_{1}$ as 
\begin{align*}
    \rho_{g_1}= \frac{1}{|\cU|}\sum_{\eta\in \hcG_2}\sum_{\chi,\chi'\in \hcG_{1}}\sqrt{\lambda_{(\chi,\eta)}\lambda_{(\chi',\eta)}}\chi(g_1)\chi'^{-1}(g_1)\ketbra{(\chi,\eta)}{(\chi',\eta)}.
\end{align*}
Next, consider the unitary $U_{\text{split}}$ defined by
\begin{align}
    U_{\text{split}}\ket{(\chi,\eta)}=\ket{\chi}\ket{\eta}
\end{align}
for all $(\chi,\eta)\in \hcG_1\times \hcG_2$. Then applying $U_{\text{split}}$ on the state $\rho_{g_1}$ we get 
\begin{align*}
   \rho'_{g_1}\coloneqq  U_{\text{split}}\rho_{g_1} U_{\text{split}}^{\dagger}=\frac{1}{|\cU|}\sum_{\eta\in \hcG_2}\ketbra{\eta}{\eta}\otimes \left(\sum_{\chi,\chi'\in \hcG_{1}}\sqrt{\lambda_{(\chi,\eta)}\lambda_{(\chi',\eta)}}\chi(g_1)\chi'^{-1}(g_1)\ketbra{\chi}{\chi'}\right).
\end{align*}
Hence, we get a heralded mixture of PSCs which is characterized via ensemble $\{p_{\eta},\blambda^{\eta}\}_{\eta\in \hcG_2}$. The probability $p_{\eta}$ that corresponds to the eigen list $\blambda^{\eta}$ is computed as 
\begin{align}
    p_{\eta} & =\Tr\left(\bra{\eta}\otimes \mI \rho'_{g_1}\ket{\eta}\otimes \mI\right)\\
    & = \Tr\left(\frac{1}{|\cU|}\sum_{\chi,\chi'\in \hcG_{1}}\sqrt{\lambda_{(\chi,\eta)}\lambda_{(\chi',\eta)}}\chi(g_1)\chi'^{-1}(g_1)\ketbra{\chi}{\chi'}\right)\\
    & = \frac{1}{|\cU|}\sum_{\chi\in \hcG_{1}}\lambda_{(\chi,\eta)}.
\end{align}
Similarly, the post measurement output PSC  after marginalization is computed as 
\begin{align}
  &  \frac{\bra{\eta}\otimes \mI \rho'_{g_1}\ket{\eta}\otimes \mI}{\Tr\left(\bra{\eta}\otimes \mI \rho'_{g_1}\ket{\eta}\otimes \mI\right)}\\
  & = \frac{\bra{\eta}\otimes \mI \rho'_{g_1}\ket{\eta}\otimes \mI}{p_{\eta}}\\
  & = \frac{1}{|\cU|p_{\eta}}\sum_{\chi,\chi'\in \hcG_{1}}\sqrt{\lambda_{(\chi,\eta)}\lambda_{(\chi',\eta)}}\chi(g_1)\chi'^{-1}(g_1)\ketbra{\chi}{\chi'}.
\end{align}
Defining the state $\ket{\psi^{\eta}_{g_1}}$ for $g_1\in \cG_{1}$ with
\begin{align*}
    \ket{\psi_{g_1}^{\eta}}\coloneqq \frac{1}{\sqrt{|\cG_1|}}\sum_{\chi\in \hcG_{1}}\sqrt{\lambda^{\eta}_{\chi}}\chi(g_1)\ket{\chi}
\end{align*}
and setting $\lambda^{\eta}_{\chi}=\frac{1}{|\cG_2|p_{\eta}}\lambda_{(\chi,\eta)}$ we get 
\begin{align}
    \ketbra{\psi_{g_1}^{\eta}}{\psi_{g_1}^{\eta}} & = \frac{1}{|\cG_1|}\sum_{\chi,\chi'\in \hcG_{1}}\sqrt{\lambda^{\eta}_{\chi}\lambda^{\eta}_{\chi'}}\chi(g_1)\chi'^{-1}(g_1)\ketbra{\chi}{\chi'}\\
    & = \frac{1}{|\cU|p_{\eta}}\sum_{\chi,\chi'\in \hcG_{1}}\sqrt{\lambda_{(\chi,\eta)}\lambda_{(\chi',\eta)}}\chi(g_1)\chi'^{-1}(g_1)\ketbra{\chi}{\chi'}.
\end{align}
Hence, the  elements of the eigen list $\blambda^{\eta}=\{\lambda^{\eta}_{\chi}\}_{\chi\in \hcG_{1}}$ satisfy the following relation
\begin{align*}
    \lambda^{\eta}_{\chi}=\frac{1}{|\cG_2|p_{\eta}}\lambda_{(\chi,\eta)}.
\end{align*}
\end{proof}
Therefore, the marginalization update maps a $\cU$-covariant PSC to a finite heralded mixture of $\cG_1$-covariant PSCs, with herald register indexed by $\eta\in \widehat{\cG}_2$. Equivalently, the outgoing message is completely described by the ensemble $\{p_{\eta},\blambda^{\eta}\}_{\eta\in \widehat{\cG}_2}$.

Operationally, the unitary $U_{\mathrm{split}}$ separates the retained and discarded dual coordinates. The average over $g_2\in \cG_2$ forces equality of the discarded dual index through character orthogonality, thereby making the second register classical. This is the origin of the heralded decomposition in the marginalization update.
\subsection{Quantum Message Passing for Automorphism Factor}

Finally, we consider automorphism factors. In this case, no genuine mixing occurs as an automorphism simply relabels the group elements, and in the character domain this becomes a permutation induced by the dual automorphism. Consequently, the output remains a single $\cG$-covariant PSC.

\begin{lem}\label{lem:automorphism factor}
Consider $\cG$-covariant PSC $W\colon g\rightarrow \ketbra{\psi_{g}}{\psi_{g}}$ for $g\in \cG$ with Gram matrix $\Gamma_W$ characterized by the eigen list $\blambda=\{\lambda_{\chi}\}_{\chi\in \hcG}$. Let  $\phi:\cG\rightarrow \cG$ be an automorphism on the abelian group $\cG$. Then the induced CQ channel $W_{\phi}$ is a $\cG-$covariant PSC and is characterized by the eigen list $\blambda^{\phi}=\{\lambda^{\phi}_{\chi}\}_{\chi\in \hcG}$ such that 
\begin{align*}
    \lambda^{\phi}_{\chi}=\lambda_{\hat{\phi}(\chi)}
\end{align*}
for all $\chi\in \hcG$, and $\hat{\phi}\colon\hcG\rightarrow \hcG$ is the dual map.
    
\end{lem}
\begin{proof}
    The induced CQ channel $W_{\phi}$ has output states of the form 
    \begin{align*}
        \rho_{g}^{\phi}= \ketbra{\psi_{\phi^{-1}(g)}}{\psi_{\phi^{-1}(g)}}
    \end{align*}
    for all $g\in \cG$. 
Let $\Gamma_{W_{\phi}}$ be the Gram matrix for the output PSC where elements of the first row satisfy $\gamma_{g}^{\phi}=\gamma_{\phi^{-1}(g)}$ for all  $g\in \cG$. Then, $\forall \chi\in \hcG$, we get 
\begin{align*}
    \lambda^{\phi}_{\chi} & =\sum_{g\in \cG}\gamma^{\phi}_{g}\chi^{-1}(g)\\
    & = \sum_{g\in \cG}\gamma_{\phi^{-1}(g)}\chi^{-1}(g)\\
    & = \sum_{g\in \cG}\gamma_{g}\chi^{-1}(\phi(g))\\
    & =\sum_{g\in \cG}\gamma_{g}(\hat{\phi}(\chi))^{-1}(g)\\
    & = \lambda_{\hat{\phi}(\chi)},
\end{align*}
which completes the proof.
\end{proof}
Thus, the automorphism factor preserves the class of single $\cG$-covariant PSCs and does not introduce a nontrivial herald register. Operationally, it acts only by a relabeling of the character basis through the dual automorphism $\hat{\phi}$.
\begin{table*}[ht]
\caption{Local quantum message passing updates over finite abelian groups.}
\label{tab:qmp-summary}
\centering
\scriptsize
\setlength{\tabcolsep}{3pt}
\renewcommand{\arraystretch}{1.15}
\begin{tabular}{|>{\raggedright\arraybackslash}p{2.45cm}|
                >{\raggedright\arraybackslash}p{3.15cm}|
                >{\raggedright\arraybackslash}p{2.20cm}|
                >{\raggedright\arraybackslash}p{4.75cm}|}
\hline
\textbf{Factor} & \textbf{Input} & \textbf{Output} & \textbf{Update of sufficient statistics} \\
\hline

Check
&
Two $\cG$-covariant PSCs with eigen lists
$\blambda_1=\{\lambdaa_{\chi}\}_{\chi\in\hcG}$,
$\blambda_2=\{\lambdab_{\chi}\}_{\chi\in\hcG}$
&
Heralded mixture indexed by $\hcG$
&
$\begin{aligned}[t]
p_{\chi}^{\cnop}
&= \frac{1}{|\cG|^2}\sum_{\chi'\in\hcG}
   \lambdaa_{\chi\chi'}\,\lambdab_{\chi'},\\
\lambda_{\chi'}^{\cnop,\chi}
&= \frac{1}{|\cG|\,p_{\chi}^{\cnop}}\,
   \lambdaa_{\chi\chi'}\,\lambdab_{\chi'}.
\end{aligned}$ \\
\hline

Equality
&
Two $\cG$-covariant PSCs with eigen lists
$\blambda_1,\blambda_2$
&
Single $\cG$-covariant PSC
&
$\begin{aligned}[t]
\lambda_{\chi}^{\vnop}
&= \frac{1}{|\cG|}
   \sum_{\chi'\in\hcG}
   \lambdaa_{\chi'}\,\lambdab_{\chi\chi'^{-1}}.
\end{aligned}$ \\
\hline

Homomorphism
&
$\cG_1$-covariant PSC with eigen list
$\blambda_1=\{\lambdaa_{\chi}\}_{\chi\in\F{\cG_1}}$;
surjective $\phi:\cG_1\to\cG_2$
&
Heralded mixture indexed by $\cT$
&
$\begin{aligned}[t]
p_{\eta}
&= \frac{1}{|\cG_1|}
   \sum_{\chi\in \eta\,\image\F{\phi}}
   \lambdaa_{\chi},\\
\lambda_{\xi}^{\eta}
&= \frac{|\cG_2|}{|\cG_1|\,p_{\eta}}\,
   \lambdaa_{\eta\,\F{\phi}(\xi)}.
\end{aligned}$ \\
\hline

Marginalization
&
$\cU$-covariant PSC on $\cU=\cG_1\times \cG_2$ with eigen list
$\blambda=\{\lambda_{(\chi,\eta)}\}$
&
Heralded mixture indexed by $\F{\cG_2}$
&
$\begin{aligned}[t]
p_{\eta}
&= \frac{1}{|\cU|}
   \sum_{\chi\in\F{\cG_1}}
   \lambda_{(\chi,\eta)},\\
\lambda_{\chi}^{\eta}
&= \frac{1}{|\cG_2|\,p_{\eta}}\,
   \lambda_{(\chi,\eta)}.
\end{aligned}$ \\
\hline

Automorphism
&
$\cG$-covariant PSC with eigen list
$\blambda=\{\lambda_{\chi}\}_{\chi\in\hcG}$;
automorphism $\phi:\cG\to\cG$
&
Single $\cG$-covariant PSC
&
$\begin{aligned}[t]
\lambda_{\chi}^{\phi}
&= \lambda_{\F{\phi}(\chi)}.
\end{aligned}$ \\
\hline
\end{tabular}
\end{table*}
\subsection{Closure of Quantum Message Passing on Tree Factor Graphs}
% The preceding lemmas show that each local factor update preserves the same class of sufficient statistics used by quantum message passing. The next theorem summarizes this closure property for tree factor graphs assembled from the local primitives considered in this section.

The preceding lemmas establish local closure for each primitive factor. More precisely, conditioned on the incoming herald values, every local update produces either a single group-covariant PSC or a finite heralded mixture of such PSCs on the appropriate output alphabet. The next theorem shows that this local closure property composes over any tree factor graph assembled from the primitive factors considered in this section.
\begin{theorem}[Closure of quantum message passing on trees]
\label{thm:qmp-closure}
% For each finite abelian group $\cG$, consider heralded mixture of PSCs with output states of the form
% \[
% W(g)=\sum_{x\in \cX} p_x\, W_x(g)\otimes \ketbra{x}{x},
% \qquad g\in \cG,
% \]
% where $\cX$ is finite, $\{p_x\}_{x\in \cX}$ is a probability distribution, and
% each $W_x$ is a $\cG$-covariant PSC. Denote this class as $\cM(\cG)$.

For each finite abelian group $\cG$, let $\cM(\cG)$ denote the class of heralded mixtures of $\cG$-covariant PSCs introduced in
Section~\ref{sec:background}. Let $T$ be a tree factor graph whose edge variables take values in finite abelian groups and whose local factors are check, equality, homomorphism,
marginalization, or automorphism factors. Fix a directed edge $e$ of $T$, and let $\cG_e$ denote the alphabet on that edge. Assume that every incoming leaf message in the computation tree associated with $e$ belongs to $\cM(\cG_{e'})$ on the corresponding leaf alphabet $\cG_{e'}$.

Then the message produced by quantum message passing on the directed edge $e$ belongs to $\cM(\cG_e)$.
Equivalently, every directed message generated by exact quantum message passing
on $T$ is a finite heralded mixture of group-covariant PSCs on the alphabet of
that edge.
\end{theorem}

\begin{proof}
We prove the claim by induction on the depth of the computation tree of the directed edge $e$. If the depth is zero, then $e$ is a leaf edge and the claim holds by assumption. Now suppose the claim holds for every directed edge whose computation tree has depth at most $d$, and let $e$ be a directed edge whose computation tree has depth $d+1$. Let $f$ be the factor node at which the message on $e$ is computed, and let $e_1,\dots,e_r$ be the incoming directed edges to $f$. Write $\cG_{e_i}$ for the alphabet on $e_i$.
By the induction hypothesis, for each $i\in\{1,\dots,r\}$ the incoming message on $e_i$ has the form
\begin{align*}
    W_i(g_i) =
\sum_{x_i\in \cX_i}
p_i(x_i)\,W_{i,x_i}(g_i)\otimes \ketbra{x_i}{x_i},
\qquad g_i\in \cG_{e_i},
\end{align*}
where $\cX_i$ is finite and each $W_{i,x_i}$ is a $\cG_{e_i}$-covariant PSC. Set $\cX\coloneqq \cX_1\times\cdots\times \cX_r$. For $x=(x_1,\dots,x_r)\in \cX$,  define $p(x)\coloneqq \prod_{i=1}^r p_i(x_i).$ Conditioned on the classical herald $x$, the incoming messages are the group-covariant PSCs $W_{1,x_1},\dots,W_{r,x_r}.$
Applying the local quantum message passing rule corresponding to the factor $f$ to these channels produces an outgoing CQ channel on $\cG_e$ which, by the results of Section~\ref{sec:quantum message passing factor graph lemmas}, is either a $\cG_e$-covariant PSC or a finite heralded mixture of $\cG_e$-covariant PSCs. Hence, for each $x\in \cX$, there exist a finite set
$\cY_x$, a probability distribution $\{q_x(y)\}_{y\in \cY_x}$, and
$\cG_e$-covariant PSCs $\{V_{x,y}\}_{y\in \cY_x}$ such that
\begin{align*}
    \widetilde W_x(g)
=
\sum_{y\in \cY_x}
q_x(y)\,V_{x,y}(g)\otimes \ketbra{y}{y},
\qquad g\in \cG_e.
\end{align*}
Therefore, the unconditional outgoing message is
\begin{align*}
    W_e(g)
=
\sum_{x\in \cX} p(x)\,\widetilde W_x(g)\otimes \ketbra{x}{x},
\qquad g\in \cG_e.
\end{align*}
Substituting the expression for $\widetilde W_x(g)$ yields
\begin{align*}
    W_e(g)
=
\sum_{x\in \cX}\sum_{y\in \cY_x}
p(x)\,q_x(y)\,V_{x,y}(g)\otimes \ketbra{x,y}{x,y},
\qquad g\in \cG_e.
\end{align*}
This is a finite heralded mixture of $\cG_e$-covariant PSCs, so
$W_e\in \cM(\cG_e)$. This completes the induction.
\end{proof}

\section{Applications to Decoding }\label{sec:application to decoding}

In this section, we explain how the quantum message-passing rules developed in Section~\ref{sec:quantum message passing factor graph lemmas} apply to decoding classical codes on group-covariant PSCs whose symbol alphabets are finite abelian groups. The factors considered earlier, namely check, equality, homomorphism, marginalization, and automorphism factors, already provide the local building blocks for several standard code families. Therefore, once the induced CQ channels for these factors are characterized, the corresponding quantum decoding rules follow by composing the same local updates on the code's factor graph.

Let $\cG$ be a finite abelian group, and let $W:g\mapsto \ketbra{\psi_g}{\psi_g}$ for $g\in \cG$ be a $\cG$-covariant PSC. A code over $\cG$ is described by a factor graph whose edge variables take values in finite abelian groups and whose local constraints are expressed through equality constraints, parity constraints, general homomorphisms, automorphisms, and marginalization operations. Along each edge, quantum message passing propagates the induced message in the form provided by Section~\ref{sec:quantum message passing factor graph lemmas}. Since this class of messages is preserved under the local factor updates, any code whose factor graph is assembled from these primitives can be decoded within the same framework.

This viewpoint recovers the previously studied $q$-ary setting as a special case, while extending the same decoding description to arbitrary finite abelian groups and to factor graphs whose local constraints are described by homomorphisms between products of abelian groups.

% Since for polar codes and LDPC codes, the relevant eigen list based quantum message passing updates were already established in the symmetric $q$-ary setting, and our purpose here is primarily to show how those constructions extend to arbitrary finite abelian groups through homomorphism and automorphism factors \cite{mandal2026belief}. By contrast, the trellis-based convolutional recursion in terms of
% character-indexed eigen lists and marginalization is new in the present framework. Accordingly, we treat the convolutional decoder in greater detail and then use it as the constituent building block for turbo decoding.

For polar codes and LDPC/group codes, the relevant eigen list based quantum message-passing updates were already identified in the symmetric $q$-ary setting \cite{mandal2026belief}. Our purpose here is therefore not to repeat the full development in the symmetric $q$-ary setting, but rather to show how those constructions extend to arbitrary finite abelian groups once the local constraints are expressed through homomorphism and automorphism factors. In the polar case, this extension appears through the decomposition of the kernel into automorphism, check, and equality factors. In the LDPC/group-code case, it appears through the observation that Tanner constraints over $\cG$ are assembled from equality, check, homomorphism, and automorphism factors, so the closure result of Theorem~\ref{thm:qmp-closure} applies directly on trees. By contrast, for convolutional and turbo codes, the finite-state-machine recursion in terms of character-indexed eigen lists together with marginalization is a more substantial part of the present framework. Accordingly, we treat the convolutional decoder in greater detail and then use it as the constituent building block for turbo decoding.
\subsection{Polar Codes over Abelian Groups}

A polar transform over a finite abelian group $\cG$ is specified by an invertible homomorphism $K\colon \cG^2\rightarrow \cG^2$.
If $(u_1,u_2)\in \cG^2$ denotes the input vector and
$(x_1,x_2)=K(u_1,u_2)$ denotes the encoded vector, then the factor graph for the kernel can be decomposed into elementary homomorphism, check, equality, and
automorphism factors. Therefore, the synthetic channels of the polar transform are obtained recursively by repeated application of the update rules of Section~\ref{sec:quantum message passing factor graph lemmas}. In the abelian group setting, when we consider the Ar{\i}kan kernel of the input pair $(u_1,u_2)$, the encoded pair $(x_1,x_2)$ satisfy the following relation
\[
(x_1,x_2)=(u_1u_2,u_2).
\]
Let the encoded pair $(x_1,x_2)$ be sent through the $\cG$-covariant CQ channel $W_1$ and $W_2$. Then, for polar decoding, we first construct the induced CQ channels seen by $u_1$ and $u_2$. The \emph{bad} channel, i.e. the CQ channel seen by $u_1$ can be written as 
\begin{align*}
    W^{-}(u_1) & =\frac{1}{|\cG|}\sum_{u_2\in \cG} W_1(u_1u_2)\otimes W_2(u_2)\\
    & = \frac{1}{|\cG|}\sum_{u\in \cG} W_1(u)\otimes W_2(uu_{1}^{-1}).
\end{align*}
Next, we define the $\cG$-covariant CQ channel $W_{2}^{\text{inv}}$ such that it satisfies 
\begin{align*}
    W_{2}^{\text{inv}}(g)\coloneqq W_{2}(g^{-1}).
\end{align*}
Thus, $W^{-}$ can be written as 
\begin{align*}
    W^{-}(u_1) & = \frac{1}{|\cG|}\sum_{u\in \cG} W_1(u)\otimes W_2^{\text{inv}}(u^{-1}u_{1})\\
    & = [W_{1}\cnop W_2^{\text{inv}}](u_1).
\end{align*}
Since $W_{2}^{\text{inv}}$ corresponds to an induced CQ channel from automorphism factor, the channel $W^{-}$ can be analyzed via the automorphism factor and check factor. 
The CQ channel seen by $u_2$ uses the knowledge of $u_1$ as classical side information. Thus, the induced CQ channel, which is also referred to as the \emph{good} channel, can be written as 
\begin{align*}
    W^{+}(u_2)=\frac{1}{|\cG|}\sum_{u_1\in \cG}\ketbra{u_1}{u_1}\otimes W_{1}(u_1u_2)\otimes W_2(u_2).
\end{align*}
Let the channel $W_1$ be $\cG$-covariant under the unitary representation $\{U_{g}\}_{g\in \cG}$ such that 
\begin{align*}
    U_{g}W_{1}(u_1)U_{g}^{\dagger}=W_{1}(u_1g).
\end{align*}
Define the unitary $U^{\text{flip}}$ as follows
\begin{align*}
    U^{\text{flip}}= \sum_{u_1\in \cG}\ketbra{u_1}{u_{1}}\otimes U_{u_1}^{\dagger}.
\end{align*}
Applying $U^{\text{flip}}$ on the outputs of channel $W^+ $ $\forall u_2\in \cG$, we get 
\begin{align}\label{eq:bit node polarization}
 &   U^{\text{flip}} W^+(u_2)(U^{\text{flip}})^{\dagger}\\
 &  =\frac{1}{|\cG|}\sum_{u_1\in \cG}\ketbra{u_1}{u_1}\otimes U_{u_1}^{\dagger} W_1(u_1u_2) U_{u_1}\otimes W_2(u_2)\nonumber\\
 & = \frac{\mI}{|\cG|}\otimes W_{1}(u_2)\otimes W_{2}(u_2)\nonumber\\
 & = \frac{\mI}{|\cG|}\otimes [W_{1}\vnop W_2](u_2).
\end{align}
Thus, the channel $W^{+}$ is isometrically equivalent to an induced CQ channel of equality factor. Since we can use this construction to recursively define polar codes of arbitrary length ($2^n$) with $n\in \mathbb{N}$, the quantum message passing framework can be used to decode polar codes on group-covariant PSCs.
The advantage of the abelian-group formulation is that the polar kernel need not arise from field multiplication. It suffices that the local operations be homomorphisms or automorphisms of $\cG$. Hence, the same polar quantum message based decoding description applies not only to cyclic alphabets and additive field alphabets, but also to non-cyclic groups such as direct products of cyclic groups.

\subsection{LDPC Codes and Group Codes}
A Tanner graph whose edge variables take values in finite abelian groups consists of variable nodes and check nodes.
% A Tanner graph over a finite alphabet of abelian groups $\cG$ consists of variable nodes and check nodes. 
Each variable node enforces equality among all edge copies of the same symbol, and therefore its local quantum message passing rule is given by the equality lemma of Section~\ref{sec:quantum message passing factor graph lemmas}.
Each check node enforces a group constraint of the form
\[
\phi_{1}(x_{v_1})\phi_{2}(x_{v_2})\cdots \phi_{d_a}(x_{v_{d}})=e_{\cH},
\] where $\cH$ is an abelian group associated with the check and each
$\phi_{j}:\cG\to \cH$ is a homomorphism. Thus, the quantum message passing rule at a check node is an application of the homomorphism and check factor update.
In the simplest case, $\cH=\cG$ and each $\phi_{j}$ is either the identity map or an automorphism of $\cG$. Then the parity constraint reduces to a generalized group parity equation, and the quantum message passing update is obtained by repeated use of the check lemma
together with automorphism relabeling on the incident edges. When
$\cG=(\mathbb F_q,+)$, multiplication by a nonzero field element is an automorphism of the additive group, so the present formulation recovers the standard nonbinary LDPC setting as a special case.

% Accordingly, a quantum message passing decoder for a Tanner graph over a finite abelian group requires only the following local operations: equality-node combination at variable nodes, check or
% homomorphism combination at parity nodes, and automorphism relabeling where needed. Thus, the Section~\ref{sec:quantum message passing factor graph lemmas} lemmas already provide the exact local decoding rules for LDPC and more general group codes.

Accordingly, a quantum message passing decoder for a Tanner graph over a finite abelian group requires only the following local operations: equality-node combination at variable nodes, check or homomorphism combination at parity nodes, and automorphism relabeling where needed. Thus, the lemmas of Section~\ref{sec:quantum message passing factor graph lemmas} already provide the exact local decoding rules for LDPC and more general group codes.

When the Tanner graph is cycle-free, this description is exact in the same sense as belief propagation on a tree. Indeed, a Tanner tree over $\cG$ is assembled entirely from equality, check, homomorphism, and automorphism factors, so Theorem~\ref{thm:qmp-closure} applies directly. Hence every directed message produced by quantum message passing along an edge $e$ belongs to the class $\cM(\cG_e)$, where $\cG_e$ denotes the alphabet associated with that edge. In particular, the posterior message at any target variable node is obtained by exact composition of the incoming channel observations and the parity-check constraints, and remains a finite heralded mixture of $\cG$-covariant PSCs.

\subsection{Convolutional and Turbo Decoding}

% We now specialize the local update rules from
% Section~\ref{sec:quantum message passing factor graph lemmas} to trellis-based decoding over finite abelian groups. A convolutional encoder over an abelian group admits a finite-state-machine representation, and one trellis section can be viewed
% as a local factor graph on the branch variable. The observed branch symbols are described by surjective homomorphisms, so the corresponding quantum message update is obtained by pullback through these homomorphisms, followed by local combination
% and marginalization to the next state. We first describe this constituent trellis decoder for a convolutional code over $\cG$, and then use the same constituent decoder as the basic ingredient in the turbo construction.

We now specialize the local update rules from
Section~\ref{sec:quantum message passing factor graph lemmas} to trellis-based decoding over finite abelian groups.
A convolutional encoder over an abelian group admits a finite-state-machine representation, and one trellis section can be viewed
as a local factor graph on the branch variable. The observed branch symbols are described by surjective homomorphisms, so the corresponding quantum message update is obtained by the special surjective homomorphism lemma, followed by local eigen list combination on the branch variable and marginalization to the next state.
We first describe this constituent trellis decoder for a convolutional code over $\cG$ \cite{piveteau2025efficient}, and then use the same constituent decoder as the basic ingredient in the turbo construction.

\subsubsection{Convolutional Decoding}
% For simplicity, we first consider a single-input convolutional encoder with memory $m$, whose input alphabet is a finite abelian group $\cG$.
% The encoder setup is as follows. We consider a finite state machine which at time $t$ admits symbol $g_t$ uniformly chosen from abelian group $\cG$. The machine also holds a state $S_t$ which stores previous $m$ input symbols such that 
% \begin{align*}
%     S_{t}=(g_{t-1},\dots,g_{t-m})
% \end{align*}
We first describe the classical finite-state-machine representation of a convolutional encoder over a finite abelian group $\cG$. At time $t$, the encoder takes an input symbol $g_t \in \cG$ and stores a state $S_t$ consisting of the previous $m$ input symbols. The state $S_t$ is given by
\begin{align*}
    S_t = (g_{t-1},\dots,g_{t-m}) \in \cG^m.
\end{align*}

% It is convenient to regard one trellis section as being indexed by the branch variable $U_{t}\coloneqq (g_t,S_t)\in \cG^{m+1}$.
% The encoder maps the input $U_t$ to generate $r$ outputs $x^{i}_t\in\cH$ via surjective homomorphisms $L_{i}\colon \cG^{m+1}\rightarrow \cH$ for $1\leq i\leq r$ and $\cH$ is some abelian group. More formally, we have 
% \begin{align*}
%       L_{i}(U_t)= x_{t}^i,\quad 1\leq i\leq r.
% \end{align*}
It is convenient to regard one trellis section as being indexed by the pair $(g_t,S_t) \in \cG^{m+1}$. The encoder maps this pair to $r$ output symbols $x_t^i \in \cH$, for $1 \le i \le r$, through surjective homomorphisms $L_i : \cG^{m+1} \to \cH$. Explicitly, the map is given by
\begin{align*}
    L_i(g_t,S_t)=x_t^i,\qquad 1 \le i \le r.
\end{align*}

% Thus
% \begin{align*}
%     S_{t+1} = (g_t,g_{t-1},\dots,g_{t-m+1}) \in \cG^m.
% \end{align*}
% This defines a single trellis section. Thus a convolutional codeword is obtained by traversing
% the trellis and applying the branch homomorphisms $\{L_i\}_{i=1}^r$ at each time. Then the state is updated to $S_{t+1}$ by erasing element $g_{t-m}$ from the memory and adding the current element $g_t$. 
This defines one trellis section. A convolutional codeword is obtained by traversing the trellis and applying the branch homomorphisms $\{L_i\}_{i=1}^r$ at each time step. The next state is obtained by shifting the memory and appending the current input symbol.  Therefore, the classical trellis description is determined by the input $(g_t,S_t)$, the output homomorphisms $\{L_i\}_{i=1}^r$, and the state update $S_t \mapsto S_{t+1}$.

% In the context of decoding convolutional codes on PSCs, we consider a $\cH$-covariant PSC $W$ which takes output symbols $x_{t}^i$ for $1\leq i\leq r$ and $1\leq t\leq T$ and outputs a state $\ketbra{\psi_{x_t^i}}{\psi_{x_t^i}}$. Thus, if we continue the encoding procedure till $T$ time steps, we receive the following PSC state at the receiver
% \begin{align}
%  \ket{\Psi}=  \bigotimes_{t=1}^{T} \bigotimes_{i=1}^{r}\ket{\psi_{x_{t}^{i}}}.
% \end{align}
We now specialize this trellis description to decoding on a $\cH$-covariant PSC $W$, where each output symbol $x_t^i \in \cH$ is transmitted through the channel $W:x \mapsto \ketbra{\psi_x}{\psi_x}$ for $x \in \cH$. Thus, after $T$ trellis sections, the receiver observes the tensor-product state
\begin{align*}
    \ket{\Psi}=\bigotimes_{t=1}^T \bigotimes_{i=1}^r \ket{\psi_{x_t^i}}.
\end{align*}

% For decoding, the key object is the induced CQ channel on one trellis section. At time
% $t$, the branch is indexed by the pair $(g_t,S_t)\in \cG^{m+1}$, and the observed
% quantum outputs on the parity branches induce, through the homomorphisms $L_i$, a
% message on this product group. 

% To analyze decoding on one trellis section, we consider the induced CQ channel whose input is the branch variable $(g_t,S_t) \in \cG^{m+1}$ and whose outputs are the quantum states corresponding to the symbols $\{x_t^i\}_{i=1}^r$. Since each output symbol is obtained from $(g_t,S_t)$ through the homomorphism $L_i$, the local trellis update is determined by the homomorphism-factor rules from Section~\ref{sec:quantum message passing factor graph lemmas}, followed by marginalization to the next state.
% The role of quantum message passing is to represent this branch message by
% its sufficient statistics in the character domain and then propagate it through the trellis.
% Accordingly, the state message on $S_t$ is concatenated with uniform random input $g_t\in \cG$, then combined with the observed branch messages, and finally
% marginalized to produce the message on the next state $S_{t+1}$.

To analyze one trellis section, we consider the induced CQ channel whose input is the branch variable $(g_t,S_t)\in\cG^{m+1}$ and whose outputs are the quantum states corresponding to the symbols $\{x_t^i\}_{i=1}^r$. Since each output symbol is obtained
from $(g_t,S_t)$ through the homomorphism $L_i$, the local trellis update is determined by the homomorphism factor rules from
Section~\ref{sec:quantum message passing factor graph lemmas}, followed by marginalization to the next state. Accordingly, the trellis recursion consists of three steps- lift the current state message to the branch variable by adjoining a
uniform input symbol $g_t\in\cG$, combine this branch-prior message with the observed branch messages, and then marginalize to obtain the message on $S_{t+1}$.
The goal is to estimate the input sequence $(g_1,\dots, g_{T})$ from the received state $\ket{\Psi}$ using the knowledge of Gram matrix eigen list $\{\lambdaa_{\chi}\}_{\chi\in \hcH}$ and homomorphisms $\{L_{i}\}_{i=1}^{r}$.

% In quantum message passing, each trellis step consists of four operations- lift the incoming state message to the branch variable, combine with local channel observations via the branch homomorphisms, marginalize the symbol leaving the state to obtain the next-state message, and repeat symmetrically backward for the BCJR recursion

% Denote $\cG_t\coloneqq \cG^{m}$. 

% We initialize the state $S_0=(g_{-1},\dots,g_{-m})$ using fixed elements from group $\cG$. Similarly we terminate the machine with state $S_{T+m+1}=(g_{T+m},\dots,g_{T+1})$. Due to group symmetry, the choice of the elements does not matter for decoding. At time $t$, consider the state $S_t$ and input $g_{t}$. 
% Let $\blambda^{t}=\{\lambda_{\chi}^t\}_{\chi\in \hcGm}$ be an eigen list for quantum state associated with $S_t$. 
% Consider the direct product $\cU\coloneqq \cG^{m+1}$. 

As in the usual trellis setting, we assume fixed initialization and termination states. More specifically, we take
\begin{align*}
    S_0=(g_{-1},\dots,g_{-m}),\qquad
S_{T+m+1}=(g_{T+m},\dots,g_{T+1}),
\end{align*}
where the boundary symbols are fixed elements of $\cG$. By group covariance, the particular choice of these boundary elements does not affect the form of the decoding recursion.

% Then to implement the quantum message passing algorithm based on the state machine, we first compute the effective eigen lists for homomorphisms $L_{i}$ $\forall i$ and the concatenation map $\phi$ which lifts the eigen list message of $S_{t}$ from dual group $\hcGm$ to the dual group $\hcGma$.  
To implement the trellis recursion, we describe the state message at time $t$ by a CQ channel
\[
W_t\colon S_t \mapsto \rho_{S_t},\qquad S_t\in \cG^m.
\]
In general, $W_t$ is not a single $\cG^m$-covariant PSC. Rather, by the marginalization step from the previous trellis section, it is a finite heralded mixture of $\cG^m$-covariant PSCs. Thus we write
\begin{align}
    W_t(S_t)
    =
    \sum_{x\in \mathcal{X}_t}
    p_x^t\,
    W_t^{x}(S_t)\otimes \ketbra{x}{x},
    \qquad S_t\in \cG^m,
\end{align}
where $\mathcal{X}_t$ is a finite herald alphabet, $\{p_x^t\}_{x\in\mathcal{X}_t}$ is a probability distribution, and for each $x\in\mathcal{X}_t$, the channel
\[
W_t^x\colon S_t\mapsto \ketbra{\psi_{S_t}^{\,t,x}}{\psi_{S_t}^{\,t,x}}
\]
is a $\cG^m$-covariant PSC characterized by an eigen list
\[
\blambda^{t,x}
=
\{\lambda_{\zeta}^{t,x}\}_{\zeta\in \hcGm}.
\]
For the initialized boundary state, this description is degenerate, with $\mathcal{X}_0$ consisting of a single element.

For each $x\in\mathcal{X}_t$, we first lift the state message on $\hcGm$ to a branch-prior message on $\hcGma$ by adjoining the trivial character on the new input coordinate. Equivalently, this corresponds to taking the current input symbol $g_t\in\cG$ to be uniform and independent of the state message. The resulting branch-prior eigen list
\[
\blambda^{\mathrm{In},t,x}
=
\{\lambda_{(\zeta,\eta)}^{\mathrm{In},t,x}\}_{(\zeta,\eta)\in \hcGma}
\]
satisfies
\begin{align}
    \lambda_{(\zeta,\eta)}^{\mathrm{In},t,x}
    =
    \begin{cases}
        |\cG|\,\lambda_{\zeta}^{t,x},
        & \eta=\chitriv,\\[0.8ex]
        0,
        & \eta\neq \chitriv,
    \end{cases}
\end{align}
for $\zeta\in \hcGm$ and $\eta\in \hcG$.

Next, for each branch output $1\le i\le r$, let
\[
\blambda^{\mathrm{ch},(i),t}
=
\{\lambda_{\xi}^{\mathrm{ch},(i),t}\}_{\xi\in \widehat{\cH}}
\]
denote the eigen list of the observed channel message on the $i$-th parity branch. Since
$L_i\colon \cG^{m+1}\to \cH$ is a surjective homomorphism, Lemma~\ref{lem:surjective homomorphism special}
shows that the induced branch message is a single $\cG^{m+1}$-covariant PSC with eigen list
\[
\blambda^{(i),t}
=
\{\lambda_{\chi}^{(i),t}\}_{\chi\in \hcGma}
\]
given by
\begin{align}
    \lambda_{\chi}^{(i),t}
    =
    \begin{cases}
        \dfrac{|\cG^{m+1}|}{|\cH|}\lambda_{\xi}^{\mathrm{ch},(
        i),t},
        & \chi=\hat{L}_i(\xi)\text{ for some }\xi\in \widehat{\cH},\\[2ex]
        0,
        & \chi\notin \image(\hat{L}_i),
    \end{cases}
\end{align}
where $\hat{L}_i\colon \widehat{\cH}\to \hcGma$ is the dual map of $L_i$.

Therefore, for each herald value $x\in\mathcal{X}_t$, the combined input message is again a single $\cG^{m+1}$-covariant PSC with eigen list
\begin{align}
    \blambda^{\mathrm{out},t,x}
    =
    \blambda^{\mathrm{In},t,x}
    \vnop
    \Bigl(\vnop_{i=1}^{r}\blambda^{(i),t}\Bigr).
\end{align}

We now obtain the message on the next state $S_{t+1}$ by marginalizing the input message over the symbol that leaves the memory.
Equivalently, we reparameterize the input variable $(g_t,S_t)\in \cG^{m+1}$ as $(S_{t+1},g_{t-m})\in \cG^{m+1}$, where
\begin{align}
    S_{t+1}=(g_t,g_{t-1},\dots,g_{t-m+1})
\end{align}
is the shifted state and $g_{t-m}$ is the discarded memory symbol.

For each fixed $x\in\mathcal{X}_t$, write the canonical branch state as
\begin{align}
    \ket{\psi_{(S_{t+1},g_{t-m})}^{\,t,x}}
    =
    \frac{1}{\sqrt{|\cG^{m+1}|}}
    \sum_{\zeta\in \hcGm,\eta\in \hcG}
    \sqrt{\lambda_{(\zeta,\eta)}^{\mathrm{out},t,x}}
    \,\zeta(S_{t+1})\eta(g_{t-m})\ket{(\zeta,\eta)}.
\end{align}
Averaging uniformly over the discarded symbol $g_{t-m}\in \cG$ and applying
Lemma~\ref{lem:marginalization factor}, we obtain a heralded mixture on the next state.
More precisely, for each $x\in\mathcal{X}_t$,
\begin{align}
    \frac{1}{|\cG|}
    \sum_{g_{t-m}\in \cG}
    \ketbra{\psi_{(S_{t+1},g_{t-m})}^{\,t,x}}{\psi_{(S_{t+1},g_{t-m})}^{\,t,x}}
\end{align}
is characterized by the ensemble
\[
\Bigl\{
\bigl(p_{\eta}^{t+1\mid x},\blambda^{t+1,x,\eta}\bigr)
\Bigr\}_{\eta\in \hcG},
\]
where
\begin{align}
    p_{\eta}^{t+1\mid x}
    =
    \frac{1}{|\cG|^{m+1}}
    \sum_{\zeta\in \hcGm}
    \lambda_{\zeta,\eta}^{\mathrm{out},t,x}
\end{align}
and
\begin{align}
    \lambda_{\zeta}^{t+1,x,\eta}
    =
    \frac{1}{|\cG|\,p_{\eta}^{t+1\mid x}}
    \lambda_{\zeta,\eta}^{\mathrm{out},t,x},
    \qquad \zeta\in \hcGm.
\end{align}
Hence the full next-state message is the heralded mixture
\begin{align}
    W_{t+1}(S_{t+1})
    =
    \sum_{x\in \mathcal{X}_t}\sum_{\eta\in \hcG}
    p_x^t\,p_{\eta}^{t+1\mid x}\,
    W_{t+1}^{x,\eta}(S_{t+1})
    \otimes \ketbra{x,\eta}{x,\eta},
\end{align}
where each $W_{t+1}^{x,\eta}$ is a $\cG^m$-covariant PSC with eigen list
$\blambda^{t+1,x,\eta}$.

Thus, the forward recursion propagates the state channel $W_t$, or equivalently the heralded family
\[
\Bigl\{
\bigl(p_x^t,\blambda^{t,x}\bigr)
\Bigr\}_{x\in\mathcal{X}_t},
\]
through the trellis rather than a single eigen list. This is the abelian-group quantum message-passing analogue of the classical BCJR forward recursion. The backward recursion is obtained by the same construction in reverse time, starting from the fixed terminal state. Once the forward and backward state channels are available, they are combined with the local branch observations to obtain the posterior message on $(g_t,S_t)$, and marginalization over the state yields the posterior message on $g_t$. Hence, the resulting trellis decoder performs symbol-wise MAP decoding of the input symbols.

The same construction extends in a straightforward way to multi-input convolutional encoders by replacing the single input symbol $g_t$ with an input vector and enlarging the trellis state accordingly. Therefore, the finite-state-machine construction above
provides a complete constituent decoder for convolutional codes over finite abelian groups. When $\cG=\mathbb{Z}_q$, this recursion reduces to the usual $q$-ary trellis
viewpoint; for general $\cG$, the same decoder is expressed through eigen list messages indexed by characters and branch constraints given by surjective homomorphisms. This constituent trellis decoder will be the basic ingredient in the turbo decoder described
next.

\subsubsection{Turbo Decoding}

A turbo code is obtained by coupling two constituent convolutional encoders through the same
information sequence and an interleaver. Let $\bg=(g_1,\dots,g_T)\in \cG^T$ be the information block. The first constituent encoder acts directly on $\bg$, while the
second constituent encoder acts on the interleaved block $\Pi(\bg)=(g_{\pi(1)},\dots,g_{\pi(T)}),$ where $\pi$ is a permutation of $\{1,\dots,T\}$ and $\Pi\colon \cG^T\to \cG^T$ is the
induced automorphism of the product group $\cG^T$.

For each constituent encoder $c\in\{1,2\}$, let $m_c$ denote its memory and let
\[
u_t^{(1)}\coloneqq g_t,
\qquad
u_t^{(2)}\coloneqq g_{\pi(t)},
\qquad t=1,\dots,T.
\]
The trellis state of constituent $c$ at time $t$ is
\[
S_t^{(c)}=(u_{t-1}^{(c)},\dots,u_{t-m_c}^{(c)})\in \cG^{m_c},
\] 
and the corresponding input variable is $(u_t^{(c)},S_t^{(c)})\in \cG^{m_c+1}$. For $i=1,\dots,r_c$, let $L_i^{(c)}\colon \cG^{m_c+1}\to \cH$ be the surjective homomorphisms describing the parity outputs of the $c$-th constituent
encoder on one trellis section. As in the convolutional subsection, the boundary states are
taken to be fixed by initialization and termination of the trellis.
To incorporate the systematic observation and the message received from the other constituent,
we also introduce the symbol projection
\[
L_0^{(c)}(u_t^{(c)},S_t^{(c)})=u_t^{(c)}.
\]
Let $\blambda^{\mathrm{sys},(c),t} =\{\lambda^{\mathrm{sys},(c),t}_{\chi}\}_{\chi\in \hcG}$
denote the eigen list of the observed channel message on the systematic symbol $u_t^{(c)}$,
and let $\blambda^{\mathrm{apr},(c),t}
   =\{\lambda^{\mathrm{apr},(c),t}_{\chi}\}_{\chi\in \hcG}$
denote the eigen list of the a priori message entering constituent $c$ on the same symbol. For the parity outputs, let $\blambda^{\mathrm{ch},(c),i,t}
   =
   \{\lambda_{\chi}^{\mathrm{ch},(c),i,t}\}_{\chi\in \widehat{\cH}},
   \qquad i=1,\dots,r_c,$
denote the eigen lists of the observed channel messages on the output symbols.

Exactly as in the convolutional subsection, the constituent state message at time $t$ is in
general not a single $\cG^{m_c}$-covariant PSC. Rather, it is a finite heralded mixture.
Thus, for each constituent $c\in\{1,2\}$, we write the forward state channel as
\[
W_t^{(c),\rightarrow}\colon S_t^{(c)} \mapsto \rho_{S_t^{(c)}}^{(c),t,\rightarrow},
\qquad
S_t^{(c)}\in \cG^{m_c},
\]
with representation
\begin{align}
W_t^{(c),\rightarrow}(S_t^{(c)})
=
\sum_{x\in \mathcal{X}_t^{(c),\rightarrow}}
p_x^{(c),t,\rightarrow}\,
W_t^{(c),x,\rightarrow}(S_t^{(c)})
\otimes \ketbra{x}{x},
\qquad
S_t^{(c)}\in \cG^{m_c},
\end{align}
where $\mathcal{X}_t^{(c),\rightarrow}$ is a finite herald alphabet,
$\{p_x^{(c),t,\rightarrow}\}_{x\in \mathcal{X}_t^{(c),\rightarrow}}$ is a probability
distribution, and for each $x\in \mathcal{X}_t^{(c),\rightarrow}$, the channel
\[
W_t^{(c),x,\rightarrow}\colon
S_t^{(c)}
\mapsto
\ketbra{\psi_{S_t^{(c)}}^{\,(c),t,x,\rightarrow}}
       {\psi_{S_t^{(c)}}^{\,(c),t,x,\rightarrow}}
\]
is a $\cG^{m_c}$-covariant PSC characterized by an eigen list
\[
\blambda^{(c),t,x,\rightarrow}
=
\{\lambda_{\zeta}^{(c),t,x,\rightarrow}\}_{\zeta\in \widehat{\cG^{m_c}}}.
\]
For the initialized boundary state, this description is degenerate, with
$\mathcal{X}_0^{(c),\rightarrow}$ consisting of a single element.

For each $x\in \mathcal{X}_t^{(c),\rightarrow}$, we first lift the state message on
$\widehat{\cG^{m_c}}$ to a branch-prior message on $\widehat{\cG^{m_c+1}}$ by adjoining the
trivial character on the fresh input coordinate. The resulting branch-prior eigen list
\[
\blambda^{\mathrm{In},(c),t,x,\rightarrow}
=
\{\lambda_{(\zeta,\eta)}^{\mathrm{In},(c),t,x,\rightarrow}\}_{(\zeta,\eta)\in \widehat{\cG^{m_c+1}}}
\]
satisfies
\begin{align}
\lambda_{(\zeta,\eta)}^{\mathrm{In},(c),t,x,\rightarrow}
=
\begin{cases}
|\cG|\,\lambda_{\zeta}^{(c),t,x,\rightarrow},
& \eta=\chitriv,\\[0.8ex]
0,
& \eta\neq \chitriv,
\end{cases}
\end{align}
for $\zeta\in \widehat{\cG^{m_c}}$ and $\eta\in \hcG$.

Let
\[
\blambda^{\mathrm{sym},(c),t}
=
\blambda^{\mathrm{sys},(c),t}
\vnop
\blambda^{\mathrm{apr},(c),t}
\]
denote the combined message on the symbol coordinate. Then, for each
$x\in \mathcal{X}_t^{(c),\rightarrow}$, the branch message is a single
$\cG^{m_c+1}$-covariant PSC with eigen list
\begin{align}
\blambda^{\mathrm{out},(c),t,x,\rightarrow}
=
\blambda^{\mathrm{In},(c),t,x,\rightarrow}
\vnop
\blambda^{\mathrm{sym},(c),t}
\vnop
\Bigl(\vnop_{i=1}^{r_c}\blambda^{(c),i,t}\Bigr).
\end{align}
This is the direct analogue of the local branch update in the convolutional decoder, with the
only difference that the symbol component now contains both the systematic observation and the
incoming a priori information.

Marginalizing the dropped symbol from the branch message yields, for each
$x\in \mathcal{X}_t^{(c),\rightarrow}$, a heralded mixture on the next state. In particular,
the induced message on $S_{t+1}^{(c)}$ is characterized by the ensemble
\[
\Bigl\{
\bigl(p_{\eta}^{(c),t+1,\rightarrow\mid x},
      \blambda^{(c),t+1,x,\eta,\rightarrow}\bigr)
\Bigr\}_{\eta\in \hcG},
\]
where
\begin{align}
p_{\eta}^{(c),t+1,\rightarrow\mid x}
=
\frac{1}{|\cG|^{m_c+1}}
\sum_{\zeta\in \widehat{\cG^{m_c}}}
\lambda_{(\zeta,\eta)}^{\mathrm{out},(c),t,x,\rightarrow}
\end{align}
and
\begin{align}
\lambda_{\zeta}^{(c),t+1,x,\eta,\rightarrow}
=
\frac{1}{|\cG|\,p_{\eta}^{(c),t+1,\rightarrow\mid x}}
\lambda_{(\zeta,\eta)}^{\mathrm{out},(c),t,x,\rightarrow},
\qquad
\zeta\in \widehat{\cG^{m_c}}.
\end{align}
Hence the full next-state message is
\begin{align}
W_{t+1}^{(c),\rightarrow}(S_{t+1}^{(c)})
=
\sum_{x\in \mathcal{X}_t^{(c),\rightarrow}}
\sum_{\eta\in \hcG}
p_x^{(c),t,\rightarrow}\,
p_{\eta}^{(c),t+1,\rightarrow\mid x}\,
W_{t+1}^{(c),x,\eta,\rightarrow}(S_{t+1}^{(c)})
\otimes \ketbra{x,\eta}{x,\eta},
\end{align}
where each $W_{t+1}^{(c),x,\eta,\rightarrow}$ is a $\cG^{m_c}$-covariant PSC with eigen list
$\blambda^{(c),t+1,x,\eta,\rightarrow}$.

Thus, the forward recursion propagates the state channel
$W_t^{(c),\rightarrow}$, or equivalently the heralded family
\[
\Bigl\{
\bigl(p_x^{(c),t,\rightarrow},\blambda^{(c),t,x,\rightarrow}\bigr)
\Bigr\}_{x\in \mathcal{X}_t^{(c),\rightarrow}},
\]
through the constituent trellis rather than a single eigen list.
Similarly, the backward recursion is obtained by applying the same constituent trellis
construction in reverse time. Hence the backward state message is again a finite heralded
mixture of $\cG^{m_c}$-covariant PSCs on $S_t^{(c)}$.

Once the forward and backward state channels are available, they are combined with the local
systematic and parity observations to obtain the posterior message on the branch variable
$U_t^{(c)}=(u_t^{(c)},S_t^{(c)})$. Marginalizing this posterior over the state coordinates
yields the posterior message on the information symbol $u_t^{(c)}$. If, instead, one omits
the incoming a priori message $\blambda^{\mathrm{apr},(c),t}$ from the above branch
combination and uses only the state messages together with the systematic and parity channel
observations, then the resulting marginalized message on $u_t^{(c)}$ is the \emph{extrinsic}
message produced by constituent decoder $c$.

The two constituent decoders are coupled through these extrinsic symbol messages. More
precisely, the extrinsic message produced by the first constituent on $g_t$ is interleaved
and used as the a priori message for the second constituent on $g_{\pi(t)}$. Since
$\Pi\colon \cG^T\to \cG^T$ is an automorphism of the product group, this relabeling is an
instance of the automorphism-factor update and does not introduce any new type of message.
Conversely, the extrinsic message produced by the second constituent is deinterleaved by
$\Pi^{-1}$ and then used as the a priori message for the first constituent. The equality
constraint between the two copies of the same information symbol is therefore enforced through
the equality-factor combination, while the interleaver acts only by automorphism relabeling.

Therefore turbo decoding over a finite abelian group is obtained by combining the following
ingredients:
\begin{enumerate}
    \item the finite-state-machine quantum message passing recursion for each constituent convolutional encoder,
    \item equality-factor combination on the information symbols shared by the constituents,
    \item automorphism relabeling induced by the interleaver and deinterleaver.
\end{enumerate}
Hence, no additional local quantum message-passing rule is needed to construct the turbo
decoder beyond the rules already developed in Section~\ref{sec:quantum message passing factor graph lemmas}. In particular, when
$\cG=\mathbb{Z}_q$, the above construction reduces to the usual $q$-ary turbo-decoding
viewpoint, while the present formulation also applies to turbo constructions over arbitrary
finite abelian groups.

\section{Density Evolution}\label{sec:DE}
In this section, we discuss density evolution (DE), which is a standard asymptotic method for analyzing message-passing decoders on sparse or locally tree-like graphs \cite{richardson2008modern}. In the context of decoding binary LDPC codes on qubit classical-quantum channels, Monte Carlo based DE was proposed in \cite{brandsen2022belief}. Later, this was extended to the analysis of binary polar codes \cite{mandal2023belief,mandal2024polar} and binary turbo codes \cite{piveteau2025efficient}. More recently, eigen list based DE was developed for $q$-ary LDPC and polar codes on symmetric PSCs \cite{mandal2026belief}.

In the present framework, DE is based on the same sufficient statistics used for quantum message passing. In particular, each message is represented by the corresponding eigen list indexed by characters of the dual group, together with the associated herald values whenever a local update produces a finite heralded mixture. By Theorem~\ref{thm:qmp-closure}, exact quantum message passing on a tree factor graph preserves this class of messages. Therefore, DE in the present setting requires no new local update rule beyond those already developed in Section~\ref{sec:quantum message passing factor graph lemmas}. When $\cG=\mathbb{Z}_q$, this reduces to the eigen list based DE formulation of \cite{mandal2026belief}. For a general finite abelian group, the same idea extends directly after replacing the cyclic Fourier index by a character of $\hcG$ and using the corresponding abelian-group updates for equality, check, homomorphism, marginalization, and automorphism factors. Hence, for LDPC and polar codes, the DE procedures in \cite{mandal2026belief} generalize directly to the present setting. 
% This is also why Section~\ref{sec:application to decoding} treats those code families only briefly. Once the abelian-group local update rules are in place, both the decoder description and the associated DE formalism follow by the same composition principle as in the $q$-ary case.

Another advantage of this representation is that the performance quantities introduced in Section~\ref{sec:background} can also be tracked during DE. Indeed, once the posterior message on a target symbol is represented by an eigen list, Lemma~\ref{lem:Holevo information and channel fidelity} yields the corresponding symmetric Holevo information, while Lemma~\ref{lem:pgm} yields the corresponding PGM symbol error probability. Therefore, on a tree factor graph, DE can be used to track, at each stage, the symmetric Holevo information and the PGM error rate associated with decoding the root-node symbol. In particular, the posterior message obtained at the root gives the corresponding estimate of the symbol error probability for decoding that symbol.

For convolutional and turbo codes, the constituent decoder is already described in the previous subsection through the finite-state-machine recursion. Therefore, DE is obtained by tracking the extrinsic message on a typical information symbol under the random-interleaver ensemble. In each constituent decoder, the forward and backward state messages are combined with the local systematic and parity observations, and marginalization over the state coordinates yields the posterior message on the information symbol. Omitting the incoming a priori message on the target symbol gives the corresponding extrinsic message. Thus, one DE update for turbo decoding is obtained by sampling incoming a priori eigen lists, applying the constituent forward-backward recursion on a selected window, and recording the resulting extrinsic eigen list on the target symbol. The extrinsic message from one constituent is then interleaved and used as the a priori message for the other constituent, while deinterleaving returns the updated message for the next outer iteration. Since the interleaver is an automorphism of the product group, this step again introduces no new type of message.

Because the exact trellis recursion propagates finite heralded mixtures, a Monte Carlo implementation of DE may sample the herald values generated by the local updates and retain the corresponding outgoing eigen list. This is sufficient for eigen list based DE, where one iteratively updates a population of extrinsic messages and estimates the resulting symbol error probability from the posterior eigen list using the pretty-good measurement. In this sense, DE for turbo codes in the present paper is the direct abelian-group extension of the eigen list based DE viewpoint of \cite{mandal2026belief}, with the constituent update supplied by the finite-state-machine recursion above. In the binary turbo setting, our eigen list based DE implementation recovers the numerical thresholds reported in \cite{piveteau2025efficient}.

Consider the turbo code which is characterized via constituent convolutional codes with the function $G(D)=p(D)/q(D)$ and the PSC $W$ with eigen list of the form $\blambda=[\lambda_0,\frac{q-\lambda_0}{q-1},\dots, \frac{q-\lambda_0}{q-1}]$ with $q=3$.
In Fig.~\ref{plot:turbo_threshold_57_rate1-3_q3}, we obtain the DE threshold for rate $\frac{1}{3}$ turbo code with $G(D)=\frac{1+D^2}{1+D+D^2}$. In this case, we vary $\lambda_0$ to obtain the DE threshold which is $\lambda_{\text{DE}}=2.641$ while the Holevo threshold is $\lambda_{H}=2.7287$.
% \begin{figure}[ht]
% \centering
  
%     \includegraphics[width=0.7\linewidth]{pdf_figs/turbo_threshold_57_rate1-3_q3.pdf}
%     \caption{DE Threshold Curve for Rate $\frac{1}{3}$ Turbo Code with $G(D)=\frac{1+D^2}{1+D+D^2}$}
%     \label{plot:turbo_threshold_57_rate1-3_q3}

% \end{figure}

In Fig.~\ref{plot:turbo_threshold_2d_57_rate1-3_q3}, we obtain the two-dimensional heatmap of the DE success probability on the eigen list simplex
$\{\blambda \in \mathbb{R}_{\ge 0}^3 : \lambda_0+\lambda_1+\lambda_2=3\}$
for the rate-$\frac{1}{3}$ turbo code over $\mathbb{Z}_3$ with constituent transfer function
$G(D)=\frac{1+D^2}{1+D+D^2}$.
The solid contour indicates the DE boundary, and the dashed contour indicates the Holevo boundary.

Similarly, in Fig.~\ref{plot:turbo_threshold_57_rate1-4_q3} and Fig.~\ref{plot:turbo_threshold_2d_57_rate1-4_q3}, we plot DE threshold curve and two-dimensional heatmap for rate $\frac{1}{4}$ turbo code over $\mathbb{Z}_3$ with each constituent convolutional generator $G(D)=\frac{1+D^2}{1+D+D^2}$.
% \begin{figure}[ht]
% \centering
  
%     \includegraphics[width=0.7\linewidth]{pdf_figs/turbo_threshold_2d_57_rate1-3_q3.pdf}
%     \caption{Two-Dimensional Heatmap for DE Success Probability for Rate $\frac{1}{3}$ Turbo Code with $G(D)=\frac{1+D^2}{1+D+D^2}$}
%     \label{plot:turbo_threshold_2d_57_rate1-3_q3}

% \end{figure}

% \begin{figure}[ht]
% \centering
  
%     \includegraphics[width=0.7\linewidth]{pdf_figs/turbo_threshold_57_rate1-4_q3.pdf}
%     \caption{DE Threshold Curve for Rate $\frac{1}{4}$ Turbo Code with $G(D)=\frac{1+D^2}{1+D+D^2}$}
%     \label{plot:turbo_threshold_57_rate1-4_q3}

% \end{figure}

\begin{figure}[t]
    
    \subfloat[Rate $\frac{1}{3}$ Turbo Code \label{plot:turbo_threshold_57_rate1-3_q3}]{%
        \includegraphics[width=0.5\linewidth]{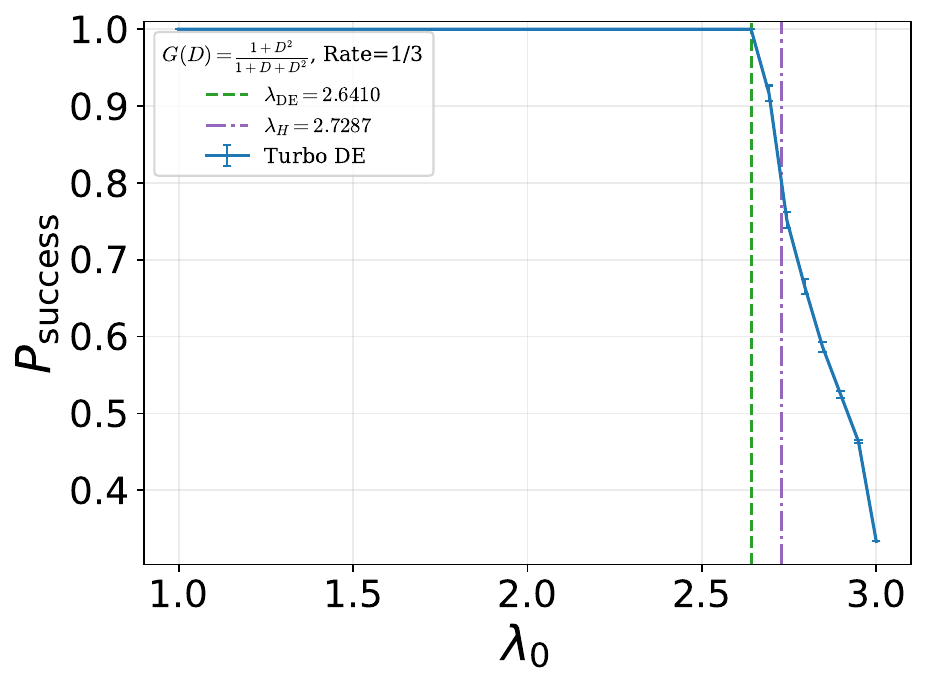}}
    \hfill
    \subfloat[Rate $\frac{1}{4}$ Turbo Code \label{plot:turbo_threshold_57_rate1-4_q3}]{%
        \includegraphics[width=0.5\linewidth]{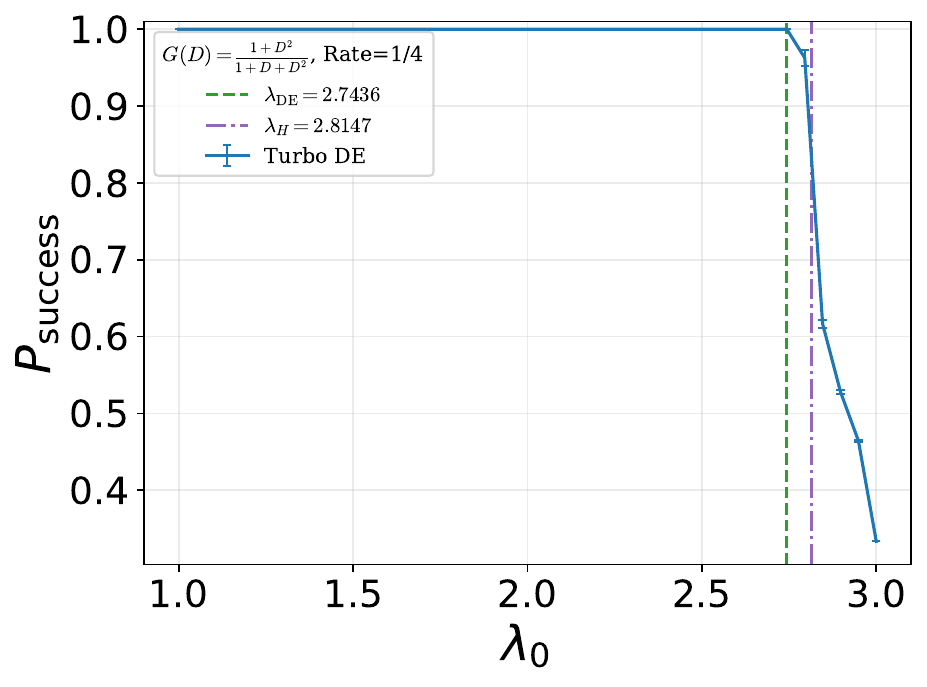}}
    
    \caption{DE Threshold Curve for Turbo Codes with  each constituent convolutional decoder with $G(D)=\frac{1+D^2}{1+D+D^2}$}
    % \label{}
\end{figure}
% \begin{figure}[ht]
% \centering
  
%     \includegraphics[width=0.7\linewidth]{pdf_figs/turbo_threshold_2d_57_rate1-4_q3.pdf}
%     \caption{Two-Dimensional Heatmap for DE Success Probability for Rate $\frac{1}{4}$ Turbo Code with $G(D)=\frac{1+D^2}{1+D+D^2}$}
%     \label{plot:turbo_threshold_2d_57_rate1-4_q3}

% \end{figure}

\begin{figure}[t]
    \centering
    \subfloat[Rate $\frac{1}{3}$ Turbo Code \label{plot:turbo_threshold_2d_57_rate1-3_q3}]{%
        \includegraphics[width=0.5\linewidth]{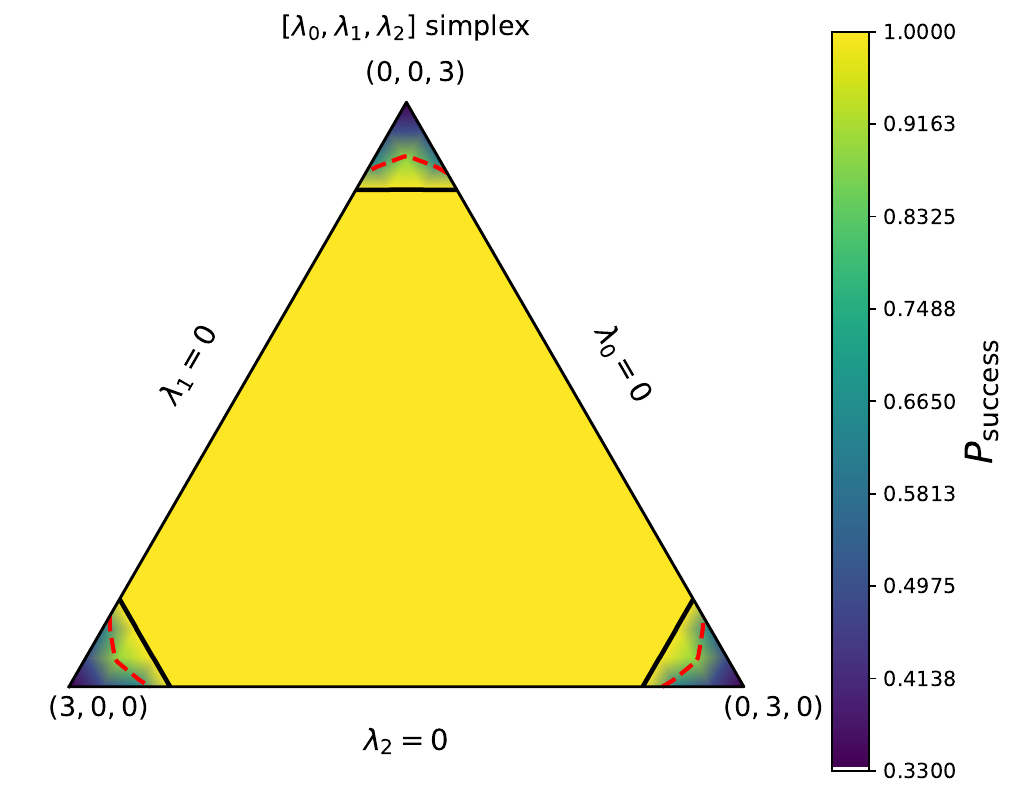}}
    \hfill
    \subfloat[Rate $\frac{1}{4}$ Turbo Code \label{plot:turbo_threshold_2d_57_rate1-4_q3}]{%
        \includegraphics[width=0.5\linewidth]{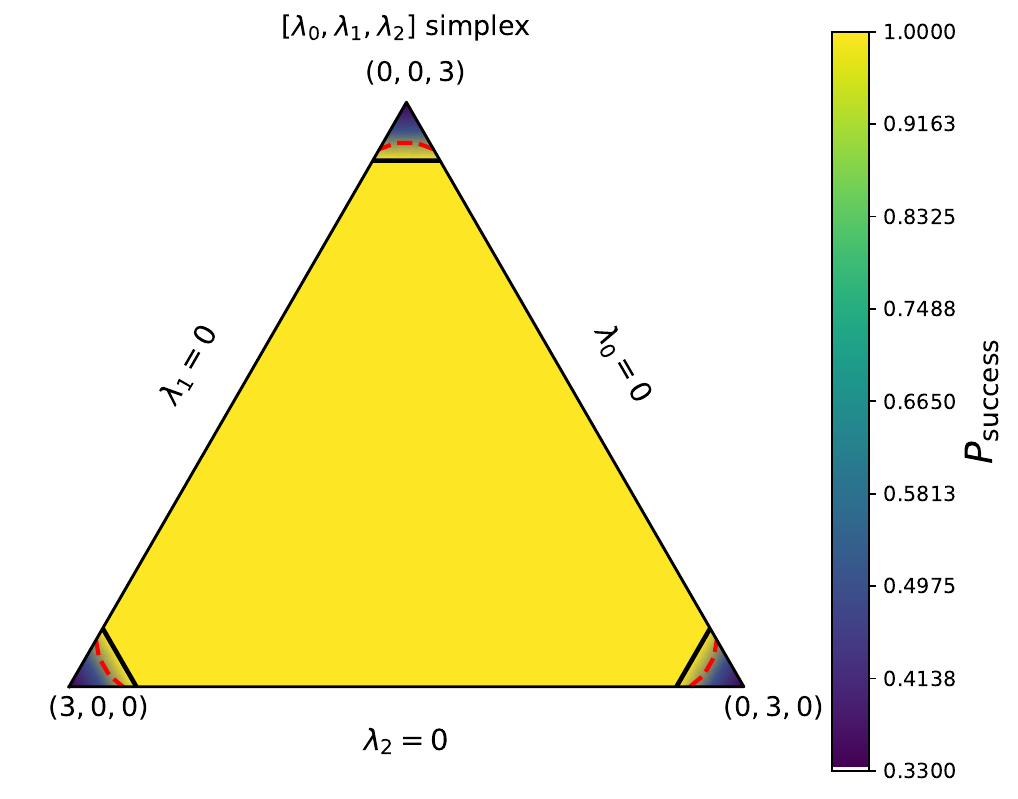}}
    
    \caption{Two-Dimensional Heatmap for DE Success Probability for Turbo Code with each constituent convolutional decoder with $G(D)=\frac{1+D^2}{1+D+D^2}$}
    % \label{}
\end{figure}
\clearpage
\printbibliography
\clearpage
\appendix
\renewcommand{\thesection}{{\color{black}Appendix} \Alph{section}} % Forces the word "Appendix"

\section{Example for Check-Factor Quantum Message Passing }\label{appendix:check factor example}
   \begin{example}
       
Consider $\cG=\mathbb{Z}_{3}\times \mathbb{Z}_{2}$, written additively, and define the characters
\begin{align*}
    \chi_{u,v}(a,b)=\omega^{ua}(-1)^{vb},
\end{align*}
where $\omega=e^{2\pi i/3}$ and $(u,v),(a,b)\in \mathbb{Z}_{3}\times \mathbb{Z}_{2}$. Then $\hcG=\{\chi_{u,v}:u\in \mathbb{Z}_{3},v\in \mathbb{Z}_{2}\}$ and
\begin{align*}
    (\chi_{u,v})^{-1}=\chi_{-u,v}.
\end{align*}
Index the eigen lists in the order
\begin{align*}
    \chi_{0,0},\chi_{1,0},\chi_{2,0},\chi_{0,1},\chi_{1,1},\chi_{2,1}.
\end{align*}
Let
\begin{align*}
    \blambda_{1}=(2,1,0,2,1,0),\qquad
    \blambda_{2}=(2,0,1,1,0,2).
\end{align*}
Then $\sum_{\chi\in \hcG}\lambdaa_{\chi}=\sum_{\chi\in \hcG}\lambdab_{\chi}=6=|\cG|$.

For $\chi=\chi_{1,0}$, Lemma~\ref{lem:check node abelian} gives
\begin{align*}
    p_{\chi_{1,0}}^{\cnop}
    & =\frac{1}{|\cG|^{2}}\sum_{\chi'\in \hcG}\lambdaa_{\chi_{1,0}\chi'}\lambdab_{\chi'}\\
    & =\frac{1}{36}\Big(
    \lambdaa_{\chi_{1,0}}\lambdab_{\chi_{0,0}}
    +\lambdaa_{\chi_{2,0}}\lambdab_{\chi_{1,0}}
    +\lambdaa_{\chi_{0,0}}\lambdab_{\chi_{2,0}}\\
    & \hspace{1.4cm}
    +\lambdaa_{\chi_{1,1}}\lambdab_{\chi_{0,1}}
    +\lambdaa_{\chi_{2,1}}\lambdab_{\chi_{1,1}}
    +\lambdaa_{\chi_{0,1}}\lambdab_{\chi_{2,1}}
    \Big)\\
    & =\frac{1}{36}(2+0+2+1+0+4)=\frac{1}{4}.
\end{align*}
Hence
\begin{align*}
    \blambda_{\chi_{1,0}}^{\cnop}
   & =\left\{\lambda_{\chi'}^{(\cnop,\chi_{1,0})}\right\}_{\chi'\in\hcG}\\
   &  =\frac{1}{6p_{\chi_{1,0}}^{\cnop}}
    \left(
    \lambdaa_{\chi_{1,0}}\lambdab_{\chi_{0,0}},
    \lambdaa_{\chi_{2,0}}\lambdab_{\chi_{1,0}},
    \lambdaa_{\chi_{0,0}}\lambdab_{\chi_{2,0}},
    \lambdaa_{\chi_{1,1}}\lambdab_{\chi_{0,1}},
    \lambdaa_{\chi_{2,1}}\lambdab_{\chi_{1,1}},
    \lambdaa_{\chi_{0,1}}\lambdab_{\chi_{2,1}}
    \right)\\
    &  =\left(\frac{4}{3},0,\frac{4}{3},\frac{2}{3},0,\frac{8}{3}\right).
\end{align*}

Similarly, one obtains
\begin{align*}
    p_{\chi_{0,0}}^{\cnop}=p_{\chi_{0,1}}^{\cnop}=\frac{1}{6},\qquad
    \blambda_{\chi_{0,0}}^{\cnop}=\blambda_{\chi_{0,1}}^{\cnop}
    =(4,0,0,2,0,0),
\end{align*}
\begin{align*}
    p_{\chi_{1,0}}^{\cnop}=p_{\chi_{1,1}}^{\cnop}=\frac{1}{4},\qquad
    \blambda_{\chi_{1,0}}^{\cnop}=\blambda_{\chi_{1,1}}^{\cnop}
    =\left(\frac{4}{3},0,\frac{4}{3},\frac{2}{3},0,\frac{8}{3}\right),
\end{align*}
and
\begin{align*}
    p_{\chi_{2,0}}^{\cnop}=p_{\chi_{2,1}}^{\cnop}=\frac{1}{12},\qquad
    \blambda_{\chi_{2,0}}^{\cnop}=\blambda_{\chi_{2,1}}^{\cnop}
    =(0,0,2,0,0,4).
\end{align*}
   \end{example}
\section{Example of Equality-Factor Quantum Message Passing} \label{appendix:equality factor example}
\begin{example}
Consider $\cG=\mathbb{Z}_{3}\times \mathbb{Z}_{2}$, written additively, and define the characters
\begin{align*}
    \chi_{u,v}(a,b)=\omega^{ua}(-1)^{vb},
\end{align*}
where $\omega=e^{2\pi i/3}$ and $(u,v),(a,b)\in \mathbb{Z}_{3}\times \mathbb{Z}_{2}$. Then $\hcG=\{\chi_{u,v}:u\in \mathbb{Z}_{3},v\in \mathbb{Z}_{2}\}$ and
\begin{align*}
    (\chi_{u,v})^{-1}=\chi_{-u,v}.
\end{align*}
Index the eigen lists in the order
\begin{align*}
    \chi_{0,0},\chi_{1,0},\chi_{2,0},\chi_{0,1},\chi_{1,1},\chi_{2,1}.
\end{align*}
Let
\begin{align*}
    \blambda_{1}=(2,1,0,2,1,0),\qquad
    \blambda_{2}=(2,0,1,1,0,2).
\end{align*}
Then $\sum_{\chi\in \hcG}\lambdaa_{\chi}=\sum_{\chi\in \hcG}\lambdab_{\chi}=6=|\cG|$.

For $\chi=\chi_{1,0}$, Lemma~\ref{lem:bit node abelian} gives
\begin{align*}
    \lambda_{\chi_{1,0}}^{\vnop}
    & =\frac{1}{|\cG|}\sum_{\chi'\in \hcG}\lambdaa_{\chi'}\lambdab_{\chi_{1,0}(\chi')^{-1}}\\
    & =\frac{1}{6}\Big(
    \lambdaa_{\chi_{0,0}}\lambdab_{\chi_{1,0}}
    +\lambdaa_{\chi_{1,0}}\lambdab_{\chi_{0,0}}
    +\lambdaa_{\chi_{2,0}}\lambdab_{\chi_{2,0}}\\
    & \hspace{1.2cm}
    +\lambdaa_{\chi_{0,1}}\lambdab_{\chi_{1,1}}
    +\lambdaa_{\chi_{1,1}}\lambdab_{\chi_{0,1}}
    +\lambdaa_{\chi_{2,1}}\lambdab_{\chi_{2,1}}
    \Big)\\
    & =\frac{1}{6}(0+2+0+0+1+0)=\frac{1}{2}.
\end{align*}
Similarly,
\begin{align*}
    \lambda_{\chi_{0,0}}^{\vnop}
    & =\frac{1}{6}\Big(
    \lambdaa_{\chi_{0,0}}\lambdab_{\chi_{0,0}}
    +\lambdaa_{\chi_{1,0}}\lambdab_{\chi_{2,0}}
    +\lambdaa_{\chi_{2,0}}\lambdab_{\chi_{1,0}}\\
    & \hspace{1.2cm}
    +\lambdaa_{\chi_{0,1}}\lambdab_{\chi_{0,1}}
    +\lambdaa_{\chi_{1,1}}\lambdab_{\chi_{2,1}}
    +\lambdaa_{\chi_{2,1}}\lambdab_{\chi_{1,1}}
    \Big)\\
    & =\frac{1}{6}(4+1+0+2+2+0)=\frac{3}{2},
\end{align*}
and
\begin{align*}
    \lambda_{\chi_{2,0}}^{\vnop}
    & =\frac{1}{6}\Big(
    \lambdaa_{\chi_{0,0}}\lambdab_{\chi_{2,0}}
    +\lambdaa_{\chi_{1,0}}\lambdab_{\chi_{1,0}}
    +\lambdaa_{\chi_{2,0}}\lambdab_{\chi_{0,0}}\\
    & \hspace{1.2cm}
    +\lambdaa_{\chi_{0,1}}\lambdab_{\chi_{2,1}}
    +\lambdaa_{\chi_{1,1}}\lambdab_{\chi_{1,1}}
    +\lambdaa_{\chi_{2,1}}\lambdab_{\chi_{0,1}}
    \Big)\\
    & =\frac{1}{6}(2+0+0+4+0+0)=1.
\end{align*}
Also,
\begin{align*}
    \lambda_{\chi_{0,1}}^{\vnop}=\frac{3}{2},\qquad
    \lambda_{\chi_{1,1}}^{\vnop}=\frac{1}{2},\qquad
    \lambda_{\chi_{2,1}}^{\vnop}=1.
\end{align*}
Hence
\begin{align*}
    \blambda^{\vnop}
    =\left\{\lambda_{\chi}^{\vnop}\right\}_{\chi\in\hcG}
    =\left(\frac{3}{2},\frac{1}{2},1,\frac{3}{2},\frac{1}{2},1\right).
\end{align*}
\end{example}
\section{Examples for Homomorphism Factor Quantum Message Passing}\label{appendix:homomorphism factor example}
\begin{example}
Consider
\begin{align*}
    \cG_{1}=\mathbb{Z}_{4}\times \mathbb{Z}_{3}\times \mathbb{Z}_{2},
    \qquad
    \cG_{2}=\mathbb{Z}_{4}\times \mathbb{Z}_{3},
\end{align*}
written additively. Define the characters of $\hcG_{1}$ by
\begin{align*}
    \chi_{u,v,w}(a,b,c)=i^{ua}\omega^{vb}(-1)^{wc},
\end{align*}
where $u\in\mathbb{Z}_{4}$, $v\in\mathbb{Z}_{3}$, $w\in\mathbb{Z}_{2}$, and $\omega=e^{2\pi i/3}$. Define the characters of $\hcG_{2}$ by
\begin{align*}
    \xi_{r,s}(x,y)=i^{rx}\omega^{sy},
\end{align*}
where $r\in\mathbb{Z}_{4}$ and $s\in\mathbb{Z}_{3}$.

Let the input channel have eigen list $\lambda_{1}=\{\lambda_{\chi}^{(1)}\}_{\chi\in\hcG_{1}}$ given by
\begin{align*}
    \lambda_{\chi_{u,v,w}}^{(1)}=
    \begin{cases}
        2, & (u,w)\in\{(0,0),(1,1)\},\\
        1, & (u,w)\in\{(2,0),(0,1),(2,1),(3,1)\},\\
        0, & (u,w)\in\{(1,0),(3,0)\},
    \end{cases}
\end{align*}
for all $v\in\mathbb{Z}_{3}$. Then
\begin{align*}
    \sum_{\chi\in\hcG_{1}}\lambda_{\chi}^{(1)}=24=|\cG_{1}|.
\end{align*}

Now consider the surjective homomorphism
\begin{align*}
    \phi(a,b,c)=(a+2c,b).
\end{align*}
For $\xi_{r,s}\in\hcG_{2}$, the dual map satisfies
\begin{align*}
    \hat{\phi}(\xi_{r,s})(a,b,c)
    =\xi_{r,s}(\phi(a,b,c))
    =i^{r(a+2c)}\omega^{sb}
    =\chi_{r,s,r \bmod 2}(a,b,c).
\end{align*}
Hence
\begin{align*}
    \mathrm{Im}\,\hat{\phi}
    =\{\chi_{u,v,w}\in\hcG_{1}: w=u \bmod 2\}.
\end{align*}
Take
\begin{align*}
    \mathcal{T}=\{\chi_{0,0,0},\chi_{0,0,1}\}.
\end{align*}
Then
\begin{align*}
    \hcG_{1}
    =\chi_{0,0,0}\,\mathrm{Im}\,\hat{\phi}
    \,\sqcup\,
    \chi_{0,0,1}\,\mathrm{Im}\,\hat{\phi}.
\end{align*}

For $\eta=\chi_{0,0,0}$, Lemma~\ref{lem:surjective homomorphism} gives
\begin{align*}
    p_{\chi_{0,0,0}}
    &=\frac{1}{|\cG_{1}|}\sum_{\chi\in \chi_{0,0,0}\mathrm{Im}\,\hat{\phi}}\lambda_{\chi}^{(1)} \\
    &=\frac{1}{24}\sum_{v\in\mathbb{Z}_{3}}
    \Big(
    \lambda_{\chi_{0,v,0}}^{(1)}
    +\lambda_{\chi_{1,v,1}}^{(1)}
    +\lambda_{\chi_{2,v,0}}^{(1)}
    +\lambda_{\chi_{3,v,1}}^{(1)}
    \Big) \\
    &=\frac{1}{24}\cdot 3(2+2+1+1)=\frac{3}{4}.
\end{align*}
Moreover,
\begin{align*}
    \lambda_{\xi_{r,s}}^{\chi_{0,0,0}}
    &=
    \frac{|\cG_{2}|}{|\cG_{1}|\,p_{\chi_{0,0,0}}}
    \lambda_{\chi_{0,0,0}\hat{\phi}(\xi_{r,s})}^{(1)} \\
    &=
    \frac{12}{24(3/4)}\lambda_{\chi_{r,s,r \bmod 2}}^{(1)}
    =
    \begin{cases}
        \frac{4}{3}, & r=0,1,\\[1mm]
        \frac{2}{3}, & r=2,3.
    \end{cases}
\end{align*}

For $\eta=\chi_{0,0,1}$, we obtain
\begin{align*}
    p_{\chi_{0,0,1}}
    &=\frac{1}{24}\sum_{v\in\mathbb{Z}_{3}}
    \Big(
    \lambda_{\chi_{0,v,1}}^{(1)}
    +\lambda_{\chi_{1,v,0}}^{(1)}
    +\lambda_{\chi_{2,v,1}}^{(1)}
    +\lambda_{\chi_{3,v,0}}^{(1)}
    \Big) \\
    &=\frac{1}{24}\cdot 3(1+0+1+0)=\frac{1}{4},
\end{align*}
and
\begin{align*}
    \lambda_{\xi_{r,s}}^{\chi_{0,0,1}}
    &=
    \frac{|\cG_{2}|}{|\cG_{1}|\,p_{\chi_{0,0,1}}}
    \lambda_{\chi_{0,0,1}\hat{\phi}(\xi_{r,s})}^{(1)} \\
    &=
    \frac{12}{24(1/4)}\lambda_{\chi_{r,s,1+r \bmod 2}}^{(1)}
    =
    \begin{cases}
        2, & r=0,2,\\[1mm]
        0, & r=1,3.
    \end{cases}
\end{align*}

Ordering the characters of $\hcG_{2}$ as
\begin{align*}
    \xi_{0,0},\xi_{0,1},\xi_{0,2},
    \xi_{1,0},\xi_{1,1},\xi_{1,2},
    \xi_{2,0},\xi_{2,1},\xi_{2,2},
    \xi_{3,0},\xi_{3,1},\xi_{3,2},
\end{align*}
the two output eigen lists are
\begin{align*}
    \lambda^{\chi_{0,0,0}}
    =
    \left(
    \frac{4}{3},\frac{4}{3},\frac{4}{3},
    \frac{4}{3},\frac{4}{3},\frac{4}{3},
    \frac{2}{3},\frac{2}{3},\frac{2}{3},
    \frac{2}{3},\frac{2}{3},\frac{2}{3}
    \right),
\end{align*}
and
\begin{align*}
    \lambda^{\chi_{0,0,1}}
    =
    \left(
    2,2,2,0,0,0,2,2,2,0,0,0
    \right).
\end{align*}
\end{example}

\begin{example}
Consider
\begin{align*}
    \phi(a,b,c)=(2a+2c,b),
\end{align*}
as a homomorphism from $\cG_{1}=\mathbb{Z}_{4}\times \mathbb{Z}_{3}\times \mathbb{Z}_{2}$ into $\mathbb{Z}_{4}\times \mathbb{Z}_{3}$. Its image is
\begin{align*}
    \mathrm{Im}\,\phi=\{0,2\}\times \mathbb{Z}_{3}\cong \mathbb{Z}_{2}\times \mathbb{Z}_{3}.
\end{align*}
Thus, to apply the previous construction, it is enough to consider the surjective homomorphism
\begin{align*}
    \tilde{\phi}(a,b,c)=(a+c,b)
\end{align*}
from $\cG_{1}$ onto $\mathbb{Z}_{2}\times \mathbb{Z}_{3}$.

Let $\xi_{r,s}\in\widehat{\mathbb{Z}_{2}\times \mathbb{Z}_{3}}$ be given by
\begin{align*}
    \xi_{r,s}(x,y)=(-1)^{rx}\omega^{sy},
\end{align*}
where $r\in\mathbb{Z}_{2}$ and $s\in\mathbb{Z}_{3}$. Then
\begin{align*}
    \hat{\tilde{\phi}}(\xi_{r,s})(a,b,c)
    =\xi_{r,s}(\tilde{\phi}(a,b,c))
    =(-1)^{r(a+c)}\omega^{sb}
    =\chi_{2r,s,r}(a,b,c).
\end{align*}
Hence
\begin{align*}
    \mathrm{Im}\,\hat{\tilde{\phi}}
    =\{\chi_{2r,s,r}:r\in\mathbb{Z}_{2},\,s\in\mathbb{Z}_{3}\}.
\end{align*}
Take
\begin{align*}
    \mathcal{T}
    =
    \{\chi_{0,0,0},\chi_{1,0,0},\chi_{0,0,1},\chi_{1,0,1}\}.
\end{align*}

For $\eta=\chi_{0,0,0}$, we have
\begin{align*}
    p_{\chi_{0,0,0}}
    &=\frac{1}{24}\sum_{s\in\mathbb{Z}_{3}}
    \Big(
    \lambda_{\chi_{0,s,0}}^{(1)}
    +\lambda_{\chi_{2,s,1}}^{(1)}
    \Big) \\
    &=\frac{1}{24}\cdot 3(2+1)=\frac{3}{8},
\end{align*}
and
\begin{align*}
    \lambda_{\xi_{r,s}}^{\chi_{0,0,0}}
    &=
    \frac{6}{24(3/8)}
    \lambda_{\chi_{0,0,0}\hat{\tilde{\phi}}(\xi_{r,s})}^{(1)} \\
    &=
    \begin{cases}
        \frac{4}{3}, & r=0,\\[1mm]
        \frac{2}{3}, & r=1.
    \end{cases}
\end{align*}
Thus
\begin{align*}
    \lambda^{\chi_{0,0,0}}
    =
    \left(
    \frac{4}{3},\frac{4}{3},\frac{4}{3},
    \frac{2}{3},\frac{2}{3},\frac{2}{3}
    \right).
\end{align*}

For $\eta=\chi_{1,0,0}$, we obtain
\begin{align*}
    p_{\chi_{1,0,0}}
    &=\frac{1}{24}\sum_{s\in\mathbb{Z}_{3}}
    \Big(
    \lambda_{\chi_{1,s,0}}^{(1)}
    +\lambda_{\chi_{3,s,1}}^{(1)}
    \Big) \\
    &=\frac{1}{24}\cdot 3(0+1)=\frac{1}{8},
\end{align*}
and
\begin{align*}
    \lambda_{\xi_{r,s}}^{\chi_{1,0,0}}
    &=
    \frac{6}{24(1/8)}
    \lambda_{\chi_{1,0,0}\hat{\tilde{\phi}}(\xi_{r,s})}^{(1)} \\
    &=
    \begin{cases}
        0, & r=0,\\[1mm]
        2, & r=1.
    \end{cases}
\end{align*}
Hence
\begin{align*}
    \lambda^{\chi_{1,0,0}}=(0,0,0,2,2,2).
\end{align*}

For $\eta=\chi_{0,0,1}$, we have
\begin{align*}
    p_{\chi_{0,0,1}}
    &=\frac{1}{24}\sum_{s\in\mathbb{Z}_{3}}
    \Big(
    \lambda_{\chi_{0,s,1}}^{(1)}
    +\lambda_{\chi_{2,s,0}}^{(1)}
    \Big) \\
    &=\frac{1}{24}\cdot 3(1+1)=\frac{1}{4},
\end{align*}
and
\begin{align*}
    \lambda_{\xi_{r,s}}^{\chi_{0,0,1}}
    =
    \frac{6}{24(1/4)}
    \lambda_{\chi_{0,0,1}\hat{\tilde{\phi}}(\xi_{r,s})}^{(1)}
    =1.
\end{align*}
Therefore
\begin{align*}
    \lambda^{\chi_{0,0,1}}=(1,1,1,1,1,1).
\end{align*}

Finally, for $\eta=\chi_{1,0,1}$,
\begin{align*}
    p_{\chi_{1,0,1}}
    &=\frac{1}{24}\sum_{s\in\mathbb{Z}_{3}}
    \Big(
    \lambda_{\chi_{1,s,1}}^{(1)}
    +\lambda_{\chi_{3,s,0}}^{(1)}
    \Big) \\
    &=\frac{1}{24}\cdot 3(2+0)=\frac{1}{4},
\end{align*}
and
\begin{align*}
    \lambda_{\xi_{r,s}}^{\chi_{1,0,1}}
    &=
    \frac{6}{24(1/4)}
    \lambda_{\chi_{1,0,1}\hat{\tilde{\phi}}(\xi_{r,s})}^{(1)} \\
    &=
    \begin{cases}
        2, & r=0,\\[1mm]
        0, & r=1.
    \end{cases}
\end{align*}
Hence
\begin{align*}
    \lambda^{\chi_{1,0,1}}=(2,2,2,0,0,0).
\end{align*}
\end{example}
\begin{example}
Consider
\begin{align*}
    \cG_{1}=\mathbb{Z}_{4}\times \mathbb{Z}_{3}\times \mathbb{Z}_{2},
    \qquad
    \cG_{2}=\mathbb{Z}_{4}\times \mathbb{Z}_{3},
\end{align*}
written additively, and define the surjective homomorphism
\begin{align*}
    \phi(a,b,c)=(a+2c,b).
\end{align*}
Its kernel is
\begin{align*}
    \ker\phi=\{(0,0,0),(2,0,1)\}.
\end{align*}
Let the characters of $\widehat{\cG_{1}}$ be
\begin{align*}
    \chi_{u,v,w}(a,b,c)=i^{ua}\omega^{vb}(-1)^{wc},
\end{align*}
where $u\in\mathbb{Z}_{4}$, $v\in\mathbb{Z}_{3}$, $w\in\mathbb{Z}_{2}$, and $\omega=e^{2\pi i/3}$. Let the characters of $\widehat{\cG_{2}}$ be
\begin{align*}
    \xi_{r,s}(x,y)=i^{rx}\omega^{sy},
\end{align*}
where $r\in\mathbb{Z}_{4}$ and $s\in\mathbb{Z}_{3}$. Then
\begin{align*}
    \hat{\phi}(\xi_{r,s})=\chi_{r,s,r \bmod 2},
\end{align*}
and hence
\begin{align*}
    \image\hat{\phi}
    =
    \{\chi_{u,v,w}\in \widehat{\cG_{1}}: w=u \bmod 2\}.
\end{align*}

Now consider the input eigen list $\blambda_{1}=\{\lambdaa_{\chi}\}_{\chi\in \widehat{\cG_{1}}}$ defined by
\begin{align*}
    \lambdaa_{\chi_{u,v,w}}=
    \begin{cases}
        2, & (u,w)=(0,0),\\
        1, & (u,w)=(1,1),\\
        3, & (u,w)=(2,0),\\
        2, & (u,w)=(3,1),\\
        0, & w\neq u \bmod 2,
    \end{cases}
\end{align*}
for all $v\in\mathbb{Z}_{3}$. Then
\begin{align*}
    \sum_{\chi\in \widehat{\cG_{1}}}\lambdaa_{\chi}
    =
    3(2+1+3+2)
    =
    24
    =
    |\cG_{1}|,
\end{align*}
so $\blambda_{1}$ is a valid eigen list. Also,
\begin{align*}
    \lambdaa_{\chi}=0,\qquad \forall \chi\notin \image(\hat{\phi}).
\end{align*}
Hence the hypothesis of Lemma~\ref{lem:surjective homomorphism special} is satisfied.

Next, let
\begin{align*}
    k=(2,0,1)\in\ker\phi.
\end{align*}
If $\chi_{u,v,w}\in \image(\hat{\phi})$, then $w=u \bmod 2$, and therefore
\begin{align*}
    \chi_{u,v,w}(k)
    =
    i^{2u}(-1)^w
    =
    (-1)^u(-1)^w
    =
    1.
\end{align*}
Thus, for every $(a,b,c)\in \cG_{1}$,
\begin{align*}
    \ket{\psia_{(a,b,c)+(2,0,1)}}
    &=
    \frac{1}{\sqrt{|\cG_{1}|}}
    \sum_{\chi\in \widehat{\cG_{1}}}
    \sqrt{\lambdaa_{\chi}}\,
    \chi\big((a,b,c)+(2,0,1)\big)\ket{\chi}\\
    &=
    \frac{1}{\sqrt{|\cG_{1}|}}
    \sum_{\chi\in \image(\hat{\phi})}
    \sqrt{\lambdaa_{\chi}}\,
    \chi(a,b,c)\chi(2,0,1)\ket{\chi}\\
    &=
    \frac{1}{\sqrt{|\cG_{1}|}}
    \sum_{\chi\in \image(\hat{\phi})}
    \sqrt{\lambdaa_{\chi}}\,
    \chi(a,b,c)\ket{\chi}\\
    &=
    \ket{\psia_{(a,b,c)}}.
\end{align*}
This shows that the states are unchanged along $\ker\phi$.

By Lemma~\ref{lem:surjective homomorphism special}, the induced channel $W_{2}$ on $\cG_{2}$ is a $\cG_{2}$-covariant PSC with eigen list $\blambda_{2}=\{\lambdab_{\xi}\}_{\xi\in \widehat{\cG_{2}}}$ satisfying
\begin{align*}
    \lambdab_{\xi_{r,s}}
    =
    \frac{|\cG_{2}|}{|\cG_{1}|}\lambdaa_{\hat{\phi}(\xi_{r,s})}
    =
    \frac{1}{2}\lambdaa_{\chi_{r,s,r \bmod 2}}.
\end{align*}
Hence
\begin{align*}
    \lambdab_{\xi_{r,s}}
    =
    \begin{cases}
        1, & r=0,\\[1mm]
        \frac{1}{2}, & r=1,\\[1mm]
        \frac{3}{2}, & r=2,\\[1mm]
        1, & r=3,
    \end{cases}
\end{align*}
independent of $s$. Ordering the characters of $\widehat{\cG_{2}}$ as
\begin{align*}
    \xi_{0,0},\xi_{0,1},\xi_{0,2},
    \xi_{1,0},\xi_{1,1},\xi_{1,2},
    \xi_{2,0},\xi_{2,1},\xi_{2,2},
    \xi_{3,0},\xi_{3,1},\xi_{3,2},
\end{align*}
the output eigen list is
\begin{align*}
    \blambda_{2}
    =
    \left(
    1,1,1,
    \frac{1}{2},\frac{1}{2},\frac{1}{2},
    \frac{3}{2},\frac{3}{2},\frac{3}{2},
    1,1,1
    \right).
\end{align*}

This example shows that, under the support condition in Lemma~\ref{lem:surjective homomorphism special}, the induced channel depends only on $\phi(a,b,c)$, even though the original label lies in $\cG_{1}$.
\end{example}

\end{document}